\documentclass[reqno,12pt]{amsart}
\usepackage[margin=1in,marginparwidth=75pt,centering]{geometry}
\usepackage{mathtools,mathrsfs,lineno,paralist,graphicx,float}
\usepackage{amssymb,amsthm,amsmath}
\usepackage[T1]{fontenc}
\usepackage[utf8]{inputenc}
\usepackage{color}
\usepackage[svgnames]{xcolor}
\usepackage{ifpdf}

\ifpdf \usepackage[pdftex,pdfstartview=FitH,pdfpagemode=none,colorlinks,bookmarks,linkcolor=blue,
citecolor=Indigo
]{hyperref} \else  \usepackage[hypertex]{hyperref} \fi

\usepackage[nobysame,abbrev,alphabetic]{amsrefs}

\usepackage{pgfplots}
\usepackage{tikz}
\usetikzlibrary{arrows}
\usepackage{comment}
\usepackage{bm}
\usepackage{ulem}
\usepackage{enumerate}
\usepackage{enumitem}

\definecolor{fadeblue}{RGB}{0,57,128}
\def\fadeblue{\color{fadeblue}}
\usepackage{etoolbox}
\patchcmd{\section}{\normalfont}{\normalfont \fadeblue}{}{}
\patchcmd{\subsection}{\normalfont}{\normalfont \fadeblue}{}{}
\patchcmd{\subsubsection}{\normalfont}{\normalfont \fadeblue}{}{}

\usepackage{calc}
\renewcommand{\leq}{\leqslant}
\renewcommand{\geq}{\geqslant}
\newcommand{\hide}[1]{}

\newcommand\set[1]{\left\{#1\right\}} 
\newcommand\pa[1]{\left(#1\right)}

\newcommand\op[1]{\operatorname{#1}}

\newcommand\Item[1][]{%
  \ifx\relax#1\relax  \item \else \item[#1] \fi
  \abovedisplayskip=0pt
  \abovedisplayshortskip=0pt~
  \vspace*{
  -8pt
  }
  }
\newcommand\IItem[1][]{%
  \ifx\relax#1\relax  \item \else \item[#1] \fi
  \abovedisplayskip=0pt
  \abovedisplayshortskip=0pt~
  \vspace*{
  -\baselineskip
  }
  }

\date{}

\author{Xianzhe Li}
\address{Chern Institute of Mathematics and LPMC, Nankai University, Tianjin 300071, China} 
\email{xianzheli@mail.nankai.edu.cn}
 \author {Jiangong You}
 \address{
 Chern Institute of Mathematics and LPMC, Nankai University, Tianjin 300071, China} 
 \email{jyou@nankai.edu.cn}

 \author{Qi Zhou}
 \address{Chern Institute of Mathematics and LPMC, Nankai University, Tianjin 300071, China}
 \email{qizhou@nankai.edu.cn}

\newtheorem{theorem}{Theorem}[section]
\newtheorem{corollary}{Corollary}[section]
\newtheorem{proposition}{Proposition}[section]
\newtheorem{lemma}{Lemma}[section]

\newtheorem{remark}{Remark}[section]
\numberwithin{equation}{section}

\newcommand{\mc}[1]{\mathcal{#1}}

\newcommand{\R}{{\mathbb R}}
\newcommand{\C}{\mathbb{C}}
\newcommand{\T}{{\mathbb T}}

\newcommand{\Z}{{\mathbb Z}}
\newcommand{\ii}{{i}}
\newcommand{\dif}{\, d}   
\newcommand{\e}{\varepsilon} 


\makeatletter
\def\saveenum{\xdef\@savedenum{\the\c@enumi\relax}}
\def\resetenum{\global\c@enumi\@savedenum}
\makeatother

\begin{document}

\date{\today}
\pagestyle{plain}
\title{Exact local distribution of the absolutely continuous spectral measure 
}
\maketitle
\begin{abstract}
It is well-established that the spectral measure for one-frequency Schr\"odinger operators with Diophantine frequencies exhibits optimal $1/2$-H\"older continuity within the absolutely continuous spectrum. This study extends these findings by precisely characterizing the local distribution of the spectral measure for dense small potentials, including a notable result for any subcritical almost Mathieu operators. Additionally, we investigate the stratified H\"older continuity of the spectral measure at subcritical energies.
\end{abstract}
\tableofcontents
\section{Introduction}

Let $\mu$ be a Borel probability measure on $\mathbb{R}$. The local distribution, or the fine property, of $\mu$ provides insights into the behavior of $\mu$ at a local level. The central concept is the local dimension of the measure:
$$  \underline{D}\mu(E)=  \liminf_{\epsilon\to 0}\frac{\ln \mu (E-\epsilon,E+\epsilon)}{\ln \epsilon}      $$
is called the lower local dimension of $\mu$,   while $$ \overline{D}\mu(E) =\limsup_{\epsilon\to 0}\frac{\ln \mu (E-\epsilon,E+\epsilon)}{\ln \epsilon}
$$ 
is called the upper local dimension of $\mu$. If $ \underline{D}\mu(E)= \overline{D}\mu(E) $ exists, one then call  
\[
    D \mu(E)=\lim_{\epsilon\to 0}\frac{\ln \mu (E-\epsilon,E+\epsilon)}{\ln \epsilon}
\] 
the local dimension, or the scaling index of $\mu$ (at $ E $).  If the local dimension of $\mu $ exists and is constant $\mu $-almost everywhere, then $\mu $ is {\it exact dimensional}.

The study of these properties enables mathematicians to comprehend the intricate structures and patterns that emerge in complex systems, offering insights into their behavior at different scales. These properties have already played a pivotal role in various mathematical domains. For example, it helps in understanding the irregularities and variations in the fractal structure, and helps in bridging the gap between geometric and measure-theoretic perspectives \cite{mattila1999geometry,fan2002,FengHu,Fal2014Fractal}, it helps in determining the long-term statistical behavior, the stability of dynamical systems \cite{YoungDim,LedrappierYoung,BPSDim,Hochman,Rapaport,ledrappier2021exact}, and ets.

\hspace{2em}

In this paper, we are interested in the local distribution of the spectral measure of the quasi-periodic Schr\"odinger operator on $\ell^2(\mathbb{Z})$:
\begin{equation} \label{Op1.1}
    (H_{ v,\alpha,\theta}u)(n)=u(n+1)+u(n-1)+ v(n\alpha+\theta)u(n),
\end{equation} 
where the potential $v \in C^\omega(\mathbb{T}^d, \mathbb{R})$ is analytic, the frequency  $\alpha \in (\mathbb{R} \backslash \mathbb{Q})^d$ is rationally independent, and the phase $\theta \in \mathbb{R}^d$. The spectrum $\Sigma_{v,\alpha}$ is a compact set of $\R$ independent of $\theta$. 

A central example is the almost Mathieu operator (AMO):
$$
(H_{\lambda,\alpha,\theta}u)(n) = u(n+1) + u(n-1) + 2\lambda \cos2\pi(n\alpha+\theta)u(n),
$$
which was initially introduced by Peierls \cite{Peierls} as a model for an electron in a 2D lattice with a homogeneous magnetic field \cite{harper, R}. AMO also plays a crucial role in the Thouless et al. theory of the integer quantum Hall effect \cite{ntw}.

Absolutely continuous spectrum occurs only rarely in one-dimensional Schr\"odinger operators \cite{remling2011absolutely}. The Kotani-Last conjecture says that within the class of ergodic Schr\"odinger operators, absolutely continuous spectrum only occurs for almost periodic potentials. This conjecture was recently disproven \cite{avila2015kotani,volberg2014kotani,you2015simple,damanik2016counterexamples}. However, in the context of quasiperiodic operators, if the potential $v$ is small, then $H_{ v,\alpha,\theta}$ always has a purely absolutely continuous spectrum \cite{eliasson,avilaACspec,aj1}.

We will be concerned with the local distribution of absolutely continuous spectral measures. By the spectral theorem, there exist Borel probability measures $\mu^{i}_{v,\alpha,\theta}$ on $\mathbb{R}$ such that
\[
    \langle \delta_i, f(H_{v,\alpha,\theta})\delta_i\rangle=\int f(E)\, d \mu^i_{v,\alpha,\theta}(E)
\] 
for all bounded measurable functions $f$. We take $\mu_{v,\alpha,\theta} = \mu^0_{v,\alpha,\theta} + \mu^1_{v,\alpha,\theta}$ as the \textit{canonical spectral measure}, since any spectral measure of \eqref{Op1.1} is absolutely continuous with respect to it. The \textit{density of states measure} (DOS) $n_{v,\alpha}$ is given by the $\theta$-average of those spectral measures $\mu_{ v,\alpha,\theta}^0$ with respect to the Lebesgue measure, that is,
\[
    \int_{\T^d} \langle \delta_0, f(H_{v,\alpha,\theta})\delta_0\rangle\, d\theta=\int f(E)\, d n_{v,\alpha}(E).
\]
The \textit{integrated density of states} (IDS) is the distribution function of $n_{v,\alpha}$, that is,
\[
\mathcal{N}_{v, \alpha}(E):=n_{v,\alpha}\pa{(-\infty,E]}.
\] 
For simplicity, we denote the spectrum, canonical spectral measure, DOS and IDS of AMO by $\Sigma_{\lambda,\alpha}$, $\mu_{\lambda,\alpha,\theta}$, $n_{\lambda,\alpha}$ and $\mc{N}_{\lambda,\alpha}(E)$, respectively.
The following is a natural question: if $\mu_{v,\alpha,\theta}$ is absolutely continuous, what is the local distribution or the fine properties of $\mu_{v,\alpha,\theta}$ (resp. the DOS $n_{v,\alpha}$)?

\subsection{Exact local distribution of ac measure: almost Mathieu case}
In 1986, when studying the fractal problems of AMO,  the physicist Tang and Kohmoto  made the following conjecture \cite{TK1986Global}:
\begin{quote}
``An absolutely continuous spectrum (extended states), it is dominated by points with `trivial' scaling index $D \nu (E)=1$\footnote{Here $ \nu $ is the density of states measure.}, and a fraction dimension $1$. It can contain 
a finite or a countably number of singularities with  $D \nu(E)\neq 1$, perhaps van Hove singularities.''
\end{quote} In this paper, we intend to answer Tang-Kohmoto's conjectrue, even more, we give the local dimension at any energy $E$.

To introduce our precise results, we first introduce some necessary concepts. Throughout the paper, we will always assume the  frequency $ \alpha $ is {\it Diophantine}. Recall that 
$\alpha$ is Diophantine if $\alpha\in\mathrm{DC}_d:=\bigcup_{\gamma>0,\tau>d}\mathrm{DC}_d(\gamma,\tau) $ 
 where
\[
    \mathrm{DC}_d(\gamma,\tau):=\left\{x\in\R^d:\inf_{j\in\Z}|\langle n,x\rangle- j|\geq\frac{\gamma}{|n|^\tau}, n\in\mathbb{Z}^d-\{\mathbf{0}\}\right\}.
\]  
We find that the local dimension of spectral measures of AMO is actually determined by the resonances of IDS. For any $\varphi \in [0,1]$, we define
\begin{equation}
\delta(\alpha,\varphi)=\limsup\limits_{k\rightarrow\infty}-\frac{\ln\| \varphi -\langle k,\alpha\rangle\|_{\R/\Z}}{|k|}.
\end{equation}
which characterize the strength of resonances of $\varphi$ with respect to the frequency $\alpha$. 
By the famous ``Gap Labelling Theorem'' \cite{Joh1986Exponential}, if $E\notin\Sigma_{\lambda,\alpha}$, then  $ \mathcal{N}_{\lambda,\alpha}(E)= k\alpha \mod \Z $, which we  will formally set  $ \delta(\alpha,\mathcal{N}_{\lambda,\alpha}(E))=\infty$. 
The following result completely addresses Tang and Kohmoto's question.

\begin{theorem}\label{MainCor} 
    Let $ 0< |\lambda|<1 $, $ \alpha\in \mathrm{DC}_1 $. 
Then for  any $ E\in\Sigma_{\lambda,\alpha} $, we have the following arithmetic description:
  \begin{enumerate}
  \item If $ \mathcal{N}_{\lambda,\alpha}(E)\neq k\alpha\mod \Z $. Let $\delta(\alpha,\mathcal{N}_{\lambda,\alpha}(E))=\delta$. Then   for all $ \theta\in\R $, we have 
      \[
    \underline{D}  \mu_{\lambda,\alpha,\theta}(E)=\frac{\delta}{\max\{2\delta + \ln \lambda  ,\delta\}}\in [1/2,1],\ \overline{D} \mu_{\lambda,\alpha,\theta}(E) =1.
      \] 
      \item If $ \mathcal{N}_{\lambda,\alpha}(E)= k\alpha\mod \Z $. Then   for all $ \theta\in\R $, we have 
       \[
    \underline{D} \mu_{\lambda,\alpha,\theta}(E)=  \overline{D} \mu_{\lambda,\alpha,\theta}(E) =1/2.
      \] 
      
      \end{enumerate}
    
\end{theorem}

\begin{remark}\leavevmode
\begin{enumerate}
\item Theorem \ref{MainCor} implies that   $\overline{D} \mu_{\lambda,\alpha,\theta}(E) =1$ except for a    countably number of singularties with  $\overline{D} \mu_{\lambda,\alpha,\theta}(E) =1/2.$
\item As the result holds for any phase  $\theta$,  it immediately implies that all results in  Theorem \ref{MainCor} holds for the DOS measure $ n_{\lambda,\alpha} $.

\end{enumerate}
\end{remark}

Indeed, we can not only calculate the local dimension of the spectral measure at any energy, but also establish the exact local distribution of the spectral measure, which exhibits fractal properties governed by the resonances of the IDS.  To explain this,  let $ \alpha\in\R^d $, $ \theta\in\R $, $ \epsilon_0>0 $, we say that $k$ is a $\epsilon_0$-resonance of $\theta$, if 
\begin{equation}\label{upper+}
\|\theta -\langle k,\alpha \rangle\|_{\R/\Z}\leq e^{-\epsilon_0 |k| }.
\end{equation} 
Choose any $0<\epsilon_0<-\ln\lambda $. We consider now the set of $ \epsilon_0$-resonances  of  $\mathcal{N}_{\lambda,\alpha}(E) $. Take and order all the $ \epsilon_0$-resonances by $ 0=\ell_0<|\ell_1|\leq |\ell_2|\leq \cdots $. 
For any $ \mathcal{N}_{\lambda,\alpha}(E)\in[0,1] $, and any $ \ell_s $, denoting $n=|\ell_s|$, we define $ \eta_n $ by
\begin{equation}
    \eta_n=-\frac{\ln \|\mathcal{N}_{\lambda,\alpha}(E)-\ell_s\alpha\|_\T}{|\ell_s|}\in (0,\infty].
\end{equation} 
One can obtain the exact local distribution of the absolutely spectral measure as follows:

\begin{theorem}\label{MainTHM}
  Let $ 0< \lambda<1 $, $ \alpha\in \mathrm{DC}_1 $,
  $ E\in\Sigma_{\lambda,\alpha
    } $. There exists  $C(\lambda,\alpha,E)>0$, and 
for any $ \varepsilon>0 $, there exists $ \epsilon_*
    =\epsilon_*(\lambda,\alpha,E,\varepsilon)>0 
    $ such that for any $\epsilon\leq \epsilon_*$, and for any  $ \theta\in\R $, we have  the following:
    \begin{enumerate}
        \item \label{Thm2:case1}
        If $\delta(\alpha,\mathcal{N}_{\lambda,\alpha}(E))=\delta\geq -\ln \lambda$, and $ \mathcal{N}_{\lambda,\alpha}(E)\neq k\alpha\mod \Z $. Then 
        \[
            \epsilon^{f(\epsilon)+\varepsilon }\leq\mu_{\lambda,\alpha,\theta}(E-\epsilon,E+\epsilon)\leq \epsilon^{f(\epsilon)-\varepsilon},
        \]  where $ f(\epsilon):(0,\epsilon_*]\to (1/2,1] $ has the following expression:
        \begin{enumerate}[label=$(\mathit{\roman*})$\ ,align=left,widest=ii,labelsep=0pt,leftmargin=2em]
            \item If $\eta_n \leq -\ln \lambda $, then $ f(\epsilon)=1 $ for $ \epsilon\in [e^{ |\ell_{s+1}| \ln \lambda},e^{ |\ell_{s}| \ln \lambda }] $.
            \item If  $\eta_n\geq -\ln \lambda $, then 
            \begin{equation}\label{eq:hat-f}
            f(\epsilon)=\left\lbrace\begin{aligned}
                    &\frac{1}{2}- \frac{\ln \lambda|\ell_s|}{2\ln\epsilon^{-1}} &  &\text{for } \epsilon \in [e^{-(2\eta_{n}+ \ln \lambda)|\ell_s|}, e^{   |\ell_s| \ln \lambda } ],\\
                    &\frac{1}{1-b_s}+\frac{b_s}{1-b_s}\frac{\ln \lambda|\ell_{s+1}|}{\ln\epsilon^{-1}}& &\text{for } \epsilon \in[ e^{ |\ell_{s+1}| \ln \lambda }, e^{-(2\eta_{n}+\ln \lambda)|\ell_s|}].
                \end{aligned}\right.
            \end{equation} 
            where $ b_s=\frac{(\eta_n+\ln \lambda)|\ell_s|}{-\ln \lambda|\ell_{s+1}|-\eta_n |\ell_s|} $.
        \end{enumerate}

        \item If $\delta(\alpha,\mathcal{N}_{\lambda,\alpha}(E))=\delta< -\ln \lambda$. Then 
        \[
            C^{-1}\epsilon^{1+\varepsilon }\leq\mu_{\lambda,\alpha,\theta}(E-\epsilon,E+\epsilon)\leq C\epsilon.
        \]
        \item If $ \mathcal{N}_{\lambda,\alpha}(E)= k\alpha\mod\Z $ for some $ k\in\Z $.
         Then 
        \[
            C^{-1}\epsilon^{\frac{1}{2}+\varepsilon }\leq\mu_{\lambda,\alpha,\theta}(E-\epsilon,E+\epsilon)\leq C\epsilon^{\frac{1}{2}}.
        \]
    \end{enumerate}
 \end{theorem}

Let us explain more about this exact local distribution, as illustrated in Figure \ref{Fig2}. The asymptotic function $f(\epsilon)$  has a universal structure which does not depend on the choice of $\epsilon_0<-\ln \lambda $. It depends only on the resonances $n_s$ and $\eta_n$,  and displays a fractal structure  around all  these ``resonance windows'' $ [e^{ |\ell_{s+1}| \ln \lambda},e^{ |\ell_{s}| \ln \lambda }] $. Although the proof relies on the almost reducibility procedure, which is of a local nature, $f(\epsilon)$  is a globally defined function. Furthermore, those effective $\epsilon_0$-resonances  are actually from  $$\|\mathcal{N}_{\lambda,\alpha}(E)-\ell_s\alpha\|_\T\leq e^{\ln \lambda|\ell_s|} .$$ 
Indeed, if the resonance is weak (i.e. $\eta_n \leq -\ln \lambda  $), then $f(\epsilon)$ is always equal to $1$, which implies that   $\mu_{\lambda,\alpha,\theta}$ is almost Lipschitz, if the resonance is strong (i.e. $\eta_n \geq -\ln \lambda  $), then $f(\epsilon)$ has some fluctuation, but at least   $\mu_{\lambda,\alpha,\theta}$  has exact H\"older exponent  $\frac{\eta_{n}}{2\eta_{n}+ \ln \lambda}$ in this ``resonance windows'' . We should remark that \eqref{eq:hat-f} was also motivated by recent 
result of Jitomirskaya-Liu \cite{JLiu,jLiu1}, 
who revealed a universal hierarchical structure in localized eigenfunctions for supercritical AMO. Specifically, they not only determine the optimal asymptotics of the localized  eigenfunction but also uncover the hierarchical structure of local maxima. This global/local structure is characterized by a universal function $f(k)$
, which is defined by the resonances of the phase. Meanwhile, the exact 
asymptotics of the extended  eigenfunction (generalized eigenfunction in the absolutely continuous spectrum regime) was explored quite recently \cite{Universal}.

\begin{center}
    \begin{tikzpicture}[scale=0.75]
    \draw [->](3,1.5)--(23.6,1.5);
    \draw [->](3,1.5)--(3,6);
    \draw [thick] plot [smooth] coordinates {(4,5)(4.5,2.5)(5.5,4.7)(7,4.93)(8,4.95)(9,5)};
    \draw [thick] plot [smooth] coordinates {(11,5)(12,3.5)(15,4.8)(17,4.97)(22,5)};
    \draw [dashed] (4,5)--(4,1.5);
    \draw [dashed] (4.5,5)--(4.5,1.5);
    \draw (4.5,2.5)--(6,2.5);
    \node [right] at (6,2.5){$\frac{\eta_{n}}{2\eta_{n}+ \ln \lambda}$};
    \draw [dashed] (3,5)--(4,5);
    \draw [dashed] (4,5)--(9,5);
    \draw [dashed] (9,5)--(11,5);
    \node  at (10,3.25){ $ \cdots\cdots $};
    \draw [dashed] (11,5)--(22.5,5);
    \draw [dashed] (9,5)--(9,1.5);
    
    \draw [dashed] (11,5)--(11,1.5);
    \draw [dashed] (12,5)--(12,1.5);
    \draw (12,3.5)--(16,3.5);
    \node [right] at (16,3.5){$\frac{\eta_{n'}}{2\eta_{n'}+ \ln \lambda}$};
    \draw [dashed] (22,5)--(22,1.5);
    \node [above] at (4,0.3){\tiny ${-\ln \lambda |\ell_{s}|}$};
    \node [below] at (5.5,1.5){\tiny${(2\eta_{n}+ \ln \lambda) |\ell_{s}|}$};
    \node [above] at (9,0.3){\tiny${-\ln \lambda |n_{s+1}|}$};
    \node [above] at (11,0.3){\tiny${-\ln \lambda  |n_{s'}|}$};
    \node [below] at (13,1.5){\tiny${(2\eta_{n'}+\ln \lambda) |\ell_{s'}|}$};
    \node [above] at (22,0.3){\tiny${-\ln \lambda |\ell_{s'+1}|}$};
    \node [above] at (23.5,1.6){$\ln \epsilon^{-1}$};
    \node [left] at (3,5){$1$};

    \node [left] at (3,1.5){$\frac{1}{2}$};
    \node [right] at (3.1,6){$f(\epsilon)$};
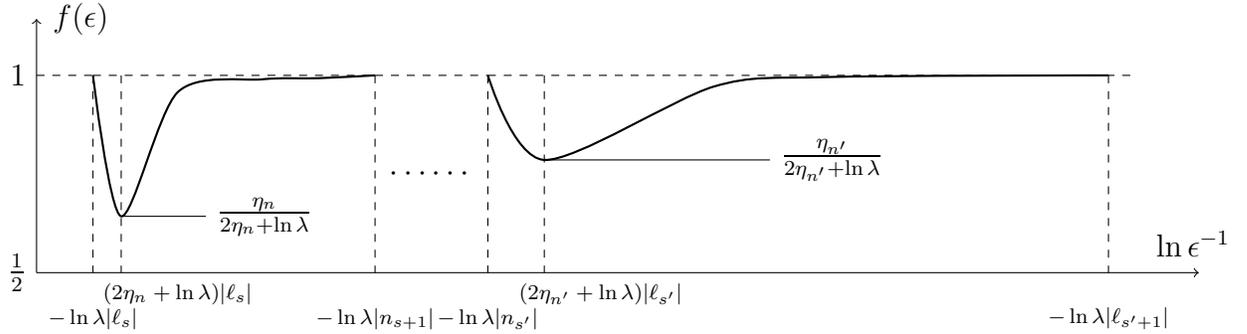
    \node [below] at (10,0){Figure 1.1. Mode of $ f(\epsilon) $ in two resonant windows, $n=|\ell_s|$, $n'=|\ell_{s'}|$};\label{Fig2}
    \end{tikzpicture}
    \end{center}

To the best of our knowledge, Theorem \ref{MainCor} and Theorem \ref{MainTHM} are the first results on the local dimensional property of individual absolutely continuous spectral measures of ergodic operators at {\it any} energy.
Prior to our work, \cite{AvilaHolderContinuityAbsolutely2011} established the $1/2$-H\"older continuity of the absolutely continuous spectral measure. Theorem \ref{MainTHM} (3) further shows that their result is optimal. One can also consider \cite{DG2013Holder,DGY2016Fibonacci,Qu2017ExactDimensional,cao2023almost}  for  the  dimensional properties of the DOS measure in the Sturm Hamiltonian case, where the spectral measure is singular.

\subsection{Stratified H\"older continuity}

Theorem \ref{MainCor}  establish the upper (resp. lower) local dimension of  the spectral measure for AMO, in other language,  Theorem \ref{MainCor}  shows the exact local H\"older exponent of the absolutely continuous spectral measure of AMO  at any energy. For general potentials, let us recall the Spectral Dichotomy Conjecture proved by Avila \cite{avila2015global}: \\

\textit{Spectral Dichotomy Conjecture.} For typical $v$, $\alpha$, $\theta$,  $H_{v,\alpha,\theta}$ is the direct sum of operators
$H_+$ and $H_-$ with disjoint spectra such that $H_+$ is ``localized''  and $H_-$ is absolutely continuous.  \\

Actually one has dynamical information on $H_-$. Given $A \in C^\omega(\mathbb{T}^d, \op{SL}(2,\mathbb{R}))$ and rationally independent $\alpha \in \mathbb{R}^d$, recall a quasiperiodic cocycle $(\alpha, A)$ is defined as:
$$
(\alpha,A)\colon \left\{
\begin{array}{rcl}
	\T^d \times \R^2 &\to& \T^d \times \R^2\\[1mm]
	(x,v) &\mapsto& (x+\alpha,A(x)\cdot v)
\end{array}
\right.  
$$
An one-frequency cocycle $(\alpha, A)$ is called subcritical, if there is a uniform subexponential bound on the growth of the $n $-th iterate $\|\mathcal{A}_n(\cdot)\|$ through some band $\{z:| \Im z| <h\} $, here $h$ is called the {\it subcritical radius}.  
Note \eqref{Op1.1} naturally induces a Schr\"odinger cocycle $(\alpha, S_{E}^{v})$, where \begin{eqnarray*}
S_{E}^{v}(\theta)=\left( \begin{array}{ccc}
 E-v(\theta) &  -1\cr
  1 & 0\end{array} \right)\in \op{SL}(2,\mathbb{R}),
\end{eqnarray*}
then $H_-$ is just $H$ restricted to $\Sigma^{-}_{v,\alpha} $:  the energies $E\in \Sigma_{v,\alpha} $ such that  $(\alpha, S_{E}^{v})$ is subcritical \cite{avila2015global,Avi2023KAM}. \\

Now the natural question is whether Theorem \ref{MainCor} is true for any absolutely continuous spectral measure.  
 Of course, the answer is no because of the possible collapes of gaps (one can also consult Proposition \ref{Prop:finite-resonnace}). Nevertheless, our method allows us to obtain the 
stratified H\"older continuity at subcritical energies.

\begin{theorem}\label{Thm:generalholder}
    Let $ \alpha\in \mathrm{DC}_1 $, $ v\in C^\omega(\T,\R)$, $E\in \Sigma^{-}_{v,\alpha}$ with subcritical radius $h$.  Then for $\epsilon$ sufficiently small, and for all $\theta$,  we have the following: 
    \begin{enumerate}
    \item If $2\pi h\leq \delta(\alpha,\mathcal{N}_{ v, \alpha}(E))=\delta<\infty$, then we have  
    \[
                    \mu_{v,\alpha,\theta}(E- \epsilon, E+ \epsilon)\leq \epsilon^{\frac{\delta}{2\delta-2\pi h}-o(1)},      \]
        \item If $\delta(\alpha,\mathcal{N}_{ v, \alpha}(E))< 2\pi h$, then  we have 
        \[
   \mu_{v,\alpha,\theta}(E-\epsilon,E+\epsilon)\leq C\epsilon. 
        \]
        \item If $\delta(\alpha,\mathcal{N}_{ v, \alpha}(E))=\infty$, then  we have
        \[
            \mu_{v,\alpha,\theta}(E-\epsilon,E+\epsilon)\leq C\epsilon^{\frac{1}{2}}.
        \]
    \end{enumerate}
     
\end{theorem}

  The stratified H\"older continuity of the absolutely continuous spectral measure holds for any phase $\theta$, thus the corresponding result also holds for the integrated density of states (IDS):

\begin{corollary}\label{cor:generalholder}
    Let $ \alpha\in \mathrm{DC}_1 $, $ v\in C^\omega(\T,\R)$, $E\in \Sigma^{-}_{v,\alpha}$ with subcritical radius $h$.  Then for $\epsilon$ sufficiently small,  we have the following:     \begin{enumerate}
    \item If $2\pi h\leq \delta(\alpha,\mathcal{N}_{ v, \alpha}(E))=\delta<\infty$, we have  
    \[
                   \mathcal{N}_{ v, \alpha}(E+ \epsilon)- \mathcal{N}_{ v, \alpha}(E- \epsilon)\leq \epsilon^{\frac{\delta}{2\delta-2\pi h}-o(1)},      \]
        \item If $\delta(\alpha,\mathcal{N}_{ v, \alpha}(E))< 2\pi h$, we have 
        \[
 \mathcal{N}_{ v, \alpha}(E+ \epsilon)- \mathcal{N}_{ v, \alpha}(E- \epsilon)\leq C\epsilon. 
        \]
        \item If $\delta(\alpha,\mathcal{N}_{ v, \alpha}(E))=\infty$,  we have
        \[
        \mathcal{N}_{ v, \alpha}(E+ \epsilon)- \mathcal{N}_{ v, \alpha}(E- \epsilon)\leq  C\epsilon^{\frac{1}{2}}.
        \]
    \end{enumerate}
     Moreover, for the almost Mathieu operator $v(\cdot)=2 \lambda \cos (\cdot)$, the result holds for any $\lambda\neq 1$. 
\end{corollary}

 Let us revisit previous results on H\"older continuity of IDS.   In the regime of zero Lyapunov exponent,  sharp bounds ($1/2$-H\"older continuity) for the integrated density of states for Diophantine frequencies were derived through analyses by \cite{eliasson} and \cite{aj1}.
 These results were obtained in both perturbative and non-perturbative regimes, as discussed in \cite{amor} and \cite{aj1}, respectively.  However, Corollary \ref{cor:generalholder} indicates that this sharp $1/2 $-H\"older continuity can only occur in the case $\delta=\infty$.
 In the regime of positive Lyapunov exponents, 
 Goldstein and Schlag \cite{goldstein2008fine} demonstrated H\"older continuity of the integrated density of states for a full Lebesgue measure subset of Diophantine frequencies.  Their results were nearly optimal for the supercritical almost Mathieu operator, achieving $(1/2 - o(1))$-H\"older continuity for $|\lambda| > 1$. 
 Before this, Bourgain \cite{bourgain2000holder} had established almost $1/2$-H\"older continuity for almost Mathieu operators in the perturbative regime, for Diophantine $\alpha$, and sufficiently large $\lambda$ (depending on $\alpha$).
If the potential is not analytic, one can consult \cite{CCYZ,Klein,WangZhangRegu} and the references therein for more results on the regularity  of  IDS.

\subsection{Exact local distribution of ac measure: general analytic potential}

For general analytic potentials, Theorem \ref{Thm:generalholder} shows that although it may not be  possible to obtain the local dimension of the spectral measure as precisely as in the case of AMO (Theorem \ref{MainCor}), one can achieve stratified H\"older continuity of the absolutely continuous spectral measure. Note that Theorem \ref{MainCor} is a direct corollary of the exact local distribution (Theorem \ref{MainTHM}). This raises the natural question of whether this kind of exact local distribution is specific to AMO. In the following, we show that it is indeed possible to obtain the exact local distribution for dense small analytic potentials:

\begin{theorem}\label{THM1.2}
    Let $ \alpha\in \mathrm{DC}_d $, $ v\in C_h^\omega(\T^d,\R)$ with 
    \begin{equation}\label{re-n}2\pi h <  \delta(\alpha,\mathcal{N}_{ v, \alpha}(E))=\delta<\infty\end{equation}
   For any $ h'\in(0,h) $,  there exists $c_*=c_*(\alpha,h,h')>0 $ such that if $
         \|v\|_h\leq c_*  $, then  for any  $ \hat{\epsilon}>0 $,  there exist a sequence $(k_{s})$, $ \tilde{v}\in C^\omega_{h'}(\T,\R) $ with $ \|\tilde{v}-v\|_{h'}<\hat{\epsilon} $ and 
         $$\delta(\alpha,\mathcal{N}_{\tilde{v}, \alpha}(E))=\delta$$ 
         such that for any $\theta\in\R$, we have         
     $$
     \underline{D}\mu_{\tilde{v},\alpha,\theta}(E)=\frac{\delta}{2\delta-2\pi h},\  \overline{D}\mu_{\tilde{v},\alpha,\theta}(E)=1.
     $$ 
    More precisely, for any $ \varepsilon>0 $, there exists $ \epsilon_*=\epsilon_*(\varepsilon)>0 $, such that
     if $\epsilon\leq\epsilon_*$, 
    \[
        \epsilon^{f(\epsilon)+\varepsilon}\leq\mu_{\tilde{v},\alpha,\theta}(E-\epsilon,E+\epsilon)\leq \epsilon^{f(\epsilon)-\varepsilon},
    \]  
    where $f(\epsilon)$ has the following form:
    \begin{equation} \label{eq:AKf(epsilon)}
            f(\epsilon)=\left\lbrace\begin{aligned}
                    &\frac{1}{2}+\frac{2\pi h|k_{s}|}{2\ln\epsilon^{-1}} &  &\text{for } \epsilon\in[ e^{-(2\delta-2\pi h)|k_{s}|},e^{-2\pi h|k_{s}|}],\\
                    &\frac{1}{1-b_s}-\frac{b_s}{1-b_s} 
                    \frac{2\pi h |k_{s+1}|}{\ln\epsilon^{-1}} & &\text{for } \epsilon\in[e^{-2\pi h|k_{s+1}|},e^{-(2\delta-2\pi h)|k_{s}|}],
                \end{aligned}\right.
     \end{equation}
    and $ b_s=\frac{(\delta-2\pi h)|k_{s}|}{2\pi h|k_{s+1}|-\delta |k_{s}|} $.
\end{theorem}

\subsection{Optimality of reducibility condition}

Let us discuss why the condition \eqref{re-n} is essential. As demonstrated in Theorem \ref{MainTHM} and Theorem \ref{Thm:generalholder}, if $2\pi h >  \delta(\alpha,\mathcal{N}_{ v, \alpha}(E))$, the spectral measure always exhibits the trivial bound. The main reason for this is the reducibility theory of quasiperiodic cocycles.
Recall that $(\alpha,A_1)$ is conjugated to $(\alpha,A_2)$, if there exists $B\in C^\omega(\T^d,\op{PSL}(2,\R))$
such that 
$$B^{-1}(\theta+\alpha)A_1(\theta)B(\theta)=A_2(\cdot).$$ 
Then $(\alpha,A)$ is reducible if it is analytic conjugate to constant. 
In the context of quasiperiodic cocycles, the corresponding concept of IDS is the fibered rotation number $\rho(\alpha,A)$\footnote{Consult Section \ref{sec2.2} for details.}. It is noteworthy that reducibility theory has recently proven to be a powerful tool for investigating the spectral theory of quasiperiodic Schrödinger operators \cite{ds, eliasson,aj1,AvilaHolderContinuityAbsolutely2011,avila2016dry,CCYZ,LYZZ}.

\begin{theorem}\label{reducibility-main}\cite{Universal}
Assume $\alpha\in \op{DC}_d$, $A \in C_h^\omega(\mathbb{T}^d,\op{SL}(2,\mathbb{R}))$ with $2\pi h>2\pi h'>\delta(\alpha,2\rho(\alpha,A))$, $R\in \op{SL}(2,\R)$. Then there exists $c_*=c_*(\alpha,R,h,h')>0$ such that if $\|A(\cdot)-R\|_h<c_*$, then $(\alpha,A)$ is reducible. 
\end{theorem} 

One may understand the optimality in several different ways. Note if one Schr\"odinger cocycle is reducible (to elliptic), then the corresponding spectral measure is Lipschitz (Proposition \ref{Prop:finite-resonnace}), thus Theorem \ref{MainTHM}  shows that Theorem \ref{reducibility-main} is optimal in the setting of  almost Mathieu cocycles, while Theorem \ref{THM1.2}   
shows that Theorem \ref{reducibility-main} is optimal in the setting of Schr\"odinger cocycles:  if $2\pi h <  \delta(\alpha,\mathcal{N}_{ v, \alpha}(E))$, one can always perturb $ v $ to $  \tilde{v} $, such that the corresponding Schr\"odinger cocycle is not reducible.   Indeed, this is true in the setting of general analytic cocycles, actually  the proof of  Theorem \ref{THM1.2} enables us to prove the following:

\begin{theorem}\label{thm:irreducible}
Assume $\alpha\in \op{DC}_d$, $R\in \op{SL}(2,\R)$, $A \in C_h^\omega(\mathbb{T}^d, \op{SL}(2,\mathbb{R}))$ with 
$$
0<2\pi h'< 2\pi h<  \delta(\alpha,2\rho(\alpha,A))=\delta<\infty.
$$ 
There exists $c_*=c_*(\alpha, R,h,h')>0$ such that for  any  $\|A(\cdot)-R\|_h<c_* $, and any $ \hat{\epsilon}>0 $, there exists $A \in C_{h'}^\omega(\mathbb{T}^d, \op{SL}(2,\mathbb{R}))$ with $\|A-A'\|_{h'} \leq \hat{\epsilon}$, and
$$ \delta(\alpha,2\rho(\alpha,A'))=\delta,$$
such that $(\alpha,A')$ is not reducible. \end{theorem}

\subsection{Ideas of the proof, novelty}
We will establish a general criterion (Proposition \ref{THM5.1}) to obtain the exact local distribution of the spectral measure, extending Avila-Jitomirskaya's general criterion on $1/2$-Hölder continuity of the absolutely continuous spectral measure \cite{AvilaHolderContinuityAbsolutely2011}, based on: the dynamical reformulation of the power-law subordinacy techniques \cite{DamanikScalingEstimatesSolutions2005, KillipDynamicalUpperBounds2003,JL1999Powerlaw,JL2000Powerlaw}, and insights from the recently developed structured quantitative almost reducibility \cite{Universal}.
Recall that $(\alpha, A)$ is almost reducible in the band $\{\theta: |\Im \theta|<h\}$,  if there exist $B_j \in C_{h}^\omega(\T^d, \op{PSL}(2, \R))$, $A_j \in \op{SL}(2, \R)$, and $f_j \in C_{h}^\omega(\T^d, \op{sl}(2, \R))$ such that the following holds:
$$
B_j^{-1}(\theta + \alpha) A(\theta) B_j(\theta) = A_j e^{f_j(\theta)} = M^{-1} \exp\begin{pmatrix} i\rho_j & \nu_j \\ \bar{\nu}_j & -i\rho_j \end{pmatrix} M e^{f_j(\theta)}.
$$
Structured quantitative almost reducibility involves tracking how the conjugacy $B_j(\cdot)$, the constant $A_j$, and the perturbation $f_j(\cdot)$ depend on the resonances $n_s$ of the fibered rotation number. This structured and quantitative estimate is crucial for establishing the exact local distribution of the spectral measure and serves as the starting point of our proof. 
Within this dynamical input, to prove the exact local distribution, the main difficulty is to prove the lower bound of  $ \mu_{v,\alpha,\theta}(E- \epsilon, E+ \epsilon)$ (the method of the upper bound are mainly developed in \cite{AvilaHolderContinuityAbsolutely2011}).   
It is well-known the spectrum measure is related with Weyl-Titchmarsh $ m $-function as follows: 
  \begin{equation}\label{m-mu}
        \epsilon\Im M(E+\ii\epsilon)=\int\frac{\epsilon^2}{(E'-E)^2+\epsilon^{2}}\dif \mu(E').
\end{equation}
To obtain the lower bound, there are two  two novel aspects:
\begin{itemize}
    \item {\it Symmetry argument} to obtain the lower bound of $ \epsilon\Im M(E+\ii\epsilon)$ (Lemma  \ref{P_kX}).
    \item {\it Dislocation argument} to obtain lower bound of $ \mu_{v,\alpha,\theta}(E- \epsilon, E+ \epsilon)$  (Lemma \ref{mu2}). To obtain the lower bound,  one key is the asymptotic upper bound of  $ \mu_{v,\alpha,\theta}(E- \epsilon, E+ \epsilon)$, another key is the   {\it dislocation} of $\epsilon$:  in order to estimate $ \mu_{v,\alpha,\theta}(E- \epsilon, E+ \epsilon)$, we actually need to 
 estimate $\epsilon^{1+\tau}\Im M(E + \ii \epsilon^{1+\tau})$ in \eqref{m-mu}. 
    
\end{itemize}
With this criterion in hand:
\begin{itemize}
    \item For AMO, one can apply the recent developments in structured quantitative almost reducibility results for AMO \cite{Universal} to achieve the desired results.
    \item For general potentials, the fibered Anosov-Katok construction, first appearing in \cite{KXZ2020AnosovKatok}, can be further developed. The advantage of our approach is that it works for all Liouvillean rotation numbers.
\end{itemize}

Finally, it is necessary to remark that the arithmetic assumption on the frequency is quite crucial. While for weakly Liouvillean frequencies, IDS remains H\"older continuous \cite{LY2015Holder}, suggesting that similar H\"older continuity of the spectral measure can be expected for such frequencies with a small rate of exponential approximation \cite{AvilaHolderContinuityAbsolutely2011}. However,  for general potential $v$ and generic $\alpha$, the integrated density of states $\mathcal{N}_{ v, \alpha}$ is not H\"older. This is due to the discontinuity of the Lyapunov exponent at rational $\alpha$, implying non-H\"older continuity for generic $\alpha$. Recently, Avila-Last-Shamis-Zhou \cite{ALSZ2024abominable} even demonstrated that the Log-H\"older continuity of IDS, as obtained from the Thouless formula, is optimal for extremely Liouvillean frequencies. The determination of the exact arithmetic constraint on $\alpha$ for which Theorem \ref{MainTHM} holds remains an open question.

\section{Preliminaries}
For an $ n\times n $ complex matrix $A$, we will use the spectral norm $ \|A\|:=\sqrt[]{\lambda_{\max}(A^*A) } $, where $ A^* $ is the conjugate transpose of $ A $, and $ \lambda_{\max}(B^*B) $ is the maximum eigenvalue of $ B^*B $. 
For a bounded analytic (possibly matrix-valued) function $ F $ defined on $ \{\theta:|\Im \theta|<h\} $, let $ \|F\|_h=\sup_{|\Im \theta|<h}\|F(\theta)\| $, and denote by $ C_h^\omega(\T^d,*) $ the set of all these $ * $-valued functions ($ * $ will usually denote $ \R $, $ \mathrm{SL}(2,\R) $, $ \mathrm{SU}(1,1) $). Also we denote $ C^\omega(\T^d,*)=\bigcup_{h>0}C^\omega_h(\T^d,*) $. We use  
$\|\cdot\|_\T:=\inf_{j\in\Z}|\cdot- j| $.
We will usually use $ c $ or $ C $ to denote small or big numerical constants (it can be different at each use).
\subsection{Lie algebra isomorphism}\label{Lie-isomorphism}
It is known that the Caley element 
\[
M=\frac{1}{\sqrt{2\ii }}\begin{pmatrix}
    1&- \ii \\1& \ii 
    \end{pmatrix}
\]
induces an isomorphism from $ \mathrm{sl}(2,\R) $ to $ \mathrm{su}(1,1)$, which is given by $B\mapsto MBM^{-1}$.  
    More precisely, a direct computation leads to
    \begin{equation}\label{isom}
    M\begin{pmatrix}
        x&y+z\\y-z&-x
        \end{pmatrix}M^{-1}=\begin{pmatrix}
        \ii z&x-\ii y\\x+\ii y&-\ii z
        \end{pmatrix}.
    \end{equation}
Then the following quantitative  normal form result is very useful:
    
     \begin{lemma}[\cite{KXZ2020AnosovKatok}]\label{normalellip}
    Let the matrix
    \[
        A=\begin{pmatrix}
            \ii t&{z}\\
            \bar{z}&-\ii t
        \end{pmatrix} \in\mathrm{su}(1,1),
    \] 
    with $ t>|z|>0$. Then, calling $ \rho=\sqrt{t^2-|z|^2} $, there exists $ D\in\mathrm{SU}(1,1) $ such that 
    \[
        D^{-1}AD=\begin{pmatrix}
            \ii\rho&\\&-\ii\rho
        \end{pmatrix},
    \] 
    where $ D $ is of the form
    \[
        (\cos 2\varphi)^{-\frac{1}{2}}\begin{pmatrix}
            \cos\varphi&e^{\ii 2\phi}\sin\varphi\\
            e^{-\ii 2\phi}\sin\varphi&\cos\varphi
        \end{pmatrix}
    \]
    with 
    \[
         \|D\|^2=\frac{t+|z|}{\rho}.
    \] 
    More precisely, $ 2\phi=\arg z-\frac{\pi}{2} $ and $ \varphi\in(-\frac{\pi}{4},\frac{\pi}{4}) $ satisfies $ \tan 2\varphi=-\frac{|z|}{\rho} $. 
\end{lemma}
A further  direct computation deduces the following result:
\begin{lemma}\label{lem:quanti-Schurlem}
    Let the matrix
    \[
        A=\begin{pmatrix}
            \ii t&{z}\\
            \bar{z}&-\ii t
        \end{pmatrix} \in\mathrm{su}(1,1),
    \] 
    with $ t\geq |z|\geq 0$ and $ \rho=\sqrt{t^2-|z|^2}\in(0,\frac{\pi}{2}] $.
  Then there exists a unitary $ U\in\mathrm{SL}(2,\C) $ such that 
    \[
        U^{-1}\exp A\ U=\begin{pmatrix}
            e^{\ii\rho}&\nu\\ &e^{-\ii\rho}
        \end{pmatrix}.
    \] 
    with $|z|\leq |\nu|\leq 2|z|$.
\end{lemma}

    \subsection{Quasi-periodic dynamics}\label{sec2.2}
    Given $ A\in C^0(\T^d,\mathrm{SL}(2,\C)) $, a quasi-periodic  \textit{cocycle} $ (\alpha,A) $ 
 is defined as 
\[
    (\alpha,A):\left\lbrace\ \begin{aligned}
        &\T^d\times\R^2 & &\to & &\T^d\times \R^2\\
        &(\theta,v) & &\mapsto & &(\theta+\alpha,A(\theta)v)
    \end{aligned}\right. ,
\] 
where $ \alpha\in\R^d $, and $ (1,\alpha) $ is assumed rationally independent.  Denote the $n$-th  iterate of $ (\alpha,A) $ by $ (n\alpha,\mathcal{A}_n) $, where
\[
    \mathcal{A}_n(\theta):=\left\{\begin{aligned}
        &A(\theta+(n-1)\alpha)\cdots A(\theta+\alpha)A(\theta), & &n\geq 0\\
        &A^{-1}(\theta+n\alpha)A^{-1}(\theta+(n+1)\alpha)\cdots A^{-1}(\theta-\alpha), & &n<0
    \end{aligned}\right. .
\] 
Then the \textit{Lyapunov exponent} of $ (\alpha,A) $ is defined by 
\[
    L(\alpha,A):=\lim_{n\to\infty}\frac{1}{n}\int_{\T^d}\ln\lVert\mathcal{A}_n(\theta)\rVert d \theta .
\]

    The cocycle $ (\alpha,A) $ is called \textit{uniformly hyperbolic} if for any $ \theta\in\T^d $, there exists a continuous splitting $ \C^2=E^s(\theta)\oplus E^u(\theta) $ such that for every $ n\geq 0 $,
\[
    \begin{aligned}
        \lVert \mathcal{A}_n(\theta)v\rVert&\leq Ce^{-cn}\lVert {v}\rVert, & &v\in E^s(\theta),\\
        \lVert \mathcal{A}_{-n}(\theta)v\rVert&\leq Ce^{-cn}\lVert {v}\rVert, & &v\in E^u({\theta}),
    \end{aligned}
\] 
for some constants $ C,c>0 $, and the splitting is invariant by the dynamics:
\[
    {A(\theta)E^s(\theta)=E^s(\theta+\alpha),\ A(\theta)E^u(\theta)=E^u(\theta+\alpha),\ \forall \theta\in\T^d.}
\]

Assume $ A\in C^0(\T^d,\mathrm{SL}(2,\R)) $ is homotopic to the identity, then there exist $ \psi:\T^d\times\T\to\R $ and $ u:\T^d\times\T\to\R_+ $ such that 
\[
    A(x)\cdot\begin{pmatrix}
        \cos 2\pi y\\\sin 2\pi y
    \end{pmatrix}=u(x,y)\begin{pmatrix}
        \cos 2\pi (y+\psi(x,y))\\\sin 2\pi (y+\psi(x,y))
    \end{pmatrix}.
\] 
The function $ \psi $ is called a \textit{lift} of $ A $. Let $ \mu $ be any probability measure on $ \T^d\times\R $ which is invariant by the continuous map $ T:(x,y)\mapsto(x+\alpha,y+\psi(x,y)) $, projecting over Lebesgue measure on the first coordinate. 
Then the \textit{fibered rotation number} of $ (\alpha,A) $ is defined by
\[
    \rho(\alpha,A)=\int\psi\ \dif\mu \mod \Z,
\] 
which does not depend on the choices of $ \psi $ and $ \mu $  \cite{Herman83,JM1982rotation}.  It follows from the definition, that 
	 for any $ A\in\mathrm{SL}(2,\R) $, there exists a numerical constant $ C $ such that 
    \begin{equation} \label{rotationnumber1}
        \|\rho(\alpha,B)-\rho(\alpha,A)\|_\T<C\|B(\cdot)-A\|_0^{\frac{1}{2}},
    \end{equation} 
    where $ \lVert\cdot\rVert_0 $ denotes the $ C^0 $ norm.

Given $ \phi\in\T^d $, let $ R_{\phi}:=\begin{pmatrix}
    \cos 2\pi \phi&-\sin 2\pi \phi\\
    \sin 2\pi \phi&\cos 2\pi \phi
\end{pmatrix} $. 
If $ B:2\T^d\to\mathrm{SL}(2,\R) $ is homotopic to {$ R_{\frac{\langle n,\theta\rangle}{2}}: 2\T^d\to \mathrm{SL}(2,\R) $} for some $ n\in\Z^d $, then $ n $ is called the \textit{degree} of $ B $, {and is }denoted by $ \deg B $.   The fibered rotation number is conjugacy invariant, i.e.  if  $ (\alpha,A) $ is conjugated to $ (\alpha, A') $ by $ B\in C^0 (2\mathbb{T}^d,\mathrm{SL}(2, \R)) $ with $ \deg B=n $, then we have
    \begin{equation}\label{degree}
	\rho(\alpha, A)=\rho(\alpha, A')+\frac{\langle n,\alpha\rangle}{2} \mod\Z.
	\end{equation}

\subsection{Schr\"odinger operators}
    We consider now quasiperiodic Schr\"odinger operators \eqref{Op1.1}.
Note that a sequence $\{ u(n)\}_{n \in \Z}$ is a formal solution of the
eigenvalue equation $H_{v,\alpha,\theta} u=Eu$ if and only if

$$\begin{pmatrix}
u(n+1)\\u(n)\end{pmatrix}=S_{E}^{v}(\theta+n\alpha) \cdot
\begin{pmatrix} u(n)\\u(n-1) \end{pmatrix},$$ where we denote
\begin{eqnarray*}
S_{E}^{v}(\theta)=\left( \begin{array}{ccc}
 E-v(\theta) &  -1\cr
  1 & 0\end{array} \right)\in \op{SL}(2,\mathbb{R}).
\end{eqnarray*}
 Then $(\alpha,S_{E}^{v} )$ can be seen as a  quasi-periodic cocycle, and we
 call it \textit{Schr\"odinger cocycle}.  When consider AMO case, that is, $ v(x)=2\lambda\cos(2\pi x) $, we simply denote the cocycle as $ (\alpha,S_E^{\lambda}) $, and  call it {\it almost Mathieu cocycle}.

It is well-known that $ E\notin \Sigma_{v,\alpha} $ if and only if $ (\alpha, S_E^v) $ is uniformly hyperbolic \cite{Joh1986Exponential}. It is easy to see that there is a continuous  branch $ \rho(\alpha,S_E^v) \in[0,1/2] $, which is related to  the integrated density of states as  \cite{JM1982rotation,AS1983Almost}
\begin{equation}\label{N=rho}
    \mathcal{N}_{v,\alpha}(E)=1-2\rho(\alpha,S_E^v).
\end{equation}

\subsection{Arithmetic resonances}

Let $ \alpha\in\R^d $, $ \theta\in\R $, $ \epsilon_0>0 $.  Recall  that $n$ is a $\epsilon_0$-resonance of $\theta$, if 
\begin{equation}\label{upper+}
\|\theta -\langle n,\alpha \rangle\|_{\T}\leq e^{-|n|\epsilon_0}
\end{equation} 
and 
\begin{equation}\label{upper++}
\|\theta -\langle n,\alpha\rangle\|_{\T}=\min_{|j|\leq |n|}\|\theta -\langle j,\alpha\rangle\|_{\T}.
\end{equation} 
For fixed $ \alpha $ and $ \theta $, we can order the $ \epsilon_0 $-resonances $ 0=n_0<|n_1|\leq |n_2|\leq \cdots $. We say that $ \theta $ is $ \epsilon_0 $-resonant if the set of $\epsilon_0$-resonances is infinite, otherwise we say $ \theta $ is $ \epsilon_0 $-non-resonant.

\begin{remark}
As pointed out by Avila-Jitomirskaya \cite{aj1}, if $\alpha \in \op{DC}_d(\gamma,\tau)$, then \eqref{upper+} implies \eqref{upper++} for $|n|>n(\gamma,\tau)$. Moreover, the Diophantine assumption  immediately implies exponential repulsion of resonances: 
\begin{equation}\label{exp-repulusion}
    |n_{s+1}|\geq c e^{c\epsilon_0|n_{s}|}
\end{equation}
where $c=c(\alpha,\epsilon_0)>0$.
\end{remark} 

\section{Spectral measures and $ m $-functions}

\subsection{$m$-function and Jitomirskaya-Last inequality}
 It is known that the study of the spectral measure $ \mu $ of the Schr\"odinger operator: 
\begin{equation} \label{H}
    (Hu)(n)=u(n+1)+u(n-1)+v(n) u(n)
\end{equation} 
is related to the study of the Weyl-Titchmarsh $ m $-function. We will briefly list the necessary related facts from \cite{AvilaHolderContinuityAbsolutely2011,DamanikScalingEstimatesSolutions2005,KillipDynamicalUpperBounds2003,JL2000Powerlaw}.
Let $ \Z^+=\{1,2,3,\dots\} $, and $ \Z^-=\{\dots, -2,-1,0\} $. Consider the family of phase boundary condition:
\begin{equation} \label{boundarycondtion}
    u(0)\cos\beta+u(1)\sin\beta=0,
\end{equation} 
where $ -\pi/2<\beta\leq \pi/2 $. We denote by $ u_{\beta}^\pm $ the solution of $ Hu=Eu $ on $ \Z^\pm $ satisfying the boundary conditons (\ref{boundarycondtion}) and normalized by $ u_{\beta}^\pm(0)^2+u_{\beta}^\pm(1)^2=1 $. 

Consider energies $ z=E+\ii\epsilon $, $ E\in\R $, $ \epsilon>0 $. Then there are non-zero solutions $ u_z^{\pm} $ of $ H u_z^{\pm}=z u_z^{\pm} $ which are $ l^2 $ at $ \pm \infty $, well defined up to normalization. 
Then the half-line $ m $-functions are defined as:
\begin{equation*} 
        m^\pm = \mp \frac{u_z^\pm(1)}{u_z^\pm(0)},
 \end{equation*}        
 and then we define       
   \begin{equation} \label{m_betapm}  
    m_\beta^+ =R_{-\beta/2\pi}\cdot m^+,\qquad
    m_\beta^- =R_{\beta/2\pi}\cdot m^-,
\end{equation} 
where we make use of the action of $ \mathrm{SL}(2,\C) $ on $ \overline{\C} $,
$$
    \begin{pmatrix}
        a&b\\c&d
    \end{pmatrix}\cdot z=\frac{az+b}{cz+d}.
$$ 
The whole-line $ m $-function is then defined for $ \Im z>0 $ by 
 \begin{equation}\label{M(z)}
     M(z)=\frac{m^+(z)m^-(z)-1}{m^+(z) + m^-(z)}.
 \end{equation} 
\hspace{2em}

We also define 
\[
    \begin{aligned}
        &(K^{+}(E)\psi)(n)=\sum_{m=1}^n k_{0}^{+}(n,m)\psi(m),\\
        &(K^{-}(E)\psi)(n)=\sum_{m=-n+1}^0 k_{0}^{-}(n,m)\psi(m) ,
    \end{aligned}
\] 
where
\[
    k_0^\pm(n,m)=u_0^\pm(n)u_{\pi/2}^\pm(m)-u_{\pi/2}^\pm(n)u_{0}^\pm(m). 
\] 
and one then can view $ K^{\pm}(E) $  as the integral operators acting on the $ [L] $ or $ [L]+1 $-dimensional Hilbert space \cite[]{KillipDynamicalUpperBounds2003}. 
Indeed,  for any  $ u^+:\Z^+\to\C $ (resp. $ u^-:\Z^-\to\C $), we denote by $ \|u^+\|_L $ (resp.  $ \|u^-\|_L$) the norm over a lattice interval of length $ L $:
$$
    \|u^+\|_L:=\left[\sum_{n=1}^{[L]}|u^+(n)|^2+(L-[L])|u^+([L]+1)|^2\right]^{1/2},
$$ 
$$
    \|u^-\|_L:=\left[\sum_{n=0}^{[L]-1}|u^-(-n)|^2+(L-[L])|u^-(-[L])|^2\right]^{1/2},
$$ 
where $ [L] $ denotes the integer part of $ L .$ The key observation is the following:

\begin{lemma}[\cite{KillipDynamicalUpperBounds2003}]
    The Hilbert-Schmidt norm of $ K^{\pm}(E) $ at the scale $ L $ is given by 
    \begin{align}\label{HSofK(E)}
        ||| K^\pm(E)|||_{L}^2
        =\inf_{\beta}\|u_{\beta}^\pm\|_L^2\|u_{\beta+\pi/2}^\pm\|_L^2.
    \end{align}
\end{lemma}

Notice that $ ||| K^\pm(E)|||_{L} $ is strictly increasing. Hence $ \epsilon||| K^\pm(E)|||_{L}=1 $ determines $ L $ as a function of $ \epsilon $, or the inverse. 
It is known that the Jitomirskaya-Last inequality \cite{JL1999Powerlaw,JL2000Powerlaw}, provides a link between $ m $-functions and the solutions of (\ref{H}) satisfied the corresponding boundary conditions.
The following result is an application of the well-known Jitomirskaya-Last inequality, but based on a different version in \cite[]{KillipDynamicalUpperBounds2003}.
\begin{lemma}[\cite{DamanikScalingEstimatesSolutions2005}]\label{Damaniklem}
   For any pair $ (\epsilon,L^\pm) $ satisfies $ ||| K^\pm(E)|||_{L^\pm(\epsilon)}^2=\frac{1}{\epsilon^2} $, we have
   \[
     C^{-1}\leq \epsilon \Im m_{\beta}^\pm(E+i\epsilon)\|u_{\beta}^{\pm}\|_{L^\pm(\epsilon)}^2\leq C.
   \] 
\end{lemma}
\begin{remark}
        The case $ \beta=0 $ is proved by \cite[Proposition 3]{DamanikScalingEstimatesSolutions2005}, while the proof  can be extended to any $ \beta $ almost without changes.
\end{remark}
 As further explored in \cite{AvilaHolderContinuityAbsolutely2011}, the invariant $ ||| K^\pm(E)|||_{L} $ has another representation.
Let $ L=2k $, $ k\geq 1 $ integer, denote   $ (\alpha,S_E^v)^n=(n\alpha,\mathcal{A}_n) $, and define 
$$
    \begin{aligned}\label{pk}
        P_{(k),+}&=\sum_{j=1}^k \mathcal{A}_{2j-1}^*(\theta+\alpha)\mathcal{A}_{2j-1}(\theta+\alpha),\\
        P_{(k),-}&=\sum_{j=1}^{k} \mathcal{A}_{-(2j-1)}^*(\theta+\alpha)\mathcal{A}_{-(2j-1)}(\theta+\alpha).
    \end{aligned}
$$ 
then $ P_{(k),\pm} $ is an increasing family of positive self-adjoint $ 2\times 2 $ matrix. It is immediate to see that if $ L=2k $ then 
\begin{equation} \label{normP_k}
    \|u_\beta^\pm\|_L^2=\langle P_{(k),\pm} \begin{pmatrix}
        u_\beta^\pm(1)\\u_\beta^\pm(0)
    \end{pmatrix},\begin{pmatrix}
        u_\beta^\pm(1)\\u_\beta^\pm(0)
    \end{pmatrix}\rangle\leq \|P_{(k),\pm}\|,
\end{equation} 
    with equality for $ \beta $ maximizing $ \|u_\beta^\pm\|_L^2 $. 
    Thus  we obtain
    \begin{equation} \label{detP_k}
        \det P_{(k),\pm}=\inf_\beta\|u_\beta^\pm\|_L^2\|u_{\beta+\pi/2}^\pm\|_L^2,
    \end{equation} 
 and   the infimum is attained at the critical points of $ \beta\mapsto\|u_\beta^\pm\|_L^2 $. Combines (\ref{HSofK(E)}) and (\ref{detP_k}), one has 
    \begin{equation}\label{kp}
     \det P_{(k),\pm}=||| K^\pm(E)|||_{L}^2 .
     \end{equation} 
We thus  give another version of Jitomirskaya-Last inequality,
 in terms of $P_{k,+}$:
\begin{lemma}[{\cite[Remark 4.1]{AvilaHolderContinuityAbsolutely2011}}]\label{AJDT}
    Let $ \epsilon $ be such that $ \det P_{(k),+}=\frac{1}{\epsilon^2} $, then we have 
    \[
        C^{-1}<\frac{2\epsilon \sup_\beta|m^+_\beta(E+i\epsilon)|}{\|P_{(k),+}^{-1}\|}<C.
    \] 
\end{lemma}

As a direct consequence, we have the following:
  \begin{corollary}\label{Corspecmeas}
        Let $(k, \epsilon)$ be such that $ \det P_{(k),+}=\frac{1}{\epsilon^2} $, then we have the following:
        \begin{eqnarray}
          \label{upp}  \epsilon\Im M(E+i\epsilon)  &\leq& C \|P_{(k),+}^{-1}\|, \\
           \label{low} \epsilon\Im M(E+i\epsilon)&\geq& c\sup_{\beta}\min\set{\|u_{\beta}^+\|_{2k}^{-2},\|u_{\beta}^-\|_{L^-(\epsilon)}^{-2}},
        \end{eqnarray}
        where $L^{-}(\epsilon)$ is given by $||| K^-(E)|||_{L^{-}}^2=\frac{1}{\epsilon^2}.$ 
    \end{corollary}

\begin{proof}
It was shown in \cite{AvilaHolderContinuityAbsolutely2011,DKL2000Uniform} that, as a corollary of the maximal modulus principle, 
  $\Im M \leq \sup_{\beta}|m^+_\beta|, $ 
then \eqref{upp} follows from Lemma \ref{AJDT}. On the other hand, by \eqref{m_betapm} and \eqref{M(z)}, direct computation shows that 
\[
    M(z)=\frac{m_{\beta}^+(z)m_{\beta}^-(z)-1}{m_{\beta}^+(z) + m_{\beta}^-(z)},
\] 
and thus
    \[
        \begin{aligned}
  \Im M&=\frac{1+|m_{\beta}^+|^2}{|m_{\beta}^+ + m_{\beta}^-|^2}\Im m_{\beta}^- +\frac{1+|m_{\beta}^-|^2}{|m_{\beta}^+ + m_{\beta}^-|^2}\Im m_{\beta}^+ \geq \frac{1}{4}\min\{\Im m_{\beta}^+,\Im m_{\beta}^-\}, 
        \end{aligned}
    \] 
then \eqref{low} follows from \eqref{kp} and  Lemma \ref{Damaniklem}.
 \end{proof}

 \begin{remark}
 Here we should be careful that while we select $L^+(\epsilon)=2k$,  it is not necessary $L^+(\epsilon)$ is equal to $L^-(\epsilon)$.
 \end{remark}

 \subsection{Spectral measure}
 Recall that the whole line $ m $-function coincides with the Borel transform of the canonical spectral measure $ \mu $ of the operator (\ref{H}): 
 \begin{equation} \label{MU}
     M(z)=\int \frac{d\mu(E')}{E'-z}=\langle \delta_0, (H-z)^{-1}\delta_0\rangle+\langle \delta_1, (H-z)^{-1}\delta_1\rangle .
 \end{equation} 
 Thus, clearly, for any $ \epsilon>0 $, we have 
 \begin{equation}\label{upperb-mu} 
    \mu(E-\epsilon,E+\epsilon)\leq 2\epsilon\Im M(E+i\epsilon) .
 \end{equation} 
Deriving a lower bound for $\mu(E-\epsilon, E+\epsilon)$ from $\epsilon\Im M(E+ i\epsilon)$ is a delicate endeavor. Our approach involves first obtaining a stratified upper bound for $\epsilon\Im M(E+ i\epsilon)$. This upper bound is then employed to deduce a stratified upper bound for the measure. Utilizing the identity: $$
\epsilon\Im M(E+i\epsilon) = \int \frac{\epsilon^2}{(E' - E)^2 + \epsilon^2} \dif \mu(E'),
$$ we further derive a lower bound for $\mu(E-\epsilon, E+\epsilon)$ in terms of $\epsilon\Im M(E+i\epsilon)$. The details of this process can be found in Section \ref{5.3}.

Corollary \ref{Corspecmeas} highlights that to estimate $\Im M(E+ i\epsilon)$, it suffices to calculate $\|P_{(k),+}^{-1}\|$ and any component of $u_\beta^\pm$'s $L^2$ norm. As we will demonstrate later, the $L^2$ norm of $u_\beta^\pm$ corresponds to the first entry of the matrix $P_{(k),\pm}$. Therefore, the foundation for understanding the asymptotic local distribution of the spectral measure lies in examining the asymptotic structure of $P_{(k),\pm}$, which we will address in the subsequent section.

\section{Asymptotic formula of $ P_{(k),\pm} $}

Given a  resonant set
 $ (n_s)_{s\geq 1}\subset \Z^d $ may be finite or infinite, and we label it as  $0<|n_1|< |n_2|<\cdots$.  If the sequence is finite, we assume the last one is $ n_m $, and formally set $ n_{m+1}=\infty $. If the sequence is infinite, we assume that    
\begin{equation} \label{cond-repulsion}
    \lim\limits_{s\to \infty}\frac{|n_s|}{|n_{s+1}|}=0.
\end{equation}  
Let us say that $ (\alpha,A)$ is $ (C_0',C_0, \gamma_0, (n_s)_{s\geq s_0}) $-good,  if for any $s\geq s_0$, and denoting $ n=|n_s| $ and $ N=|n_{s+1}| $,  there exists $ B_n\in C^\omega(\T^d,\mathrm{PSL}(2,\R)) $ with $\deg B_n=n $, $ \|B_n\|_0\leq C'_0n^{C_0} $ such that 
    \begin{equation}\label{alre}
        \begin{aligned}
            B_n^{-1}(\theta+\alpha)A(\theta)B_n(\theta)&=M^{-1} \exp 2\pi  \begin{pmatrix}
                \ii t_n &z_n \\ \overline{z_n}&-\ii t_n
            \end{pmatrix}M+F_n(\theta),
        \end{aligned}
    \end{equation}
    with 
    $$
    0\leq \sqrt{t_n^2-|z_n|^2},\, |z_n|\leq \frac{1}{4} ,
    $$ 
    and 
    $$ 
    \|F_n\|_0\leq e^{-\gamma_0 N} .
    $$ 
   
      \begin{remark}   
      If $d=1$,  similar concept of $ (C_0',C_0,\gamma_0,  (n_s)_{s\geq s_0}) $-goodness was first introduced by Avila-Jitomirskaya \cite{AvilaHolderContinuityAbsolutely2011}, and proved in \cite{aj1}. They prove that  if  $z_n$ has sub-exponential decay:
     $$|z_n| \leq C e^{-cn (\ln (1+n))^{-C} },$$
     then the spectral measure is $1/2$-H\"older continuous  (Theorem 4.1 of \cite{AvilaHolderContinuityAbsolutely2011}).
 Our contribution here is to show that when $z_n$ decays exponentially, one can achieve stratified H\"older continuity (as stated in Proposition \ref{THM5.4}). Moreover, when the decay is real exponential, implying both upper and lower bounds of $z_n$ decay exponentially, we can delve into the exact local distribution of the spectral measure.
 \end{remark}    

   \subsection{Asymptotic formula}

   In the following, we  will first show  if the Schr\"odinger cocycle  $ (\alpha,S_E^v) $  is $ (C_0',C_0, \gamma_0,  (n_s)_{s\geq s_0})$-good, then  one can obtain  the   asymptotic formula of $\|P_{(k),\pm}^{-1}\|^{-1} $ and $ \det P_{(k),\pm} $. Before stating the result, note as a direct consequence of  Lemma \ref{lem:quanti-Schurlem}, if  $ (\alpha,A) $  is $ (C_0',C_0, \gamma_0, (n_s)_{s\geq s_0})$-good, then for any $s\geq s_0$,  
   there exists $\Phi_n\in C^\omega(\T^d,\mathrm{PSL}(2,\C)) $ with $ \|\Phi_n\|_0\leq C'_0n^{C_0} $ such that 
    \[
        \begin{aligned}
           \Phi_n^{-1}(\theta+\alpha)S_E^v(\theta)\Phi_n(\theta)&=\begin{pmatrix}
            e^{2\pi \ii\rho_n}&\nu_n\\
            0&e^{-2\pi \ii\rho_n}
        \end{pmatrix}+H_n(\theta),
        \end{aligned}
    \]
    where $\rho_n=\sqrt{t_n^2-|z_n|^2}$ with estimates  $|z_n|\leq |\nu_n|\leq 2|z_n|$, $\|H_n\|_0\leq e^{-\gamma_0N}$. 

    \begin{remark}
        In the result of the paper, we will use this as formulation of  $ (C_0',C_0, \gamma_0, (n_s)_{s\geq s_0})$-goodness. We have to start with \eqref{alre}, as we need to use the invariance of rotation number (one can not define fibered rotation number for general $\op{SL}(2,\C)$ cocycles). More precisely,  apply 
     \eqref{degree} and \eqref{rotationnumber1} to \eqref{alre}, it follows that 
    \begin{equation}\label{Delta-eta}
        |\|2\rho(\alpha,S_E^v)-\langle n_s,\alpha\rangle\|_\T-\|2\rho_n\|_\T|\leq C e^{-\frac{1}{2}\gamma_0N},
    \end{equation}
    which implies that $\|2\rho_n\|_\T $ has exponential decay.
    \end{remark}

    Henceforth, we denote 
    \[
        -\frac{\ln \|2\rho_n\|_\T}{n}=\widehat{\eta}_n \in (0,\infty], \qquad
            -\frac{\ln |\nu_n|}{n}= \zeta_n\in (0,\infty],
    \] 
    for simplicity. Upon obtaining these,  we have the following:

    \begin{proposition}\label{Prop:AsyFormula}
      Suppose  that $ (\alpha,S_E^v) $ is $ (C_0',C_0, \gamma_0,  (n_s)_{s\geq s_0}) $-good, then there exists $ K_0=K_0(C_0',C_0, \gamma_0)>0$ such that  for $ k\geq K_0  $ and $e^{ \frac{\gamma_0}{5} n}\leq k\leq e^{\frac{\gamma_0}{5} N}$, we have        \[
            (\ln k)^{-8C_0} k\leq \|P_{(k),\pm}^{-1}\|^{-1}\leq (\ln k)^{8C_0}k, 
        \] 
        \[
            (\ln k)^{- 16C_0}k^{ \eta_+^n(k)  }\leq  \det P_{(k),\pm}\leq (\ln k)^{ 16C_0 }k^{ \eta_+^n(k)  },
        \] 
        where $  \eta_+^n(k) = \max\{\eta^n(k),2\}$, and  $\eta^n(k)$ is defined as: 
        \begin{equation} \label{eta_s}
            \eta^n(k)=\left\lbrace\begin{aligned}
                &4-\frac{2\zeta_n n}{\ln k}&\text{if }e^{\frac{\gamma_0}{5} n}\leq k\leq e^{\widehat{\eta}_n n},\\
                &2+ \frac{2\widehat{\eta}_n n-2\zeta_n n}{\ln k} &\text{if }e^{\widehat{\eta}_n n}\leq k\leq e^{\frac{\gamma_0}{5} N}.
            \end{aligned}\right.
        \end{equation} 
        Moreover,  there exists $ \beta_n\in(-\pi/2,\pi/2] $, such that\begin{equation*} 
            (\ln k)^{-8C_0} k\leq \|u_{\beta_n}^\pm\|_{2k}^{2}\leq (\ln k)^{8C_0}k.
        \end{equation*}
    \end{proposition}
    
    \begin{remark}\label{Rem:4.3}
    If  $\widehat{\eta}_n=\infty $, then 
   $$\eta^n(k)= 4-\frac{2\zeta_n n}{\ln k} \qquad k\geq e^{\frac{\gamma_0}{5} n}.$$
 If $\zeta_n=\infty$, then  $ \eta_+^n(k)\equiv 2 $.
    Thus, the function $ \eta_+^n(k) $ is well-defined.
    \end{remark}

    \begin{proof}
     We only prove the estimates for $ \|P_{(k),+}^{-1}\|^{-1} $, $ \|u_{\beta_n}^+\|_{2k}^{2} $ and $ \det P_{(k),+} $, the negative side is similar. 
     The proof is based on the following result. 
\begin{lemma}[\cite{AvilaHolderContinuityAbsolutely2011}]\label{lem4.3inAJ}\label{lem4.4inAJ}
    Let 
    $
        T=\begin{pmatrix}
            e^{2\pi \ii \tilde{\theta}}&t\\
            0&e^{-2\pi\ii \tilde{\theta}}
        \end{pmatrix}
    $, $ \tilde{T}:\T^d\to\mathrm{SL}(2,\C) $.
Denote \[
X=\sum_{j=1}^k T^*_{2j-1}T_{2j-1}, \qquad \tilde{X}=\sum_{j=1}^k \tilde{T}^*_{2j-1}\tilde{T}_{2j-1}.
\]  Then we have the following:
    \begin{equation*} 
        \|X\| \thickapprox k(1+|t|^2\min\{k^2,\|2\tilde{\theta}\|_\T^{-2}\}),
    \end{equation*} 
    \begin{equation*} 
        \|X^{-1}\|^{-1}\thickapprox k,
    \end{equation*} 
       \begin{equation*} 
        \|\tilde{X}-X\|_0\leq 1, \text{ provided that } \|\tilde{T}-T\|_0\leq ck^{-2}(1+2k\|t\|_0)^{-2}.
    \end{equation*} 
\end{lemma}

As $ (\alpha,S_E^v) $ is $ (C_0',C_0,\epsilon_0, \gamma_0) $-good, we can rewrite $(\alpha, S_E^v)$ as 
    $$
       \Phi_n(\theta+\alpha)^{-1}S_E^{v}(\theta)\Phi_n(\theta)=T+H_n:= \begin{pmatrix}
            e^{2\pi \ii\rho_n}&\nu_n\\
            0&e^{-2\pi \ii\rho_n}
        \end{pmatrix} + H_n.
    $$ 
with
 estimates $ \|\Phi_n\|_0\leq C'_0n^{C_0} $, $ |\nu_n|< 1 $, $ \|H_n\|_0\leq e^{-\gamma_0 N}$.  Denote    
    \begin{equation*} 
        \tilde{T}(\theta)=\Phi_n(\theta+\alpha)^{-1}S_E^{v}(\theta)\Phi_n(\theta)
  \end{equation*}   
 and define $ X $, $\tilde X$ as in Lemma \ref{lem4.3inAJ}.
First  we have the following basic observations:
    \begin{lemma}\label{P_kX}
        Denote 
         $ \tilde{X}=\begin{pmatrix}
            \tilde{x}_{11}&\tilde{x}_{12}\\ \tilde{x}_{21}&\tilde{x}_{22}
        \end{pmatrix} $, there exists $ \beta_n\in (-\pi/2,\pi/2] $ such that \begin{equation} \label{m^+}
                \|\Phi_n\|_0^{-4}\tilde{x}_{11}\leq {\|u_{\beta_n}^+\|_{2k}^2}\leq \|\Phi_n\|_0^4 \tilde{x}_{11}.
            \end{equation} 
   Moreover,  we have estimates 
        \begin{enumerate}
            \Item 
            \begin{align}\label{C1}
            \|\Phi_n\|_0^{-4}\|\tilde{X}(\theta+\alpha)\|  
                \leq  
                \|P_{(k),+}\| 
                \leq 
                \|\Phi_n\|_0^4\|\tilde{X}(\theta+\alpha)\|,
            \end{align}
            \Item   
            \begin{equation}\label{C2} 
                \|\Phi_n\|_0^{-4}\|\tilde{X}(\theta+\alpha)^{-1}\|^{-1} 
                \leq 
                \|P_{(k),+}^{-1}\|^{-1} 
                \leq 
                \|\Phi_n\|_0^{4}\|\tilde{X}(\theta+\alpha)^{-1}\|^{-1},
                     \end{equation}
            \Item  
            \begin{equation} \label{C3}
        \|\Phi_n\|_0^{-8}\|\tilde{X}\|\|\tilde{X}^{-1}\|^{-1}\leq \det P_{(k),+}\leq \|\Phi_n\|_0^8\|\tilde{X}\|\|\tilde{X}^{-1}\|^{-1}.
    \end{equation} 
        \end{enumerate}

    \end{lemma}
    \begin{proof} 
Denote   $ (\alpha,S_E^v)^s=(s\alpha,\mathcal{A}_s) $ and dentoe $\tilde{u}_{\beta}=\begin{pmatrix}
                u_{\beta}^+(1)\\u_{\beta}^+(0)
            \end{pmatrix} $. Just notice that
    \begin{equation}\label{ub}
            \begin{aligned}
                \|u_\beta^+\|_{2k}^2
                =&\sum_{j=1}^k \langle \mathcal{A}_{2j-1}(\theta+\alpha)\tilde{u}_{\beta}, \mathcal{A}_{2j-1}(\theta+\alpha)\tilde{u}_{\beta}\rangle\\
                =&\sum_{j=1}^k \|\Phi_n(\theta+2j\alpha)\tilde{T}_{2j-1}(\theta+\alpha)\Phi_n(\theta+\alpha)^{-1}\tilde{u}_{\beta}\|^2\\
                \leq& \|\Phi_n\|_0^2\sum_{j=1}^k\langle \tilde{T}_{2j-1}(\theta+\alpha)\Phi_n(\theta+\alpha)^{-1}\tilde{u}_{\beta}, \tilde{T}_{2j-1}(\theta+\alpha)\Phi_n(\theta+\alpha)^{-1}\tilde{u}_{\beta}\rangle\\
                \leq& \|\Phi_n\|_0^4\sum_{j=1}^k\langle \tilde{T}_{2j-1}(\theta+\alpha)\tilde{u}_{\beta'}, \tilde{T}_{2j-1}(\theta+\alpha)\tilde{u}_{\beta'}\rangle
                =\|\Phi_n\|_0^4\langle \tilde{X}(\theta+\alpha)\tilde{u}_{\beta'},\tilde{u}_{\beta'}\rangle
            \end{aligned}
        \end{equation}
        where  
        \begin{equation} \label{beta'}
\tilde{u}_{\beta'}=\frac{\Phi_n(\theta+\alpha)^{-1}\tilde{u}_{\beta}} {\left\|\Phi_n(\theta+\alpha)^{-1}\tilde{u}_{\beta} \right\|},
        \end{equation} 
        and symmetrically, we have 
 \begin{equation} \label{sym}
            \langle \tilde{X}(\theta+\alpha)\tilde{u}_{\beta'},\tilde{u}_{\beta'}\rangle\leq \|\Phi_n\|_0^4 \|u_{\tilde{\beta}}^+\|_L^2.
          \end{equation} 
        where
        \begin{equation} \label{tilde-beta}
            \tilde{u}_{\tilde{\beta}}=\frac{\Phi_n(\theta+\alpha)\tilde{u}_{\beta'}}{\left\|\Phi_n(\theta+\alpha)\tilde{u}_{\beta'}\right\|}.
        \end{equation}
        Combine \eqref{beta'} and \eqref{tilde-beta}, we have $ \beta=\tilde{\beta} $, so there is a mutual determination between $ \beta' $ and $ \beta $.

        We first prove (\ref{m^+}).  Let $ \beta'=0 $, there exists an unique $ \beta_n $ determined by \eqref{beta'} and \eqref{tilde-beta}, such that 
        $$
            \begin{aligned}
    \|u_{\beta_n}^+\|_{2k}^2
                \leq \|\Phi_n\|_0^4\langle \tilde{X}(\theta+\alpha)\begin{pmatrix}
                    u_0^+(1)\\u_0^+(0)
                \end{pmatrix},\begin{pmatrix}
                    u_0^+(1)\\u_0^+(0)
                \end{pmatrix}\rangle=\|\Phi_n\|_0^4 \tilde{x}_{11}.
            \end{aligned}
        $$ 
        And the LHS of (\ref{m^+}) holds by symmetry \eqref{sym}. 

    Recall that \begin{equation} \label{normP_k'}
    \|u_\beta^\pm\|_{2k}^2=\langle P_{(k),\pm} \begin{pmatrix}
        u_\beta^\pm(1)\\u_\beta^\pm(0)
    \end{pmatrix},\begin{pmatrix}
        u_\beta^\pm(1)\\u_\beta^\pm(0)
    \end{pmatrix}\rangle\leq \|P_{(k),\pm}\|,
\end{equation} 
then the  LHS of (\ref{C1}) follows from \eqref{sym} and  \eqref{normP_k'}. Note \eqref{normP_k'} takes  equality for $ \beta $ maximizing $ \|u_\beta^\pm\|_L^2 $, then one can take 
supremum  of \eqref{ub} getting the proof of  the  RHS of (\ref{C1}).  (\ref{C2}) follows from the similar argument, we omit the details, while \eqref{C3} follows from \eqref{C1}, \eqref{C2} and the fact that $$ \det P_{(k),+}= \|P_{(k),+}\|\|P_{(k),+}^{-1}\|^{-1}.$$
We thus finish the whole proof.  \qedhere

    \end{proof}
\begin{remark}
Utilizing symmetry, we derive a lower bound for $\Im M(E + \ii\epsilon)$, an essential step in our proof.  While the observation is simple, it serves as the starting point of the whole proof. 
\end{remark}
    Once we have these, we can finish the estimates.   Let $ k_{\Delta}\geq 1 $ be maximal such that for $ 1\leq k<k_{\Delta} $, we have 
      $
        \|\tilde{X}-X\|_0\leq 1.
  $ 
   Notice that 
    $
        \|\tilde{T}-T\|_0\leq \|H_n\|_0,
    $
Lemma \ref{lem4.4inAJ}  implies that 
    \begin{equation} 
        \|H_n\|_0\geq ck_{\Delta}^{-2}(1+2k_{\Delta}|\nu_n|)^{-2}\geq ck_{\Delta}^{-4},
    \end{equation} 
    so that 
    \begin{equation*}\label{k_Delta}
        k_{\Delta}\geq c^{\frac{1}{4}}e^{\frac{\gamma_0}{4} N}.
    \end{equation*} 
     Then there exists $ K_0=K_0(C_0',C_0,\gamma_0)>0 $, such that  for any $k\geq K_0$ and $ e^{\frac{\gamma_0}{5} n} \leq k\leq  e^{\frac{\gamma_0}{5}  N} $, we have 
    \begin{equation}\label{eq:lnkbound}
        k< k_\Delta,\ C'_0n^{C_0}\leq \tfrac{1}{2}(\ln k)^{2C_0}.
    \end{equation}
    Indeed, there exists $N_1(C_0',C_0,\gamma_0)>0$ such that  $\ C'_0n^{C_0}\leq \tfrac{1}{2}(\frac{\gamma_0}{5}n)^{2C_0}$ and $k<k_{\Delta}$ for $n>N_1 $. 
    If the resonant set is finite, and the last resonance $|n_m|<N_1$, there exists $N_2(C_0',C_0)>0$ such that $\ C'_0N_1^{C_0}\leq \tfrac{1}{2}(\ln k)^{2C_0}$ for $k>N_2$. It suffices to take $K_0=\max\{e^{\frac{\gamma_0}{5}N_1},N_2\} $.

  Meanwhile, as  $ X=\begin{pmatrix}
        k&x_1\\\bar{x}_1&x_{2}
    \end{pmatrix} $ (see \cite{AvilaHolderContinuityAbsolutely2011}), and for $ 1\leq k<k_{\Delta} $, 
    $$
        \begin{aligned}
            \|X\|-1&\leq & 
            &\|\tilde{X}\|&
            &\leq \|X\|+1,\\
            \|X^{-1}\|^{-1}-1&\leq &
            &\|\tilde{X}^{-1}\|^{-1}&
            &\leq\|X^{-1}\|^{-1}+1.
        \end{aligned}
    $$ 
   Lemma \ref{lem4.3inAJ}, Lemma \ref{P_kX} and the fact that $ \|\Phi_n\|\leq C'_0n^{C_0} $ imply that 
    \begin{align} 
        (C'_0n^{C_0})^{-8}\leq \frac{\det P_{(k),+}}{k^2(1+|\nu_n|^2\min\{k^2,\|2\rho_n\|_\T^{-2}\}}\leq (C'_0n^{C_0})^{8} , \ C\leq k< k_\Delta,\label{P11}
    \end{align} 
    \begin{align} 
        (C'_0n^{C_0})^{-4}\leq \frac{\|P_{(k),+}^{-1}\|^{-1}}{k}\leq (C'_0n^{C_0})^{4}, \ C\leq k< k_\Delta,\label{P21}
    \end{align} 
    \begin{align} 
        (C'_0n^{C_0})^{-4}\leq \frac{\|u_{\beta_n}^+\|_{2k}^2}{k}\leq (C'_0n^{C_0})^{4}, \ C\leq k< k_\Delta.
    \end{align}

Consequently, \eqref{eq:lnkbound} implies that
    \begin{eqnarray*} 
        (\ln k)^{-8C_0} k &\leq& \|P_{(k),+}^{-1}\|^{-1}\leq (\ln k)^{8C_0}k,\\
        (\ln k)^{-8C_0} k&\leq& \|u_{\beta_n}^+\|_{2k}^{2}\leq (\ln k)^{8C_0}k.
    \end{eqnarray*}   
Moreover,  if $e^{\frac{\gamma_0}{5}n}\leq k\leq \|2\rho_n\|_\T^{-1}$, we have 
          $$(\ln k)^{-16C_0} k^2(1+|\nu_n|^2k^2)\leq \det P_{(k),+}\leq \tfrac{1}{2}(\ln k)^{16C_0} k^2(1+|\nu_n|^2k^2), $$
          if $\|2\rho_n\|_\T^{-1}\leq k\leq e^{\frac{\gamma_0}{5}N}$, we have
          $$ (\ln k)^{-16C_0} k^2(1+\frac{|\nu_n|^2}{\|2\rho_n\|_\T^2})\leq \det P_{(k),+}\leq \tfrac{1}{2}(\ln k)^{16C_0} k^2(1+\frac{|\nu_n|^2}{\|2\rho_n\|_\T^2}),$$
          then the result follows. 
\end{proof}

While Proposition \ref{Prop:AsyFormula} mainly takes care the case the Schr\"odinger  cocycle is almost reducible, if the cocycle is reducible, one can obtain better estimates. In the following, we use the notation $a \thickapprox b$ $(a,b>0)$ denotes $C^{-1}a \leq b \leq  Ca$.

    \begin{proposition}\label{Prop:AsyFormula-2}
     We have the following:      \begin{enumerate}
          \item If $ (\alpha,S_E^v) $ is conjugated to $\begin{pmatrix}
          \cos 2\pi\rho & -\sin 2\pi\rho\\
          \sin 2\pi\rho &\cos 2\pi\rho
          \end{pmatrix}$. Then for all $ k $ sufficiently large,  we have       
          \[  \|P_{(k),\pm}^{-1}\|^{-1}\thickapprox k, \qquad
             \det P_{(k),\pm}\thickapprox k^{ 2 },
        \]
        Moreover,  there exists $ \beta\in(-\pi/2,\pi/2] $, such that $ 
            \|u_{\beta}^\pm\|_{2k}^{2}\thickapprox k. $
          \item If $ (\alpha,S_E^v) $ is conjugated to $\begin{pmatrix}
          1 & \nu\\
          0 &1
          \end{pmatrix}$ for some $\nu\neq 0$. Then for all $ k $ sufficiently large,  we have 
          \[
          \|P_{(k),\pm}^{-1}\|^{-1} \thickapprox k, 
        \qquad
             \det P_{(k),\pm}\thickapprox k^{ 4 },
        \]
        Moreover,  there exists $ \beta\in(-\pi/2,\pi/2] $, such that $ 
            \|u_{\beta}^\pm\|_{2k}^{2}\thickapprox k. $
      \end{enumerate}
    \end{proposition}
    \begin{proof}
        It is a direct application of Lemma \ref{lem4.3inAJ} and Lemma \ref{P_kX}.
    \end{proof}
 \begin{remark}\label{rem:general}
 As we can see, the proof of Proposition \ref{Prop:AsyFormula} doesn't use Schr\"odinger structure, thus the results of  Proposition \ref{Prop:AsyFormula} and Proposition \ref{Prop:AsyFormula-2} still hold for more general quasi-periodic cocycles.  It is only necessary to  take \eqref{normP_k} as the definition of $\|u_{\beta}^{\pm}\|^2_{2k} $, all the proofs remain valid.
    \end{remark}

\subsection{Stratified H\"older continuity}

As an immediate application of Proposition \ref{Prop:AsyFormula}, we show that if $\nu_n$ has uniform exponential decay, one can obtain stratified H\"older continuity:

\begin{proposition}\label{THM5.4}
    Let $\alpha\in \mathrm{DC}_d $, $v\in C^{\omega}(\T^d,\R)$, $E\in\Sigma_{v,\alpha}$ with 
    $$ 
    h\leq \delta(\alpha,2\rho(\alpha,S_E^v))=\delta<\infty .
    $$ Given $0<\e\ll \min\{1,h\} $. If there exists $s_0=s_0(\varepsilon)$ large enough, such that
\\ 
$ \mathbf{(H1)}$:  $ (\alpha,S_E^v) $ is $ (C_0',C_0, \gamma_0, (n_s)_{s\geq s_0}) $-good, with $ \gamma_0(\varepsilon)<h<\infty $; \\
$ \mathbf{(H2)}$:   $\zeta_n\in [(1-\varepsilon)h,\infty]$.\\
    Then  there exist $ \epsilon_*=\epsilon_*(C_0', C_0, \gamma_0,  (n_s), \varepsilon) >0 $,  such that for any $\epsilon<\epsilon_*$,
    \begin{equation}\label{eq:upperholder}
        \mu(E- \epsilon, E+ \epsilon)\leq \epsilon^{\frac{\delta}{2\delta-h} -C\varepsilon}. 
    \end{equation}
            \end{proposition}
    \begin{proof}
   Denote 
 \[
        -\frac{\ln   \|2\rho(\alpha,S_E^v)-\langle n_s,\alpha\rangle\|_\T
        }{n}=\eta_n, 
    \]   then by the definition of $\delta(\alpha,2\rho(\alpha,S_E^v) )$ and \eqref{cond-repulsion}, one can select $s_1\geq s_0$ to be the minimal integer such that 
    \begin{eqnarray}
        \eta_n&\leq& \min\{\delta+\e, \frac{\gamma_0 N}{5n}\}, \label{noH3:s_0prime}\\
         C e^{-\frac{1}{4}\gamma_0N}&\leq&  \frac{\varepsilon^2}{100}, \label{noH3:s_0prime2}
    \end{eqnarray}
 and  select $ K_1 \geq \max\{K_0,e^{\frac{\gamma_0}{5}|n_{s_1}|}\}$ large enough such that for any $k\geq K_1$, we also have 
      \begin{equation}
                (\ln k)^{16C_0}\leq k^{\varepsilon}.      \end{equation}   
    By Proposition \ref{Prop:AsyFormula} and Corollary \ref{Corspecmeas}, for any $ k\in [e^{\frac{\gamma_0}{5}n},e^{\frac{\gamma_0}{5}N}]\subset [K_1,\infty)  $, we have  
                \begin{align}
                    \epsilon\Im M(E+\ii \epsilon)&\leq C\|P_{(k),+}^{-1}\|
                    \leq k^{-1+\varepsilon}, \label{eq:fir} \\
                     \epsilon(k)^{-2}:=\det P_{(k),+}&\leq k^{2\max\{2-\frac{\zeta_n}{\widehat{\eta}_n},1\}+\varepsilon}, \label{eq:sec}
                \end{align}
                where the inequality in  \eqref{eq:sec} follows by \eqref{eta_s}.

    Note that by \eqref{noH3:s_0prime}, \eqref{noH3:s_0prime2} and \eqref{Delta-eta}, we have \begin{eqnarray}\label{noH3:D}
    |\eta_n -\widehat{\eta}_n |\leq |\eta_n n-\widehat{\eta}_n n|\leq C e^{-\frac{1}{4}\gamma_0 N}\leq \frac{\e^2}{100}.
    \end{eqnarray}
    Then by $(\mathbf{H2})$, \eqref{noH3:D} and \eqref{noH3:s_0prime}, if  $h\leq \delta<\infty$, we have
                    \begin{equation}\label{noH3:E}
                     \epsilon(k)^{-2}=\det P_{(k),+}\leq k^{2\max\{2-\frac{h}{\eta_n},1\}+C\varepsilon}\leq  k^{4-2\frac{h}{\delta}+C\varepsilon}.
                \end{equation}
                Combines \eqref{upperb-mu},\eqref{eq:fir},\eqref{noH3:E}, we have for $\epsilon=\epsilon(k)$, 
                \begin{equation}\label{upu}
                     \mu(E- \epsilon, E+ \epsilon)\leq 2\epsilon\Im M(E+\ii \epsilon)\leq \epsilon^{\frac{\delta}{2\delta-h}-C\varepsilon}. 
                \end{equation}

                 Since for any bounded potential and any solution $ u $ we have $ \|u\|_{L+1}\leq C\|u\|_L $, thus by (\ref{detP_k}), $ \epsilon(k+1)\geq c \epsilon(k) $. Meanwhile as $ \frac{1}{\epsilon}\Im M(E+\ii\epsilon) $ is monotonic in $ \epsilon $, it follows that 
                 \eqref{upu} holds for any $\epsilon$ with 
                 $\epsilon\leq \epsilon_*=\epsilon_*(C_0',C_0,\gamma_0,(n_s),\varepsilon)>0 $.
            \end{proof}

\begin{remark}
We make two short comments.
\begin{enumerate}
    \item $(\mathbf{H2})$ condition is non-trivial, it was 
was first observed in \cite{LYZZ} as the main input to prove exponential decay of the spectral gaps.
    \item Here we only need to consider the case  $\delta(\alpha,2\rho(\alpha,S_E^v)):=\delta\geq h$. Otherwise, if $\delta<h$, the Schr\"odinger cocycle $(\alpha,S_E^v)$ is reducible, which implies the spectral measure is Lipschitz, as was shown in Proposition \ref{Prop:finite-resonnace}.
\end{enumerate}
   
\end{remark}
        
\section{Exact local  distribution of the spectral measure}\label{Sec5}

In the previous section, Proposition \ref{Prop:AsyFormula} show that if the Schr\"odinger cocycle  $ (\alpha,S_E^v) $  is $ (C_0',C_0,\gamma_0,(n_s)_{s\geq s_0}) $-good, with 
 \[
        \begin{aligned}
           \Phi_n^{-1}(\theta+\alpha)S_E^v (\theta)\Phi_n(\theta)&=\begin{pmatrix}
            e^{2\pi \ii\rho_n}&\nu_n\\
            0&e^{-2\pi \ii\rho_n}
        \end{pmatrix}+H_n(\theta),
        \end{aligned}
    \]
then   one can obtain  the  asymptotic formula of $P_{(k),\pm}$.  We should remark that $(\mathbf{H3})$-type condition was first observed for AMO, which plays an essential role in proving non-critical Dry Ten Martini Problem \cite{avila2016dry}.

In this section, we will show that if we have further information on the exponential decay rate of $\nu_n$ (c.f. $ \mathbf{(H2)}$, $ \mathbf{(H3)}$ conditions of Proposition  \ref{THM5.1}), one can obtain exact local distribution of  the spectral measure. In the rest of the paper,  we will always fix the frequency $\alpha$, the energy $E$, and the analytic band $h$ of the potential $v$, thus usually omit the dependence on these parameters. 
Recall that   \[
        -\frac{\ln \|2\rho_n\|_\T}{n}=\widehat{\eta}_n \in (0,\infty], \qquad
            -\frac{\ln |\nu_n|}{n}= \zeta_n\in (0,\infty],
    \] 
    and
 \[
        -\frac{\ln   \|2\rho(\alpha,S_E^v)-\langle n_s,\alpha\rangle\|_\T
        }{n}=\eta_n \in (0,\infty], 
    \]  
    then we have the following criterion:

\begin{proposition}\label{THM5.1}
Let $\alpha\in \mathrm{DC}_d $, $v\in C^{\omega}(\T^d,\R)$, $E\in\Sigma_{v,\alpha}$. Given $0<\e\ll \min\{1,h\} $. If there exists $s_0=s_0(\varepsilon)$ large enough, such that:
\\ 
$ \mathbf{(H1)}$:  $ (\alpha,S_E^v) $ is $ (C_0',C_0, \gamma_0, (n_s)_{s\geq s_0}) $-good, with $ \gamma_0<h<\infty $; \\
$ \mathbf{(H2)}$:   $\zeta_n\in [(1-\varepsilon)h,\infty]$; \\
$ \mathbf{(H3)}$:   $ \zeta_n\leq (1+\varepsilon)h, \text{\ if \ }\widehat{\eta}_n \geq (1+\varepsilon)h$.\\
Then, there exist $ \epsilon_*=\epsilon_*(C_0', C_0, \gamma_0,  (n_s), \varepsilon) >0 $,  $ \widetilde{C}>0 $,  and a continuous function \[
f(\epsilon): (0, \epsilon_*]\to \left(1/2,1\right],
\]
such that
            \[
                \epsilon^{ f(\epsilon)+\widetilde{C}\varepsilon}\leq \mu(E- \epsilon, E+ \epsilon)\leq 
\epsilon^{ f(\epsilon)-\widetilde{C}\varepsilon}.
            \]    
            More precisely,  
            $ f(\epsilon) $ takes the following form:
            \begin{enumerate}
                \item If $ \eta_{n}\in[0,h] $, then $ f(\epsilon)=1 $;
                \item If $ \eta_{n}\geq h $, then  
            \end{enumerate}
            \begin{equation} \label{f(epsilon)}
                f(\epsilon)=\left\lbrace\begin{aligned}
                    &\frac{1}{2}+\frac{hn}{2\ln\epsilon^{-1}} &  &\text{for } \epsilon^{-1}\in[e^{hn}, e^{(2\eta_{n}-h)n}],\\
                    &\frac{1}{1-b_n}-\frac{b_n}{1-b_n} 
                    \frac{hN}{\ln\epsilon^{-1}} & &\text{for } \epsilon^{-1}\in[e^{(2\eta_{n}-h)n},e^{hN}],
                \end{aligned}\right.
            \end{equation} 
    where $ b_n=\frac{(\eta_{n}-h)n}{hN-\eta_{n} n} $.
        \end{proposition}
        \begin{remark}\label{Rem5.1}
            If $ N=\infty $ and $ h < \eta_n<\infty $, we have 
            \[
                f(\epsilon)=\left\lbrace\begin{aligned}
                    &\frac{1}{2}+\frac{hn}{2\ln\epsilon^{-1}} &  &\text{for } \epsilon^{-1}\in[e^{hn}, e^{(2\eta_{n}-h)n}],\\
                    &1-\frac{(\eta_n-h)n}{\ln\epsilon^{-1}} & &\text{for } \epsilon^{-1}\in[e^{(2\eta_{n}-h)n},\infty).
                \end{aligned}\right.
            \] 
        \end{remark}

        By assumption, the Schr\"odinger cocycle
        $ (\alpha,S_E^v) $ is $ (C_0',C_0,\gamma_0,(n_s)_{s\geq s_0}) $-good, then by Proposition \ref{Prop:AsyFormula}, there exist $ K_0=K_0(C_0',C_0,\gamma_0) $, such that  $ \det P_{(k),\pm} $ and $ \|P_{(k),\pm}\|^{-1} $ have well asymptotic estimates for 
         $k\geq \max \set{ e^{\frac{\gamma_0}{5}|n_{s_0}|},K_0 }.$  
 Moreover, by \eqref{cond-repulsion}, one can select $s_1\geq s_0$ to be the minimal integer such that
  \begin{equation} \label{s_0prime}
            \max\left\{\frac{4hn}{\gamma_0 N}, 32C_0\frac{\ln \frac{\gamma_0}{5} N}{\gamma_0 N}, C e^{-\frac{1}{3}\gamma_0N}\right\}\leq  \frac{\varepsilon^2}{100},
        \end{equation} 
 and  select $ K_1 \geq \max\{K_0,e^{\frac{\gamma_0}{5}|n_{s_1}|}\}$ large enough such that for any $k\geq K_1$, we also have 
      \begin{equation}
                (\ln k)^{16C_0}\leq k^{\varepsilon}.       \end{equation}

 The proof consists of three key steps, with the first being to derive a universal asymptotic estimate for $\det P_{(k),\pm}$, as detailed in Section \ref{5.1}. This universality refers to the fractal nature of the asymptotic formula, which remains consistent around any resonances. Leveraging Corollary \ref{Corspecmeas}, we can then deduce the universal asymptotic structure of $ \epsilon\Im M(E+\ii \epsilon)$  in Section \ref{5.2}. Ultimately, these results lead to the exact local distribution of the spectral measure, as presented in Section \ref{5.3}.

  \subsection{Universal asymptotic structure of $ \det P_{(k),\pm}$} \label{5.1}
  As we already explore from Proposition \ref{Prop:AsyFormula}, 
  $\det P_{(k),\pm}$ has a asymptotic structure, however the asymptotic  formula $\eta_+^n(k)$ is of local nature, i.e, it depends on $\zeta_n$, $\widehat{\eta}_n$ and $\gamma_0$ of the iteration scheme (ususally the scheme is KAM based), thus to obtain universal asymptotic structure of $ \det P_{(k),\pm}$, the key is to remove these dependence. 
  Before giving the proof, we first state the following key observation, which says that if $(\mathbf{H3})$ holds, then we can control the decay rate of $\rho_n$:

        \begin{lemma}\label{lem:possible}
    Let $ s\geq s_1 ,$  then we have 
    \begin{eqnarray}
    |\eta_n -\widehat{\eta}_n|\ 
    &\leq&  \frac{\e^2}{100}, \label{eq:posiible2}\\
     \max\{\eta_n n,\,|\eta_n n-h n|\} &\leq& \frac{\e^2}{50}\gamma_0 N. \label{eq:eta_n-hn}
    \end{eqnarray}
        \end{lemma}
        \begin{proof}
We first prove  \begin{eqnarray}\label{eq:hat-eta_n}
        \widehat{\eta}_n&\leq& \max\{4h,32C_0\frac{\ln\frac{\gamma_0}{5}N}{n}\},\label{eq:possible1}
   \end{eqnarray}
and only need to consider the case  $ N=|n_{s+1}|<\infty $.

Let $N_+=|n_{s+2}|$.   Notice that for any $e^{ \frac{\gamma_0}{5} n}\leq k\leq e^{\frac{\gamma_0}{5} N}$, we have
        \[
            (\ln k)^{- 16C_0}k^{ \eta_+^n(k)  }\leq  \det P_{(k),\pm}\leq (\ln k)^{ 16C_0 }k^{ \eta_+^n(k)  },
        \] 
        and for any $e^{ \frac{\gamma_0}{5} N}\leq k\leq e^{\frac{\gamma_0}{5} N_+}$, we have 
        \[
            (\ln k)^{- 16C_0}k^{ \eta_+^{N}(k)  }\leq  \det P_{(k),\pm}\leq (\ln k)^{ 16C_0 }k^{ \eta_+^{N}(k)  } ,
        \] 
           where we recall that
       \begin{equation*}
            \eta^n(k)=\left\lbrace\begin{aligned}
                &4-\frac{2\zeta_n n}{\ln k}&\text{if }e^{\frac{\gamma_0}{5} n}\leq k\leq e^{\widehat{\eta}_n n},\\
                &2+ \frac{2\widehat{\eta}_n n-2\zeta_n n}{\ln k} &\text{if }e^{\widehat{\eta}_n n}\leq k\leq e^{\frac{\gamma_0}{5} N}.
            \end{aligned}\right.
        \end{equation*} 
  In particular, for  $ k_0=e^{\tfrac{\gamma_0}{5}N} $, we have
            \[
                (\ln k_0)^{-16C_0}k_0^{\eta^n_+(k_0)}\leq  \det P_{(k_0),+}\leq (\ln k_0)^{ 16C_0 }k_0^{\eta^{N}_+(k_0)}.
            \] 
 
      By the assumption $(\mathbf{H2})$, $\zeta_{N}\in [(1-\varepsilon)h,\infty],$  $ \gamma_0<h$, it follows that     
    $ \eta^{N}_+(k_0)=2 $. 
 Consequently, one can estimate
         \begin{equation}\label{k1}
                0\leq\eta^n_+(k_0)-2\leq 32C_0\frac{\ln\ln k_0}{\ln k_0}\leq 32C_0\frac{\ln \frac{\gamma_0}{5}N}{\frac{\gamma_0}{5}N} \leq\frac{\e^2}{100}\ll 1. 
            \end{equation}
   \smallskip
\textbf{Case 1:   $ \widehat{\eta}_nn\geq \frac{\gamma_0}{5}N $.} \eqref{s_0prime} implies that  $\widehat{\eta}_n \geq (1+\varepsilon)h$, then by the assumption  $(\mathbf{H3})$, we have $\zeta_n \leq (1+\varepsilon)h$, consequently
             by the definition of $ \eta^n_+(k)$, we have 
            \[
                \eta^n_+(k_0)-2\geq 2- \frac{2\zeta_n n}{\tfrac{\gamma_0}{5}N}\geq 2-\frac{2(1+\varepsilon)hn}{\tfrac{\gamma_0}{5}N}\geq 2-\frac{\varepsilon^2}{20}\geq 1,
            \] 
            which  contradicts with \eqref{k1}.\\
    \smallskip
\textbf{Case 2:     
         $ 4h n<\widehat{\eta}_nn<\frac{\gamma_0}{5}N $}.  Again by the assumption $ \mathbf{(H3)}$, we have $\zeta_n \leq (1+\varepsilon)h$, consequently
            \[
                \eta^n_+(k_0)-2\geq  \frac{2\widehat{\eta}_{n} n-2\zeta_{n} n}{\tfrac{\gamma_0}{5}N}\geq \frac{\widehat{\eta}_n n}{\tfrac{\gamma_0}{5}N}.
            \] 
 Furthermore by \eqref{k1}, we have
            \[
                \widehat{\eta}_n n\leq 32C_0\ln\frac{\gamma_0}{5}N,
            \] 
            this gives \eqref{eq:hat-eta_n}.
            
            Note that by  \eqref{eq:hat-eta_n}, \eqref{Delta-eta} and \eqref{s_0prime}, we have 
    \begin{eqnarray*}
    |\eta_n n-\widehat{\eta}_n n|\leq C e^{-\frac{1}{3}\gamma_0 N}&\leq& \frac{\e^2}{100},\\
     \max\{\widehat{\eta}_n n,\,|\widehat{\eta}_n n-hn|\}&\leq& \frac{\e^2}{100}\gamma_0N,
    \end{eqnarray*}
      and then \eqref{eq:posiible2},\eqref{eq:eta_n-hn} follows immediately. 
        \end{proof}

In a word, Lemma \ref{lem:possible} shows that one can give an upper bound of  $\eta_n  $. Then we can distinguish the proof into two steps, the first step is to get rid of   $\zeta_n$ and $\widehat{\eta}_n $ in the asymptotic formula of $\det P_{(k),\pm}$.

\subsubsection{Step 1: Get rid of  $\zeta_n$ and $\widehat{\eta}_n$:} 
    For any $k\geq K_1 $,  we define $ \tilde{\eta}(k)$ as below:

        \begin{enumerate}
            \item If $ \eta_{n}\leq h $. Let $ \tilde{\eta}(k)=2 $ for $ e^{\frac{\gamma_0}{5} n}\leq k\leq e^{\frac{\gamma_0}{5} N} $. 
            \item If $ \eta_{n}>h $. Let
            $$
                \tilde{\eta}(k)=\left\lbrace\begin{aligned}
                    &4-\frac{2h n}{\ln k} &  &\text{for } e^{\frac{\gamma_0}{5} n} \leq k\leq e^{\eta_{n}n},\\
                    &2+\frac{2\eta_{n}n-2h n}{\ln k} & &\text{for } e^{\eta_{n}n}\leq k<e^{\frac{\gamma_0}{5} N}.
                \end{aligned}\right.
            $$ 
        \end{enumerate}
        Denote $ \max\{\tilde{\eta}(k),2\} $ by $ \tilde{\eta}_+(k) $ for short. Then we have

    \begin{lemma} 
For any $ k\geq K_1 $,   we have
        $$
        (\ln k)^{-16C_0}k^{\tilde{\eta}_+(k)}k^{-6\varepsilon}\leq \det P_{(k),\pm}\leq (\ln k)^{16C_0}k^{\tilde{\eta}_+(k)}k^{6\varepsilon}. 
        $$ 
    \end{lemma}
    \begin{proof}
 We distinguish the proof into three cases:\\
\smallskip
\textbf{Case 1: $ \widehat{\eta}_n>(1+\varepsilon)h $.} Then by \eqref{eq:posiible2} of Lemma \ref{lem:possible}, $\eta_{n}>(1+\frac{\e}{2})h$. By assumption  $ \mathbf{(H3)}$, we have $ \zeta_n\in[(1-\varepsilon)h,(1+\varepsilon)h] $. Let $ a_n=\min\{\widehat{\eta}_n,\eta_{n}\} $, $ b_n=\max\{\widehat{\eta}_n,\eta_{n}\} $,  consequently by Proposition \ref{Prop:AsyFormula} and \eqref{eq:posiible2}, we have the following: 
     \begin{enumerate}
        \item
    If $ \frac{\gamma_0}{5} n \leq \ln k\leq (1-\varepsilon)hn $, we have 
    \[
           \eta^n_+(k)=\max\{4-\frac{2\zeta_n n}{\ln k},2\}= 2= \max\{4-\frac{2h n}{\ln k},2\} =\tilde{\eta}_+(k)   
        \]  
\item If $ (1-\varepsilon)hn\leq \ln k\leq a_n n $, we have
    \[
        |\tilde{\eta}_+(k)-\eta^n_+(k)|\leq \left|4-\frac{2h n}{\ln k}-(4-\frac{2\zeta_n n}{\ln k})\right|\leq \frac{2\varepsilon}{1-\varepsilon}.
    \] 
\item If $a_n n\leq \ln k\leq b_n n $, we have either
    \[
        |\tilde{\eta}_+(k)-\eta^n_+(k)|\leq \left|2+\frac{2\eta_{n}n-2h n}{\ln k}-(4-\frac{2\zeta_n n}{\ln k})\right|\leq \frac{3\varepsilon}{1-\varepsilon}.
    \] 
or 
    \[
        |\tilde{\eta}_+(k)-\eta^n_+(k)|\leq \left|4-\frac{2h n}{\ln k}-( 2+ \frac{2\widehat{\eta}_n n-2\zeta_n n}{\ln k} )\right|\leq \frac{3\varepsilon}{1-\varepsilon}.
    \] 

 \item   If $ b_n n\leq \ln k\leq \frac{\gamma_0}{5}N $,
    we have 
    \[
        |\tilde{\eta}_+(k)-\eta^n_+(k)|\leq \left|2+\frac{2\eta_{n}n-2h n}{\ln k}-(2+ \frac{2\widehat{\eta}_n n-2\zeta_n n}{\ln k}) \right|\leq \frac{2\varepsilon}{1+\varepsilon}.
    \] 
            \end{enumerate}
Combine these sub-cases, it follows that
    $$
    (\ln k)^{-16C_0}k^{\tilde{\eta}_+(k)}k^{-\frac{4\varepsilon}{1-\varepsilon}}\leq \det P_{(k),\pm}\leq (\ln k)^{16C_0}k^{\tilde{\eta}_+(k)}k^{\frac{4\varepsilon}{1-\varepsilon}}.         $$ 
\smallskip
\textbf{Case 2: $ \widehat{\eta}_n\in[(1-\varepsilon)h,(1+\varepsilon)h] $.} By \eqref{eq:posiible2}, $\eta_{n}\in [(1-\frac{3\e}{2})h,(1+\frac{3\e}{2})h]$. By Proposition \ref{Prop:AsyFormula},  for any $ e^{\frac{\gamma_0}{5} n}\leq k\leq e^{\frac{\gamma_0}{5} N} $, we have
    \[
        2\leq \eta^n_+(k)\leq \max\{4-\frac{2\zeta_n n}{\widehat{\eta}_n n},2\}\leq \max\{4-\frac{2(1-\varepsilon)}{1+\varepsilon},2\}\leq 2+\frac{4\varepsilon}{1+\varepsilon}.
    \] 
    It follows that
        \[
            (\ln k)^{-16C_0}k^{2} \leq \det P_{(k),\pm}\leq (\ln k)^{16C_0}k^{2+\frac{4\varepsilon}{1+\varepsilon}}.          
        \] 
We further distinguish  two sub-cases:             
    \begin{enumerate}
        \item  If $ \eta_{n}\in[(1-\frac{3\e}{2})h,h] $, then $ \tilde{\eta}_+(k)= 2 $.
        \item  If $ \eta_{n}\in[h,(1+\frac{3\e}{2})h] $,we have 
        \[
            2\leq \tilde{\eta}_+(k)\leq \max\{4-\frac{2h n}{\eta_{n} n},2\}\leq \max\{4-\frac{2}{1+2\varepsilon},2\}\leq 2+\frac{4\varepsilon}{1+2\varepsilon }.
        \] 
                        \end{enumerate}
        Combines two sub-cases, we have
        $$
        (\ln k)^{-16C_0}k^{\tilde{\eta}_+(k)}k^{-\frac{4\varepsilon}{1+2\varepsilon }}\leq \det P_{(k),\pm}\leq (\ln k)^{16C_0}k^{\tilde{\eta}_+(k)}k^{\frac{4\varepsilon}{1+\varepsilon}}.
        $$ 
\smallskip
\textbf{Case 3:  $ \widehat{\eta}_n\in[0,(1-\varepsilon)h) $.} By \eqref{eq:posiible2}, $\eta_{n}\in [0,(1-\frac{\e}{2})h)$. Again by Proposition \ref{Prop:AsyFormula}, we have 
        \[
            (\ln k)^{-16C_0}k^{2} \leq \det P_{(k),\pm}\leq (\ln k)^{16C_0}k^{2}, \text{ for } e^{\frac{\gamma_0}{5}  n}\leq k\leq e^{\frac{\gamma_0}{5}  N},
        \] 
    Summarize all cases, the result follows.
\end{proof}

\subsubsection{Step 2: Get rid of $ \gamma_0 $ and smoothing argument:} In the first step, we are able to remove the dependence on local parameters $\zeta_n$ and $\widehat{\eta}_n$, however as we can see, if one extends $ \tilde{\eta}(k)$ to a function on $[K_1,\infty)$, it is not a continuous function anymore. The 
 second step is to further modify $ \tilde{\eta}(k)$ to $\psi(k)$, such that it can extended to a continuous function $\psi(x)$ on $[K_2,\infty)$.  Indeed,   we  can define $\psi(x)  $ as follows:

\begin{enumerate}
    \item If $ \eta_n\leq h $. Let $ \psi (x)=2 $ for $ x\in[e^{hn},e^{hN}].$
    \item If $ \eta_n>h $. Let
     \begin{equation}\label{psi(x)} 
\psi (x)=\left\lbrace\begin{aligned}
    &4-\frac{2h n}{\ln x} &  &\text{for } x\in[e^{hn}, e^{\eta_{n}n}],\\
    &2+\frac{2\eta_nn-2h n}{\ln x}\left[1-\frac{\ln x-\eta_n n}{hN-\eta_n n}\right] & &\text{for } x\in[e^{\eta_{n}n},e^{hN}].
\end{aligned}\right.
\end{equation}         \end{enumerate}
Just bear in mind, we will use the notations $k\in \Z$ and $x\in \R$.
Then we have the following:

\begin{lemma}\label{doas}
    For any $ k\geq K_2:=\max\{ e^{h|n_{s_1}|},K_1\}\geq K_1 $, we have
$$
k^{\psi(k)}k^{-C\varepsilon}\leq \det P_{(k),\pm}\leq k^{\psi(k)}k^{C\varepsilon}, 
$$
\begin{equation*} 
k^{-\varepsilon} k\leq \|P_{(k),\pm}^{-1}\|^{-1}\leq k^{\varepsilon}k.
\end{equation*} 
\end{lemma}
\begin{proof}
To see this, we distinguish into the following cases:\\
\smallskip
\textbf{Case 1:   $ \eta_n\leq h $.} Then  we have the following: \begin{enumerate}
    \item If $ hn\leq \ln k\leq \frac{\gamma_0}{5} N $, we have $ \psi(k)=\tilde{\eta}_+(k)=2 $. 
    \item If $ \frac{\gamma_0}{5} N\leq \ln k\leq hN $, if $ \eta_N\leq h $, we have $ \tilde{\eta}(k)=2$, if $ \eta_N>h $, we have  $\tilde{\eta}(k)=4-\frac{2hN}{\ln k}\leq 2$, thus $\tilde{\eta}_+(k)=\psi(k)=2 $.
\end{enumerate}
\smallskip
\textbf{Case 2:   $ \eta_n>h $.} Then we have the following:
\begin{enumerate}
\item If $ hn\leq \ln k\leq \eta_n n $, we have $ \psi(k)=\tilde{\eta}_+(k) $;
\item If  $ \eta_n n\leq \ln k\leq \frac{\gamma_0}{5}N $.
By \eqref{eq:eta_n-hn} of Lemma \ref{lem:possible} and  $ \gamma_0< h $,  we have 
\[
    |\psi(k)-\tilde{\eta}_+(k)|=\frac{2(\eta_n-h)n}{hN-\eta_nn}\frac{\ln k-\eta_n n}{\ln k}\leq\frac{2(\eta_n-h)n}{hN-\eta_nn} \leq \frac{\varepsilon\frac{\gamma_0}{25}N}{\frac{\gamma_0}{5}N}\leq\varepsilon.
\] 
\item If $ \frac{\gamma_0}{5}N\leq \ln k\leq h N $, similar  as in (2) of Case 1, if $ \eta_N\leq h $, we have $ \tilde{\eta}(k)=2$, if $ \eta_N>h $, we have  $\tilde{\eta}(k)=4-\frac{2hN}{\ln k}\leq 2$, thus $\tilde{\eta}_+(k)=2$, consequently by \eqref{eq:eta_n-hn} of Lemma \ref{lem:possible}  one can estimate 
\[
    |\psi(k)-\tilde{\eta}_+(k)|=\left|2+\frac{2\eta_nn-2h n}{\ln k}\left[1-\frac{\ln k-\eta_n n}{hN-\eta_n n}\right] -2\right|\leq \frac{2(\eta_n-h)n}{\frac{\gamma_0}{5}N} \leq \frac{\varepsilon\frac{\gamma_0}{25}N}{\frac{\gamma_0}{5}N}\leq\varepsilon.
\] 
\end{enumerate}
The result then follows. 
\end{proof}

\subsection{Universal asymptotic structure of the $M$-function}  \label{5.2}   
Once we  give the  universal asymptotic structure structure of  $ P_{(k),\pm}$,  one can  obtain the universal asymptotic structure  of  $M$-function. The link is given by Corollary \ref{Corspecmeas}, 
and the key is the selection $ \det P_{(k),+}=\frac{1}{\epsilon^2},$ which defines the function $\epsilon=\epsilon(k)$. However,
due to the complicated definition of $\det P_{(k),+}$, this functional relationship is not clear and not computable. The basic observation here is that  from Lemma \ref{doas}, the dominated terms of $\det P_{(k),+}$ is $k^{\psi(k)}$,
and our first step is to study the function $\psi(x)$. 

\subsubsection{Basic property of $\psi(x)$:}
\begin{lemma}\label{strictincrease}
Let $x\geq K_2 $, $ y=x^{1+\delta} $ for some $ \delta>0 $, then  
\begin{equation}\label{del-sti}
   \begin{aligned}
    \Delta(x,y)&:=     (1+\delta)\psi(y)-\psi(x)>2\delta/3>0.\\
\nabla(x,y)&:=|\psi(y)-\psi(x)|< 9\delta.
\end{aligned} 
\end{equation}
Moreover, if $ \delta\leq 3/2 $, then 
\begin{equation} \label{deltaless1.5}
    \Delta(x,y)\geq (2-C\varepsilon^2)\delta.
\end{equation}  
As a consequence, we have
 \begin{enumerate}
    \item[S.1] For $ x\geq K_2 $, $ x^{\psi(x)} $ is a strictly increasing continuous function.
    \item[S.2] For $ k\geq K_2 $, $ (k+1)^{\psi(k+1)}-k^{\psi(k)}\leq C k^{\psi(k)-1} $.
 \end{enumerate}

\end{lemma}  
\begin{proof}

First notice that $ \psi(x)\in [2,4) $, which implies  for any $ \delta>3/2 $, 
    $
        \Delta(x,y)\geq 2\delta-2\geq 2\delta/3 >0, \nabla(x,y)\leq 2< 2\delta  $, 
    thus we only need to  consider $ \delta\leq 3/2 $.  In the following,  we always  fix  $ x\in[e^{hn},e^{hN}] $, so $ y\in [e^{(1+\delta)hn},e^{(1+\delta)hN}] $. \\
   \smallskip \textbf{Case 1: $ \eta_n\leq h $.} In this case,  
   $\psi(x)\equiv 2 $, which  immediately implies        
$ \Delta(x,y)\geq 2\delta$.  As for $ \nabla(x,y) $, we distinguish three subcases:\\
\smallskip \textbf{Case 1.1:} If $ y\in [e^{(1+\delta)hn},e^{hN}] $, then $ \nabla(x,y)=0 $.\\
Note if $y\in [e^{hN},e^{(1+\delta)hN}]$, then we only need to consider the case $\eta_N>h$, otherwise $ \nabla(x,y)=0 $. Then we can further discussed as:\\
\smallskip \textbf{Case 1.2:} If $ y\in [e^{hN},e^{\eta_NN}] $, then $ \nabla(x,y)=2-\frac{2hN}{\ln y}\leq 2-\frac{2hN}{(1+\delta)hN}\leq \frac{2\delta}{1+\delta} $. \\
\smallskip \textbf{Case 1.3:}
If $ y\in [e^{\eta_NN},e^{(1+\delta)hN}] $, then $ \nabla(x,y)\leq 2-\frac{2hN}{\eta_NN}\leq 2-\frac{2hN}{(1+\delta)hN}\leq \frac{2\delta}{1+\delta} $.\\
   \smallskip \textbf{Case 2: $ \eta_n> h $.} 
We distinguish the proof into the following cases: \\
\smallskip \textbf{Case 2.1:}
        If both $ x,y\in [e^{hn},e^{\eta_n n}] $ (this contains the case $ \eta_n=N=\infty $), then 
        \[
            \begin{aligned}
                \Delta(x,y)  &= 4(1+\delta)-\frac{2hn}{\ln x}-4+\frac{2hn}{\ln x}=4\delta.\\
            \nabla(x,y)  &=\left| 4-\frac{2hn}{(1+\delta)\ln x}-4+\frac{2hn}{\ln x}\right|\leq \frac{2\delta}{1+\delta}.
            \end{aligned}
        \] 
\smallskip \textbf{Case 2.2:}             If both $ x,y\in [e^{\eta_n n},e^{h N}] $, then
    \[
    \begin{aligned}
        \Delta(x,y)&=2\delta+\frac{2\eta_nn-2h n}{\ln x}-\frac{2\eta_nn-2h n}{\ln x}\left(\frac{(1+\delta)\ln x-\eta_n n}{hN-\eta_n n}\right)\\
                &\quad-\frac{2\eta_nn-2h n}{\ln x}+\frac{2\eta_nn-2h n}{\ln x}\left(\frac{\ln x-\eta_n n}{hN-\eta_n n}\right)\\
                &=2\delta-\frac{2\eta_nn-2h n}{ hN-\eta_n n}\delta\geq (2-C\varepsilon^2)\delta.\\
        \nabla(x,y)&=\left| \frac{2\eta_nn-2h n}{(1+\delta)\ln x}-\frac{2\eta_nn-2h n}{(1+\delta)\ln x}\left(\frac{(1+\delta)\ln x-\eta_n n}{hN-\eta_n n}\right)\right.\\
                    &\quad\left.-\frac{2\eta_nn-2h n}{\ln x}+\frac{2\eta_nn-2h n}{\ln x}\left(\frac{\ln x-\eta_n n}{hN-\eta_n n}\right)\right|\\
                    &\leq \frac{2\delta}{1+\delta}+C\varepsilon \delta.
    \end{aligned}
    \]
        where the inequalities holds  by \eqref{eq:eta_n-hn} of Lemma \ref{lem:possible} and $\gamma_0<h$.\\
\smallskip \textbf{Case 2.3:}              If $ x\in [e^{hn},e^{\eta_n n}] $, $ y\in [e^{\eta_n n},e^{hN}] $, then let $ x_0=e^{\eta_n n} $, one can write $ x_0=x^{1+\delta_1} $, and $ y=x_0^{1+\delta_2} $, and thus $ \delta_2=\frac{\delta-\delta_1}{1+\delta_1} $. Then by \textbf{Case 2.1} and \textbf{Case 2.2},  one can estimate:
   \[ 
\begin{aligned}
                \Delta(x,y)&=(1+\delta)\psi(y)-\psi(x)\\
                &=[1+\delta_1+(1+\delta_1)\delta_2]\psi(y)-\psi(x)\\
                &=(1+\delta_1)\psi(x_0)-\psi(x)
              +(1+\delta_1)(1+\delta_2)\psi(y)-(1+\delta_1)\psi(x_0)\\
                &\geq  (2-C\varepsilon^2)\delta_1+(1+\delta_1)(2-C\varepsilon^2)\delta_2=(2-C\varepsilon^2)\delta .\\
                \nabla(x,y)
                    &\leq |\psi(x_0)-\psi(x)|+|\psi(y)-\psi(x_0)|\\
                    &\leq \frac{2\delta_1}{1+\delta_1}+ \frac{2\delta_2}{1+\delta_2}+C\varepsilon \delta_1+C\varepsilon \delta_2\\
                    &\leq \frac{4\delta}{1+\delta}+2C\varepsilon \delta.
            \end{aligned}
        \] 
\smallskip \textbf{Case 2.4:}      If $ x\in [e^{\eta_n n},e^{h N}] $, $ y\in [e^{hN},e^{(1+\delta)hN}] $,  let $ x_0=e^{hN} $, write $ x_0=x^{1+\delta_1} $ and $ y=x_0^{1+\delta_2} $. By the discussions in \textbf{Case 1},  \textbf{Case 2.1} and \textbf{Case 2.3}, and following the  same proof in \textbf{Case 2.3},  we have
        \[
            \Delta(x,y)\geq (2-C\varepsilon^2)\delta , \ \nabla(x,y)\leq \frac{8\delta}{1+\delta}+4C\varepsilon \delta.
            \]

        Above discussions cover all the possible cases, this completes the proof of \eqref{del-sti}. Concerning $(S.1)$, the monotonicity of $ x^{\psi(x)} $ immediately follows from $  \Delta(x,y)>0$.  As for $(S.2)$,  write $ k+1=k^{1+\delta} $, or equivalently, 
         \[
            (1+\delta)\ln k=\ln(k+1).
        \] 
        Thus for $k$ large enough,  we have $|\psi(k+1)-\psi(k)|< 9\delta$, and 
        \begin{equation}\label{deltak-k+1}
            \delta=\frac{\ln(k+1)-\ln k}{\ln k}=O(\frac{1}{k\ln k}),
           \end{equation} 
  which gives us
        \[
            \begin{aligned}
                (k+1)^{\psi(k+1)}-k^{\psi(k)}  =  k^{\psi(k)}\pa{k^{\psi(k+1)-\psi(k)+\delta\psi(k+1)}-1}
                \leq C k^{\psi(k)-1}. 
            \end{aligned}\qedhere
        \] 

\end{proof}
   
   Indeed, by nice structure of $\psi(x)$, direct computation gives:
 \begin{lemma}\label{in1}
 The function $x=x(\epsilon)$ defined by 
 \begin{equation} \label{inverse}
    \epsilon= x^{-\frac{\psi(x)}{2}}:=\tilde{\epsilon}(x)
 \end{equation} 
   is well-defined on $[K_2,\infty)$, and  the function 
      $$ f(\epsilon)=\frac{2}{\psi(x)}= \frac{2}{\psi(x(\epsilon))} \in [1/2,1]   $$
      has the expression \eqref{f(epsilon)}.   Moreover,  for any $ \epsilon\in(0,\tilde{\epsilon}(K_2)] $ and $ 0\leq\tau<1 $, we have 
      \begin{equation} \label{error-f} f(\epsilon^{1+\tau})-f(\epsilon)\leq \frac{3}{5}\tau. 
            \end{equation} 
 \end{lemma}
\begin{proof} 
First by Lemma \ref{strictincrease}, $ x^{\psi(x)} $ is a strictly increasing continuous function, which implies that $ f(\epsilon) $ is well-defined on $ (0,\tilde{\epsilon}(K_2)] $. We distinguish the proof into several cases:\\
\smallskip
\textbf{Case 1:  $\eta_n\in[0,h]$.  } In this case, note $ \psi(x)=2 $, which gives $ \epsilon^{-1}(x)=x $ and $ f(\epsilon)=1 $.\\
\smallskip
\textbf{Case 2.1: $ h<\eta_n $, and $ x\in [e^{hn},e^{\eta_n n}] $}. In this case,  
\begin{equation}\label{case2.1} \psi(x)=4-\frac{2h n}{\ln x}.\end{equation}
By monotonicity of $ x^{\psi(x)} $,  $ \epsilon^{-1}\in[e^{hn},e^{(2\eta_n-h)n}] $.
Combining \eqref{inverse} and \eqref{case2.1}, we have
\[
 \frac{1}{f(\epsilon)} = \frac{\psi(x)}{2}= \frac{\ln{\epsilon^{-1}}}{\ln x}=2-\frac{hn}{\ln x},
\] 
which implies that
\begin{equation}\label{ca2}
    \begin{aligned}
        \ln x&=\frac{\ln \epsilon^{-1}+hn}{2},\\
        f(\epsilon)&=\frac{1}{2}+\frac{hn}{2\ln\epsilon^{-1}}.
    \end{aligned}
\end{equation}
\smallskip
\textbf{Case 2.2: $ h<\eta_n<\infty $, and  $ x\in [e^{\eta_n n},e^{hN}] $.}   In this case,  
\begin{equation}\label{case2.2} \psi(x)=2+\frac{2\eta_nn-2h n}{\ln x}\left[1-\frac{\ln x-\eta_n n}{hN-\eta_n n}\right].\end{equation}
Again by monotonicity of $ x^{\psi(x)} $, $ \epsilon^{-1}\in[e^{(2\eta_n-h)n},e^{hN}] $.  Combining \eqref{inverse} and \eqref{case2.2}, we have
\[
 \frac{1}{f(\epsilon)} = \frac{\psi(x)}{2}= \frac{\ln{\epsilon^{-1}}}{\ln x}=1+\frac{(\eta_n-h)n}{\ln x}\left[1-\frac{\ln x-\eta_n n}{hN-\eta_n n}\right],
\] 
which implies that
\[
    \begin{aligned}
        \ln x&=\frac{\ln \epsilon^{-1}}{1-b_n}-\frac{b_n}{1-b_n}hN,\\
        f(\epsilon)&=
        \frac{1}{1-b_n}-\frac{b_n}{1-b_n}\frac{hN}{\ln\epsilon^{-1}},
    \end{aligned}
\] 
where $ b_n=\frac{(\eta_n-h)n}{hN-\eta_n n} $.

  To prove \eqref{error-f}, we always  fix  $ \epsilon^{-1}\in[e^{hn},e^{hN}] $, which gives $ \epsilon^{-(1+\tau)}\in [e^{(1+\tau)hn},e^{(1+\tau)hN}] $. 
  Just note in the Case 1,  $ f(\epsilon)=1$,  then $ f(\epsilon^{1+\tau})\leq f(\epsilon).$ Thus we need to consider the above Case 2, the basic observation is the following estimate:
  \begin{equation}\label{esti2}
        \frac{b_n}{1-b_n}\frac{hN}{(2\eta_n-h)n} \leq \frac{3}{5}.
    \end{equation}
  Indeed, if $ N<\infty $, by \eqref{eq:eta_n-hn} of Lemma \ref{lem:possible},  there holds  that $ \eta_n n\leq \frac{1}{100}hN  $, and thus
    \[
        \frac{b_n}{1-b_n}\frac{hN}{(2\eta_n-h)n}=\frac{1}{1-\frac{(\eta_n-h)n}{hN-\eta_n n}}\frac{hN}{hN-\eta_n n}\frac{\eta_n-h}{2\eta_n-h} \leq \frac{3}{5}.
    \]  
For $ N=\infty $ and $ h<\eta_n<\infty $, by Remark \ref{Rem5.1}, we actually  have 
    \[
        \frac{b_n}{1-b_n}\frac{hN}{(2\eta_n-h)n}=\frac{(\eta_n-h)n}{(2\eta_n-h)n}\leq \frac{1}{2}.
    \]

    Then we can estimate as follows: \\
\smallskip
\textbf{Case 2.1: $ h< \eta_n $, and  $ \epsilon^{-1}\in[e^{hn},e^{(2\eta_n-h)n}] $.} 
By \eqref{eq:eta_n-hn} of Lemma \ref{lem:possible}, $ (2\eta_n-h)n\leq \frac{1}{100} hN $, so for $ \tau<1 $, $ \epsilon^{-(1+\tau)} $ will not exceed $ e^{hN} $, thus we only need to consider the following two subcases:\\
If $ \epsilon^{-(1+\tau)}\in [e^{hn},e^{(2\eta_n-h)n}] $, by \eqref{ca2}, we have $ f(\epsilon^{1+\tau})\leq f(\epsilon) $; \\
If $ \epsilon^{-(1+\tau)}\in [e^{(2\eta_n-h)n},e^{hN}] $, rewrite $ \epsilon^{1+\tau}=e^{-(1+\tau_1)(2\eta_n-h)n} $, then $ \tau_1\leq \tau $, meanwhile by monotonicity and \eqref{esti2}, we have
\[
        f(\epsilon^{1+\tau})-f(\epsilon)\leq f(\epsilon^{1+\tau})-f(e^{-(2\eta_n-h)n})
         = \frac{b}{1-b}\frac{hN}{(2\eta_n-h)n}(1-\frac{1}{1+\tau_1})\leq  \frac{3}{5}\tau.
\] \\
\smallskip
\textbf{Case 2.2: $ h<\eta_n<\infty $, and  $ \epsilon^{-1}\in[e^{(2\eta_n-h)n},e^{hN}] $.} In this case, we have
\[
    f(\epsilon)=\frac{1}{1-b_n}-\frac{b_n}{1-b_n}\frac{hN}{\ln\epsilon^{-1}}.
\] 
If $ \epsilon^{-(1+\tau)}\in [e^{(2\eta_n-h)n},e^{hN}] $, again by monotonicity and \eqref{esti2}, we have
$$
        f(\epsilon^{1+\tau})-f(\epsilon)
        \leq \frac{b_n}{1-b_n}\frac{hN}{(2\eta_n-h)n}\frac{\tau}{1+\tau}\leq \frac{3}{5}\tau.
$$
If $ \epsilon^{-(1+\tau)}> e^{hN} $, write $ e^{-hN}=\epsilon^{1+\tau_1} $, so $ \tau_1<\tau $, and thus 
\[
    f(\epsilon^{1+\tau})-f(\epsilon)\leq f(e^{-hN})-f(\epsilon)\leq \frac{3}{5}\tau_1\leq \frac{3}{5}\tau.
\] 
We complete this case.  \qedhere

   \end{proof}

    As a consequence of Lemma \ref{strictincrease}, we  also have the following technical result, which will be used several times. 
\begin{lemma}\label{notexceed}
    For $ x\geq K_2 $. Let $ y=x^{1+\delta} $ with $ \delta\geq 0 $. Suppose that 
    \begin{equation} \label{L-k}
        y^{\psi(y)-C_1\varepsilon}\leq x^{\psi(x)+C_2\varepsilon},
    \end{equation} 
    then $ y \leq  x^{1+C'\varepsilon} $, where $ C'= \frac{3(C_1+C_2)}{2-3C_1\varepsilon}$.
\end{lemma}
\begin{proof}
    By  (\ref{L-k}) and  Lemma \ref{strictincrease}, we have
    \[
    2\delta/3 < \Delta(x,y)=(1+\delta)\psi(y)-\psi(x)\leq (C_1+C_1\delta+C_2)\varepsilon,
    \] 
this implies  $\delta \leq  \frac{3(C_1+C_2)}{2-3C_1\varepsilon}\varepsilon $. The result follows. 
\end{proof}

    \begin{lemma}\label{upperb}
        Let $ \epsilon(k) $ be such that $ \det P_{(k),+}=\frac{1}{\epsilon(k)^2} $, then there exists $K_3\geq K_2$ such that for all $ k\geq K_3 $, we have 
        \[
            \epsilon^{\frac{2}{\psi(k)}+C\varepsilon}\leq  \epsilon\Im M(E+\ii \epsilon) \leq \epsilon^{\frac{2}{\psi(k+1)}-C\varepsilon}, \text{ for }\epsilon\in [\epsilon(k+1),\epsilon(k)].
        \] 
    \end{lemma}
    \begin{proof}

     For any $ k\in[e^{hn},e^{hN}]\subset [K_3,\infty) $,  Proposition \ref{Prop:AsyFormula} and Lemma \ref{doas} gives
            \begin{align}
 \label{rela1}  
 k^{\psi(k)-C\varepsilon} &\leq
 \det P_{(k),\pm} 
 \leq  k^{\psi(k)+C\varepsilon},\\
 \label{rela2}       
 k^{-\varepsilon}k &\leq
 \|P_{(k),\pm}^{-1}\|^{-1}
 \leq k^{\varepsilon}k,\\
 \label{rela3}   
 k^{-\varepsilon}k &\leq
 \|u_{\beta_n}^\pm\|_{2k}^{2}
 \leq k^{\varepsilon}k.
            \end{align}  
            
For any $k\geq 1 $, we  select $ \epsilon(k)$ such that $ \det P_{(k),+}=\frac{1}{ \epsilon(k)^2} $,
  and then
  select  $ 2k_1 $ to be the closest even number of $ L^-( \epsilon(k)) $ with $ L^-( \epsilon(k))\leq 2k_1 $,
        where we recall that for any given $\epsilon$,  $L^{-}(\epsilon)$ is given by 
         $$ 
         ||| K^-(E)|||_{L^{-}}^2=\frac{1}{\epsilon^2}.$$
Take $K_3\geq K_2$ be minimal such that $k_1>K_2$, provided $k\geq K_3$. Now for any $k\geq K_3$,
      by \eqref{rela1},  we have 
           $$ (k_1-1)^{\psi(k_1-1)-C\varepsilon} \leq \frac{1}{ \epsilon(k)^2} \leq k^{\psi(k)+C\varepsilon},$$
 then Lemma \ref{notexceed} implies that
  \begin{equation} \label{rela4} k_1 \leq k^{1+C\varepsilon}. \end{equation}
  
 As a consequence of   \eqref{rela1}, \eqref{rela2}, \eqref{rela3} and \eqref{rela4}, we have 
  \begin{align*}
      \|P_{(k)}^{-1}\| &\leq k^{-(1-\varepsilon)}\leq  \epsilon(k)^{\frac{2}{\psi(k)+C\varepsilon}(1-\varepsilon)},\\
      \|u_{\beta_n}^+\|_{2k}^{-2} &\geq k^{-(1+\varepsilon)}\geq   \epsilon(k)^{\frac{2}{\psi(k)-C\varepsilon}(1+\varepsilon)},\\
      \|u_{\beta_n}^-\|_{L^-( \epsilon(k))}^{-2}&\geq \|u_{\beta_n}^-\|_{2k_1}^{-2}\geq  k_1^{-(1+\varepsilon)} \geq  \epsilon(k)^{\frac{2}{\psi(k)}+C\varepsilon}.
  \end{align*}
       By Corollary \ref{Corspecmeas}, one has 
        \[
            \epsilon(k) \Im M(E+\ii \epsilon(k))\leq \epsilon(k)^{\frac{2}{\psi(k)}-C\varepsilon},
        \] 
        \[
            \epsilon(k)\Im M(E+\ii \epsilon(k))\geq 
             \epsilon(k)^{\frac{2}{\psi(k)}+C\varepsilon}.
        \] 
        Following the same line as in the proof of Proposition \ref{THM5.4}, noting $ \epsilon(k+1)\geq c \epsilon(k) $, and $ \frac{1}{\epsilon}\Im M(E+\ii\epsilon) $ is monotonic in $ \epsilon $, it follows that for any $\epsilon\in[\epsilon(k+1),\epsilon(k)]$,
        \begin{align*}
            \epsilon\Im M(E+\ii \epsilon) &\leq C\epsilon(k+1) \Im M(E+\ii\epsilon(k+1))\leq C\epsilon(k+1)^{\frac{2}{\psi(k+1)}-C\varepsilon}\leq \epsilon^{\frac{2}{\psi(k+1)}-C\varepsilon},\\
            \epsilon\Im M(E+\ii \epsilon) &\geq c\epsilon(k) \Im M(E+\ii\epsilon(k))\geq c\epsilon(k)^{\frac{2}{\psi(k)}+C\varepsilon}\geq \epsilon^{\frac{2}{\psi(k)}+C\varepsilon}.\qedhere
        \end{align*}
    \end{proof}

    As we mentioned above, the  functional relation $  \epsilon(k)=(\det P_{(k),+})^{-2} $ is not computable,  however,  its  dominated term 
\eqref{inverse}  is computable (Lemma \ref{in1}). Then Proposition \ref{THM5.1} reduces to the following:
    
        \begin{lemma}\label{middle}
There exists $K_4\geq K_3$ such that for all  $ \epsilon\leq \tilde{\epsilon}(K_4) $, we have 
\[
    \epsilon^{f(\epsilon)+\widetilde{C}\varepsilon}\leq \epsilon\Im M(E+\ii \epsilon)\leq \epsilon^{f(\epsilon)-\widetilde{C}\varepsilon}.
\]  
  \end{lemma}
   \begin{proof} 
 Take $K_4\geq K_3$ be minimal such that for any  fixed $k\geq K_4$, one can select $k_1$ to be  
  the maximal integer  such that $K_3\leq  k_1\leq k $ and 
\begin{equation}\label{bridge}
    \det P_{(k_1),+}\leq k_1^{\psi(k_1)+C\varepsilon}\leq k^{\psi(k)}.
\end{equation}
Meanwhile, we select $k_2 \geq k+1  $ to be the minimal integer such that 
\begin{equation}\label{bbridge}
    (k+1)^{\psi(k+1)}\leq k_2^{\psi(k_2)-C\varepsilon}\leq \det P_{(k_2),+}.
\end{equation}

By our selection of $k_1$ (i.e. $ k^{\psi(k)}\leq (k_1+1)^{\psi(k_1+1)+C\varepsilon}$), Lemma \ref{notexceed} implies that 
\begin{equation}\label{bridge-1}
    k_1 \leq k\leq (k_1+1)^{1+C\varepsilon}.
\end{equation}
Same argument gives 
\begin{equation}\label{bridge-2}
    k\leq k_2-1\leq (k+1)^{1+C\varepsilon}.
\end{equation}
Thus for any $ x_*\in[k,k+1] $, $ y_*\in[k_1,k_2] $, write  $ y_*=x_*^{1+\delta} $, \eqref{bridge-1} and \eqref{bridge-2} imply that $ |\delta|<C\varepsilon $. Then by
 Lemma \ref{strictincrease},  we have
\begin{equation}\label{bridge-3}
|\psi(x_*)-\psi(y_*)|\leq C\varepsilon .
\end{equation}

On the other hand, \eqref{bridge} and \eqref{bbridge} imply
\[
    \epsilon(k_1)^{-2}\leq \tilde{\epsilon}(k)^{-2}\leq \tilde{\epsilon}(k+1)^{-2}\leq \epsilon(k_2)^{-2}.
\] 
By Lemma \ref{upperb},  for any $ \epsilon\in[\tilde{\epsilon}(k+1),\tilde{\epsilon}(k)] \subset  [\epsilon(k_2),\epsilon(k_1)] $, we have
\begin{equation}\label{bridge-4}
    \epsilon^{\frac{2}{\psi(y)}+C\varepsilon}\leq \epsilon\Im M(E+\ii \epsilon)\leq \epsilon^{\frac{2}{\psi(y+1)}-C\varepsilon}
\end{equation}
for some $y \in \Z$ with  $y\in[k_1,k_2-1] $. 
 For any $ \epsilon\in[\tilde{\epsilon}(k+1),\tilde{\epsilon}(k)]  $, by $(S.1)$ of Lemma \ref{strictincrease},  there exists a unique $ x\in[k,k+1] $ such that $ \epsilon=\tilde{\epsilon}(x) $.
By \eqref{bridge-3}, \eqref{bridge-4}, we have 
\[
    \epsilon^{\frac{2}{\psi(x)}+\widetilde{C}\varepsilon}\leq \epsilon\Im M(E+\ii \epsilon)\leq \epsilon^{\frac{2}{\psi(x)}-\widetilde{C}\varepsilon}.
\]
Once we have this, the result follows.
\end{proof}

\subsection{Exact local distribution of the spectral measure}  \label{5.3}

To estimate \( \mu(E-\epsilon, E+\epsilon) \), the traditional approach involves using its relationship with \( \epsilon \Im M(E + \ii \epsilon) \). The upper bound is given by the following:
\[
\epsilon \Im M(E + \ii \epsilon) = \int \frac{\epsilon^2}{(E' - E)^2 + \epsilon^2} \, \dif \mu(E') \gtrsim \mu(E - \epsilon, E + \epsilon).
\]
However, whether the inverse inequality holds, i.e., whether
\[
\epsilon \Im M(E + \ii \epsilon) \sim \mu(E - \epsilon, E + \epsilon),
\]
is a subtle question and still unknown. Indeed,  if we split the integral into \( I_1 = \int_{|E' - E| \leq \epsilon} \) and \( I_2 = \int_{|E' - E| \geq \epsilon} \), it is clear that
\[
I_1 \leq \mu(E - \epsilon, E + \epsilon),
\]
the main obstacle comes from \( I_2 \). Using the uniform estimate that $\mu$ is  \( \frac{1}{2} \)-Hölder \cite{AvilaHolderContinuityAbsolutely2011},  we get \( I_2= \int_{\epsilon^{4/5}\geq|E' - E| \geq \epsilon}+ \int_{|E' - E| \geq \epsilon^{4/5}} \lesssim \epsilon^{2/5} \), which is problematic.

The central argument revolves around Lemma \ref{middle}, which provides an asymptotic upper bound for $ \epsilon \Im M(E + \ii \epsilon) $, subsequently influencing the upper bound of the spectral measure. Consequently, this leads to a stronger estimate for $I_2$, specifically, $ I_2 \leq \epsilon^{f(\epsilon) - C\varepsilon} $. However, this is insufficient, as $ \epsilon \Im M(E + \ii \epsilon) $ is not the dominant term. To address this, we employ a technique involving a dislocation of $ \epsilon $, focusing on estimating $ \epsilon^{1+\tau} \Im M(E +\ii \epsilon^{1+\tau}) $.

      \begin{lemma}\label{mu2}
        For any $ \varepsilon>0 $, $ \tau\geq 0 $,  there exist $ \epsilon_*=\epsilon_*(\tilde{\epsilon}(K_4),\e)>0 $ such that if $ \epsilon\leq\epsilon_* $,  we have
        $$\mu(E-\epsilon,E+\epsilon)  \geq  \epsilon^{1+\tau}\Im M(E+\ii\epsilon^{1+\tau})-O(\epsilon^{f(\epsilon)+2\tau-3\widetilde{C}\varepsilon}).$$
    \end{lemma}
    \begin{proof}
        Take $ \epsilon_*:=\min\{\pa{\tilde{\epsilon}(K_4)}^3,(\widetilde{C}\varepsilon)^2, (2C)^{-1/\widetilde{C}\varepsilon}\} $,   then for any $\epsilon\leq \epsilon_*$, one can apply Lemma \ref{middle}, consequently by   \eqref{upperb-mu}, we have        
        \begin{equation}\label{uppmu}
            \mu(E-\epsilon,E+\epsilon)\leq \epsilon\Im M(E+\ii \epsilon) \leq \epsilon^{f(\epsilon)-\widetilde{C}\varepsilon}.
 \end{equation}

 On the other hand,  by \eqref{MU}, one has
       \[
        \epsilon^{1+\tau}\Im M(E+\ii\epsilon^{1+\tau})=\int\frac{\epsilon^{2+2\tau}}{(E'-E)^2+\epsilon^{2+2\tau}}\dif \mu(E'),
       \] 
       we  then split the integral into $ I_1=\int_{|E'-E|\leq \epsilon} $, $ I_2=\int_{\epsilon\leq |E'-E|\leq \epsilon^{1/3}} $, and $ I_3=\int_{|E'-E|\geq \epsilon^{1/3}} $. Clearly  we have 
       \(
        I_1\leq \mu(E-\epsilon,E+\epsilon),
       \) 
       and 
       \(
        I_3\leq  C\epsilon^{4/3+2\tau}, 
       \) 
       thus the key is to estimate $I_2$. To do this, by $(S.1)$ of Lemma \ref{strictincrease},
      we can let $ m$ be the integer such that 
       \[
        \begin{aligned}
     \tilde{\epsilon}(m)&\geq \epsilon^{1/3}\geq \tilde{\epsilon}(m+1). 
               \end{aligned}
       \]
Meanwhile,  there exist integer $n$ and  $x\in [n-1,n]$ such that  
\begin{equation}\label{epn}
    (n-1)^{-\frac{\psi(n-1)}{2}} =\tilde{\epsilon}(n-1)\geq \epsilon=\tilde{\epsilon}(x)\geq \tilde{\epsilon}(n)=n^{-\frac{\psi(n)}{2}}.
\end{equation}

Note that $(S.2)$ of Lemma \ref{strictincrease} gives 
\begin{equation}\label{diffk} (k+1)^{\psi(k+1)}-k^{\psi(k)}\leq C k^{\psi(k)-1} \leq C k^{\psi(k)} \end{equation}
for $ k $ sufficiently large, 
which implies that
\begin{equation}\label{m-m}
    \tilde{\epsilon}(m+1) \geq c \tilde{\epsilon}(m) \geq c\epsilon^{1/3}.
\end{equation}

       Then we can estimate 
       $$
       \begin{aligned}
        I_2&=\int_{\epsilon\leq |E'-E|\leq \epsilon^{1/3}}\frac{\epsilon^{2+2\tau}}{(E'-E)^2+\epsilon^{2+2\tau}}\dif\mu(E')\\
        &\leq\sum_{k=m+1}^{n}\int_{\tilde{\epsilon}(k)\leq |E'-E|\leq \tilde{\epsilon}(k-1) }\frac{\epsilon^{2+2\tau}}{\epsilon^{2+2\tau} +\tilde{\epsilon}(k)^2}\dif\mu(E')\\
        &\leq \frac{\epsilon^{2+2\tau}}{\epsilon^{2+2\tau}+\tilde{\epsilon}(m+1)^2}\mu(E-\tilde{\epsilon}(m),E+\tilde{\epsilon}(m))\\
        &\quad+\sum_{k=m+1}^{n-1}\left[\frac{\epsilon^{2+2\tau}}{\epsilon^{2+2\tau}+\tilde{\epsilon}(k+1)^2}-\frac{\epsilon^{2+2\tau}}{\epsilon^{2+2\tau}+\tilde{\epsilon}(k)^2}\right]\mu(E-\tilde{\epsilon}(k),E+\tilde{\epsilon}(k))\\
        &\quad-\frac{\epsilon^{2+2\tau}}{\epsilon^{2+2\tau}+\tilde{\epsilon}(n)^2}\mu(E-\tilde{\epsilon}(n),E+\tilde{\epsilon}(n))\\
        &\leq C \epsilon^{4/3+2\tau}+\epsilon^{2+2\tau}\sum_{k=m+1}^{n-1}\left[\frac{(k+1)^{\psi(k+1)}-k^{\psi(k)}}{(\epsilon^{2+2\tau}(k+1)^{\psi(k+1)}+1)(\epsilon^{2+2\tau}k^{\psi(k)}+1)}\right]k^{-\frac{\psi(k)}{2}(\frac{2}{\psi(k)}-\widetilde{C}\varepsilon)}\\
        &\leq C\epsilon^{4/3+2\tau}+ \sum_{k=m+1}^{n-1}C\epsilon^{2+2\tau}k^{\psi(k)-2+2\widetilde{C}\varepsilon}
      \end{aligned}$$
        where the third inequality follows by \eqref{uppmu},\eqref{m-m}, and the last inequality holds by
        \eqref{diffk}. 
        
   Meanwhile, by \eqref{deltak-k+1},  if we write $ k+1=k^{1+\delta} $, we have
        \[
            \delta=O\left(\frac{1}{k\ln k}\right)\leq 3/2,
        \]
        and thus by Lemma \ref{strictincrease},  
        \[
            \Delta(k,k+1)>(2-C\varepsilon^2)\delta>(2-2\widetilde{C}\varepsilon)\delta, 
        \] 
        which is equivalent to 
        \[
            (k+1)^{\psi(k+1)-2+2\widetilde{C}\varepsilon}>k^{\psi(k)-2+2\widetilde{C}\varepsilon},
        \] 
thus one can estimate
  $$ 
        I_2\leq C\epsilon^{4/3+2\tau}+C\epsilon^{2+2\tau}(n-1)^{\psi(n-1)-1+2\widetilde{C}\varepsilon} .$$
        
        By our selection of $x$ that $x\in [n-1,n]$ which satisfy \eqref{epn}. One can   
        write $ n=x^{1+\delta_1}=(n-1)^{1+\delta_2} $ with $ \delta_1<\delta_2 $. Noting $\psi(\cdot)\in [2,4]$, then by  \eqref{deltak-k+1} and Lemma \ref{strictincrease}, we have 
       \[
        \left|\frac{2}{\psi(n)}-\frac{2}{\psi(x)}\right|<\frac{1}{2}|\psi(n)-\psi(x)|<5\delta_2<\frac{C}{n\ln n}<\epsilon^{1/2}<\widetilde{C}\varepsilon,
       \] 
       where the fourth inequality follows by \eqref{epn} and the last inequality follows from the selection of $\epsilon_*$. 
       Finally, again  by \eqref{epn}, we can further estimate
       $$ \begin{aligned}
        I_2 &\leq  
         C\epsilon^{4/3+2\tau}+ C\epsilon^{2+2\tau}\epsilon^{-2}\epsilon^{(1-2\widetilde{C}\varepsilon)\frac{2}{\psi(n)}} \\
        &\leq C\epsilon^{4/3+2\tau}+ C\epsilon^{\frac{2}{\psi(x)}+2\tau-3\widetilde{C}\varepsilon}
    \end{aligned}$$
    the result follows. 
    \end{proof}

As a consequence, we have the following corollary, which finish the whole proof of Proposition \ref{THM5.1}:

 \begin{corollary}\label{lowerbound-mu}
        For any $ \varepsilon>0 $,  if $ \epsilon\leq \epsilon_*(\tilde{\epsilon}(K_4),\varepsilon)$, we have
        \[
            \epsilon^{f(\epsilon)+32\widetilde{C}\varepsilon}\leq \mu(E-\epsilon,E+\epsilon)\leq \epsilon^{f(\epsilon)-\widetilde{C}\varepsilon}.
        \] 
    \end{corollary}
    \begin{proof}
In view of \eqref{uppmu}, we only prove the lower bound.  Note by \eqref{error-f} of Lemma \ref{in1}, if $\tau \in [16\widetilde{C}\varepsilon, \frac{1}{24}]$, then we have estimate
\[
            (1+\tau)(f(\epsilon^{1+\tau})+\widetilde{C}\varepsilon) +\widetilde{C}\varepsilon\leq f(\epsilon)+2\tau-3\widetilde{C}\varepsilon.
        \] 
Once we have this, just take $\tau =16\widetilde{C}\varepsilon$, and apply Lemma \ref{mu2}, we have  estimate
        \begin{align*}
                \mu(E-\epsilon,E+\epsilon)&\geq \epsilon^{1+\tau}\Im M(E+\ii\epsilon^{1+\tau})-O(\epsilon^{f(\epsilon)+2\tau-3\widetilde{C}\varepsilon}) \\
                &\geq  \epsilon^{(1+\tau)(f(\epsilon^{1+\tau})+\widetilde{C}\varepsilon)}-O(\epsilon^{f(\epsilon)+2\tau-3\widetilde{C}\varepsilon})\\
                &\geq \epsilon^{f(\epsilon)+32\widetilde{C}\varepsilon}. 
                \qedhere
            \end{align*}
\end{proof}

\section{Proof of the Main Theorem}\label{Sec6}

\subsection{Exact local distribution for reducible energies}

 We want to emphasize  that  Proposition \ref{THM5.1} actually gives the mechanism of getting the local  distribution of the spectral measure. In our subsequent discussion, we will focus on the most elementary scenario, where the Schrödinger cocycle is reducible.

\begin{proposition}\label{Prop:finite-resonnace}
  We have the following:
      \begin{enumerate}
          \item If $ (\alpha,S_E^v) $ is conjugated to $\begin{pmatrix}
          \cos 2\pi\rho & -\sin 2\pi\rho\\
          \sin 2\pi\rho &\cos 2\pi\rho
          \end{pmatrix}$. Then for every $\e>0$, there exists $ \epsilon_*=\epsilon_*(\e)>0$ such that  for $ \epsilon\leq\epsilon_*  $ we have        
          \[
            C^{-1}\epsilon^{1+\e} \leq \mu(E-\epsilon,E+\epsilon)\leq C\epsilon. 
        \] 
        \item If $ (\alpha,S_E^v) $ is conjugated to $\begin{pmatrix}
          1 & \nu\\
          0 &1
          \end{pmatrix}$ with $\nu\neq 0$. Then for every $\e>0$, there exists $ \epsilon_*=\epsilon_*(\e)>0$ such that  for $ \epsilon\leq\epsilon_*  $ we have        
          \[
            C^{-1}\epsilon^{\frac{1}{2}+\e}\leq \mu(E-\epsilon,E+\epsilon)\leq C\epsilon^{\frac{1}{2}}.
        \] 
        
      \end{enumerate}
    \end{proposition}

  \begin{proof}
  We only prove the item (2), as the case (1) is similar. The strategy of the proof is same as Proposition \ref{THM5.1}, now instead of Lemma \ref{upperb}, we have the following estimate:
    \begin{lemma}\label{lem:finite-M}
        There exists $\epsilon_*>0$ such that for all $ \epsilon\leq \epsilon_* $, we have 
        \[
              \epsilon\Im M(E+\ii \epsilon) \thickapprox \epsilon^{\frac{1}{2}}.
        \] 
    \end{lemma}
    \begin{proof}
    By  Proposition \ref{Prop:AsyFormula-2}, there exists $K_0>0$ and $\beta\in(-\pi/2,\pi/2]$, such that for all $k\geq K_0$, we have
        \begin{align}
 \label{rela1-2}  
 \det P_{(k),\pm} 
 &\thickapprox k^{4},\\
 \label{rela2-2}       
 \|P_{(k),\pm}^{-1}\|^{-1}
 &\thickapprox  k,\\
 \label{rela3-2}   
\|u_{\beta}^\pm\|_{2k}^{2}
 &\thickapprox k.
            \end{align}  
We  select $ \epsilon(k)$ such that $ \det P_{(k),+}=\frac{1}{ \epsilon(k)^2} $,
  and then
  select  $ 2k_1 $ to be the closest even number of $ L^-( \epsilon(k)) $ with $ L^-( \epsilon(k))\leq 2k_1 $,
        where we recall that for any given $\epsilon$,  $L^{-}(\epsilon)$ is given by 
         $$ 
         ||| K^-(E)|||_{L^{-}}^2=\frac{1}{\epsilon^2}.$$
    Take $K_1\geq K_0 $ be minimal such that $k_1\geq K_0$, provided $k\geq K_1$.
      By \eqref{rela1-2},  we have 
           $$  k_1^4 \thickapprox \frac{1}{ \epsilon(k)^2}\thickapprox k^4,$$
 hence $k_1\thickapprox k$, $\epsilon(k)\thickapprox k^{-2}$. Then as a consequence of   \eqref{rela1-2}, \eqref{rela2-2}, \eqref{rela3-2}, we have 
        \[
            \|P_{(k),+}^{-1}\|\lesssim k^{-1}\lesssim \epsilon(k)^{\frac{1}{2}},
        \] 
        \[
            \|u_{\beta}^+\|_{2k}^{-2}\gtrsim  k^{-1}\gtrsim \epsilon(k)^{\frac{1}{2}},
        \] 
        \[
            \|u_{\beta}^-\|_{L^-( \epsilon(k))}^{-2}\geq \|u_{\beta}^-\|_{2k_1}^{-2}\gtrsim  k^{-1}\gtrsim \epsilon(k)^{\frac{1}{2}}.
        \] 
       By Corollary \ref{Corspecmeas}, one has for any $k\geq K_1 $,
        \[
            \epsilon(k) \Im M(E+\ii \epsilon(k))\thickapprox \epsilon(k)^{\frac{1}{2}}.
        \] 
        Following the same line as in the proof of Lemma \ref{upperb}, 
        noting $ \epsilon(k+1)\geq c \epsilon(k) $, and $ \frac{1}{\epsilon}\Im M(E+\ii\epsilon) $ is monotonic in $ \epsilon $, it follows that  
        \[
            \epsilon\Im M(E+\ii \epsilon) \thickapprox \epsilon^{\frac{1}{2}}.
        \] 
     for any  $\epsilon\leq \epsilon_*=\epsilon(K_1)$, we thus finish the proof.
    \end{proof}

      Now let $ \epsilon\leq \epsilon_*(K_1)^3 $ and small enough, then one can apply Lemma \ref{lem:finite-M}, consequently by   \eqref{upperb-mu}, we have        
        \begin{equation}\label{uppmu-2}
            \mu(E-\epsilon,E+\epsilon)\lesssim \epsilon^{\frac{1}{2}}.
 \end{equation}
 In the following, we estimate the lower bound. 
        As before, we split the integral $$  \epsilon^{1+\tau}\Im M(E+\ii\epsilon^{1+\tau})=\int\frac{\epsilon^{2+2\tau}}{(E'-E)^2+\epsilon^{2+2\tau}}\, d\mu (E'):=I$$ into three parts: $ I_1=\int_{|E'-E|\leq \epsilon} $, $ I_2=\int_{\epsilon\leq |E'-E|\leq \epsilon^{1/3}} $, and $ I_3=\int_{|E'-E|\geq \epsilon^{1/3}} $, where 
       \(
        I_1\leq \mu(E-\epsilon,E+\epsilon),
       \) 
       and 
       \(
        I_3\lesssim \epsilon^{4/3+2\tau}. 
       \) 
       Let $ m,n $ be the integer such that $$m^{-2}\geq \epsilon^{1/3}\geq (m+1)^{-2},\qquad (n-1)^{-2}\geq \epsilon\geq n^{-2}.$$     
       We now estimate $I_2$ as follows: 
       $$
       \begin{aligned}
        I_2=&\int_{\epsilon\leq |E'-E|\leq \epsilon^{1/3}}\frac{\epsilon^{2+2\tau}}{(E'-E)^2+\epsilon^{2+2\tau}}\dif\mu(E')
        \leq\sum_{k=m+1}^{n}\int_{k^{-2}\leq |E'-E|\leq(k-1)^{-2}}\frac{\epsilon^{2+2\tau}}{\epsilon^{2+2\tau} +k^{-4}}\dif\mu(E')\\
        &\leq \frac{\epsilon^{2+2\tau}}{\epsilon^{2+2\tau}+(m+1)^{-4}}\mu(E-m^{-2},E+m^{-2})\\
        &\quad+\sum_{k=m+1}^{n-1}\left[\frac{\epsilon^{2+2\tau}}{\epsilon^{2+2\tau}+(k+1)^{-4}}-\frac{\epsilon^{2+2\tau}}{\epsilon^{2+2\tau}+k^{-4}}\right]\mu(E-k^{-2},E+k^{-2})\\
        &\quad-\frac{\epsilon^{2+2\tau}}{\epsilon^{2+2\tau}+n^{-4}}\mu(E-n^{-2},E+n^{-2})\\        
        &\lesssim \epsilon^{4/3+2\tau}+\epsilon^{2+2\tau}\sum_{k=m+1}^{n-1}\left[(k+1)^{4}-k^{4}\right]k^{-1}\\
        &\lesssim \epsilon^{4/3+2\tau}+\sum_{k=m+1}^{n-1}\epsilon^{2+2\tau}k^{2}\\
        &\leq \epsilon^{4/3+2\tau}+\epsilon^{2+2\tau}(n-1)^{3} \\
        &\leq\epsilon^{4/3+2\tau}+ \epsilon^{\frac{1}{2}+2\tau}.
      \end{aligned}$$

   Once we have this, again by Lemma \ref{lem:finite-M}, we have 
   \begin{eqnarray*}
       \mu(E-\epsilon,E+\epsilon)\geq\epsilon^{1+\tau}\Im M(E+\ii\epsilon^{1+\tau})-O(\epsilon^{\frac{1}{2}+2\tau})  \gtrsim \epsilon^{\frac{1}{2}+\frac{\tau}{2}}. 
   \end{eqnarray*} Take $\tau=2\e,$ the result follows.
      \end{proof}

\subsection{Proof of Theorem \ref{Thm:generalholder}:}

We first recall the following version of quantitative almost reducibility:

    \begin{proposition}\label{GYZ-general-good-THM}\label{GYZ-general-good-THM-multi}
        Let $ \alpha\in \op{DC}_d $, $A\in C^\omega(\T^d,\op{SL}(2,\R))$, suppose that $(\alpha,A)$ is almost reducible in the band $\{\theta:|\Im \theta|<h\}$ and is not uniformly hyperbolic.
        Then for any $ \varepsilon>0 $,  there exist $ (n_s)_{s\geq 1}\subset \Z^d $ and $s_0>0$ such that 
    
         \begin{itemize}
         \item  $(\alpha,A)$ is $ (C_0'(\alpha,A,\varepsilon),C_0(\alpha),c_0\varepsilon,(n_s)_{s\geq s_0}) $-good.
         \item If $(n_s)_{s\geq 1}$ is infinite, then for $s\geq s_0$,
             \begin{eqnarray}\label{decay}
                 -\dfrac{\ln |\nu_{n}|}{n}&=&\zeta_{n}\in[(1-\e)2\pi h,\infty].
             \end{eqnarray}
         \end{itemize}   
     \end{proposition}

\begin{proof} 
The proof is essentially contained in \cite{Universal}, we give a sketch. 
Indeed, since $(\alpha,A)$ is almost reducible in the band $\{\theta : |\Im \theta|< h\}$ and is not uniformly hyperbolic, then for any $\e>0$, there exist $B_\e\in C^\omega_{h'}(\T,\op{SL}(2,\R))$ and $R_\e\in\op{SO}(2,\R)$ (see \cite{yzhou} for details)  such that $$
\|B_\e(\theta+\alpha)^{-1}A(\theta)B_\e(\theta)-R_\e\|_{h'}<\mathfrak{c}=\mathfrak{c}(\alpha,R_\e, h, (1-\tfrac{\e}{2})h).
$$ Indeed, $ R_\e $ varies in $ \mathrm{SO}(2,\R) $, one can just take 
$\mathfrak{c}(\alpha,R_\e, h, (1-\frac{\e}{2})h)$ to be  uniform with respect to $ R_\e $. Then one can apply Proposition \ref{reducibility-stru-1} to $(\alpha,\tilde{A})$, where $\tilde{A}(\theta)=B_\e(\theta+\alpha)^{-1}A(\theta)B_\e(\theta)$. 

Consequently, we have 
 $$
    B_s^{-1}(\theta+\alpha)\tilde{A}(\theta)B_s(\theta)=A_se^{f_s(\theta)}=M^{-1}exp \left(
    \begin{array}{ccc}
     i t^s &  \nu^s\\
    \bar{\nu}^s &  -i t^s
     \end{array}\right)Me^{f_s(\theta)}
    $$
    with estimates \eqref{es2}-\eqref{es5} hold. 
    Let $n_s=\deg B_\e B_s$, \eqref{es5} imply that 
    \begin{equation}\label{nmres}
        |n_s- \mathfrak{m}_{s}| \leq |n_s-\deg B_s |+ |\deg B_s-\mathfrak{m}_{s} | \leq o(1) \mathfrak{m}_{s}
    \end{equation}
    where $\mathfrak{m}_{s} $ are the KAM resonances defined in Proposition \ref{reducibility-stru-1}. 
 Combining with \eqref{es0}, \eqref{es4}, these  imply there exists $s_0>0$ such that 
    $(\alpha,A)$ is $ (C_0'(\alpha,A,\e),C_0(\alpha),c_0\e,(n_s)_{s\geq s_0}) $-good. Furthermore, \eqref{decay} follows from \eqref{es2} and \eqref{nmres}. 
\end{proof}
For $E\in\Sigma^-_{v,\alpha}$ with subcritical radius $h$, by Avila's solution to his Almost Reducibility Conjecture \cite{Avi2023KAM}, $(\alpha,S_E^v)$ is almost reducible in the band $\{\theta:|\Im \theta|<h'\}$ for any $0<h'<h$ (and is not uniformly hyperbolic). On the other hand, by Lemma \ref{relation1} and \eqref{N=rho}, $2\pi h\leq \delta<\infty$ means the sequence $(n_s)$ is infinite (see Section 6.3 for details).
Then (1) follows from Proposition \ref{GYZ-general-good-THM} and Proposition \ref{THM5.4}.  For the case $(2)$, 
choose $\e$ small enough such that $\delta<(1-\tfrac{\e}{2})h$, again by \cite{Avi2023KAM},  $(\alpha,A)$ is almost reducible in the band $\{\theta : |\Im \theta|<(1-\frac{\e}{2})h\} $. More precisely, 
there exist $B_\e\in C^\omega_{(1-\frac{\e}{2})h}(\T,\op{SL}(2,\R))$ and $R_\e\in\op{SO}(2,\R)$ such that $$\|B_\e(\theta+\alpha)^{-1}A(\theta)B_\e(\theta)-R_\e\|_{h'}<c_*(\alpha,R_\e, (1-\tfrac{\e}{2})h, (1-\e)h),$$  where $c_*$ is defined in Theorem \ref{reducibility-main}, then  $(\alpha,A)$ is reducible, and $(2)$ follows from Proposition \ref{Prop:finite-resonnace} (1).
For case $(3)$, this is already proved by \cite{AvilaHolderContinuityAbsolutely2011}. \qed

\begin{remark}
    If $d\geq 2$, we can apply Proposition \ref{GYZ-general-good-THM-multi} to Theorem 4.1 of \cite{AvilaHolderContinuityAbsolutely2011}, and conclude  \[
            \mu_{v,\alpha,\theta}(E-\epsilon,E+\epsilon)\leq C\epsilon^{\frac{1}{2}}.
        \]
\end{remark}

\subsection{Proof of Theorem \ref{MainTHM}}
As direct consequence of Proposition \ref{GYZ-general-good-THM}, we have the following:

    \begin{corollary}\label{GYZ-AMO-THM}
        Let $ \alpha\in\mathrm{DC}_1 $, $ 0<\lambda<1 $, $E\in\Sigma_{\lambda,\alpha}$. 
        Then for any $ \varepsilon>0 $, there exists $(n_s)_{s\geq 1}\subset \Z$ and  $ s_0>0 $ such that 
    
         \begin{enumerate}
         \item  $(\alpha,S_E^\lambda)$ is $ (C_0'(\lambda,\alpha,\varepsilon),C_0(\alpha),c_0\varepsilon,(n_s)_{s\geq s_0}) $-good.
         \item If $(n_s)_{s\geq 1}$ is infinite, then for $s\geq s_0$,
             $$ -\dfrac{\ln |\nu_{n}|}{n}=\zeta_{n}\in[-(1-\varepsilon)\ln \lambda,\infty],$$
        \item  $$
            \zeta_{n}\leq -(1+\varepsilon) \ln\lambda, \text{\qquad if\quad}\widehat{\eta}_{n} \geq -(1+\varepsilon) \ln\lambda. $$
         \end{enumerate}
     \end{corollary}
\begin{proof}
    If $ 0<\lambda<1 $, then $(\alpha,S_E^\lambda)$ is subcritical with analytic radius $-\frac{\ln \lambda}{2\pi}$ \cite{avila2015global}, then the result follows from   Proposition \ref{GYZ-general-good-THM} except (3), where $(3)$ follows from \cite[Proposition 8.2]{Universal}.
\end{proof}

Now we can give the proof of Theorem \ref{MainTHM},
and distinguish the proof into three cases:\\
 {\bf Case 1: If $\delta(\alpha,\mathcal{N}_{\lambda,\alpha}(E))=\delta\geq -\ln\lambda $ and $\mathcal{N}_{\lambda,\alpha}(E)\neq k\alpha \mod \Z$.} Fix $\epsilon_0=-\frac{1}{2}\ln\lambda$ for instance. By \eqref{N=rho},  $\mathcal{N}_{\lambda,\alpha}(E) $ shares the same $\epsilon_0$-resonances with $2\rho(\alpha,S_E^{\lambda})$ and 
\(
\delta(\alpha,\mathcal{N}_{\lambda,\alpha}(E))=\delta(\alpha,2\rho(\alpha,S_E^{\lambda}))
\). 
   Then the result follows from Corollary \ref{GYZ-AMO-THM}, Lemma \ref{relation1} and  Proposition \ref{THM5.1}. 
   
   Indeed, by definition, the $\epsilon_0$-resonances $(\ell_s)_{s\geq 1}$ of $2\rho(\alpha,S_E^\lambda)$ satisfies 
   \begin{equation}\label{eq:reso-1}
       \|2\rho(\alpha,S_E^\lambda)-\ell_s\alpha\|_\T\leq e^{-\epsilon_0|\ell_s|}.
   \end{equation}
   Repeating the line as in the proof of Proposition \ref{GYZ-general-good-THM-multi}: For any $\e>0$,  there exist $B_\e\in C^\omega_{-\frac{\ln\lambda}{2\pi}(1-\frac{\e}{2})}(\T,\op{SL}(2,\R))$ and $R_\e\in\op{SO}(2,\R)$ such that $B_\e(\theta+\alpha)^{-1}S_E^\lambda(\theta)B_\e(\theta)=A(\theta)$ with $\|A(\theta)-R_\e \|_{-\frac{\ln\lambda}{2\pi}(1-\frac{\e}{2})}$ small enough, and thus one can apply Proposition \ref{reducibility-stru-1}. By \eqref{degree} and \eqref{eq:reso-1}, for $s$ large enough, we have 
   \[
   \|2\rho(\alpha,A)-(\ell_s-\deg B_\e)\alpha\|_\T\leq e^{-\epsilon_0|\ell_s-\deg B_\e|/2}.
   \]
   Then by Lemma \ref{relation1}, there exist $s_0>0$ such that the sequence $(\ell_s-\deg B_\e)_{s\geq s_0}$ coincides with a subsequence of $(\deg B_s)_{s\geq 1}$. This is equivalent to the sequence $(\ell_s)_{s\geq s_0}$ coincides with a subsequence of $(n_s)_{s\geq 1}$. 
    On the other hand, notice that when apply Proposition \ref{THM5.1}, those  non-$\epsilon_0$-resonance sites, i.e. those sites $ n_s$ with 
    \[
   \|2\rho(\alpha,S_E^\lambda)-n_s\alpha\|_\T > e^{-\epsilon_0|n_s|}
   \] 
   have {\bf no} contribution to the asymptotic function $f(\epsilon)$: The effective sites $n_s$ actually satisfy $ \|2\rho(\alpha,S_E^\lambda) - n_s \alpha \|_\mathbb{T} \leq e^{|n_s| \ln |\lambda|} $, it even does not depend on the choice of $ 0 < \epsilon_0 < -\ln \lambda $.\\
{\bf Case 2: If $\delta(\alpha,\mathcal{N}_{\lambda,\alpha}(E))=\delta<-\ln \lambda $.} By
Proposition \ref{Prop:finite-resonnace} (1), one only need to prove $ (\alpha,S_E^\lambda) $ is reducible. When $\delta = 0$, the claim is established in Theorem 4.1 of \cite{aj1}. For the case where $0 < \delta < -\ln\lambda$, it has been verified by Theorem 1.2 of \cite{Universal}.  \\
{\bf Case 3: If $\mathcal{N}_{\lambda,\alpha}(E)=k \alpha\mod\Z$ for some $k \in \Z$.} By \cite{aj1}, the dry version of the Ten Martini Problem holds for $\alpha\in \op{DC}_1$, $0<\lambda<1$, which implies $ (\alpha,S_E^v) $ is reducible to $\begin{pmatrix}
          1 & \nu\\
          0 &1
          \end{pmatrix}$ with $\nu\neq 0$. Then the result follows from Proposition \ref{Prop:finite-resonnace} (2).    \qed

    \bigskip
\noindent\textbf{Proof of Theorem \ref{MainCor}:}
This is an immediate consequence of Theorem \ref{MainTHM}, we omit the details. \qed

  \section{General potential case}
  
    In this section, we use Anosov-Katok construction to show that the lower local dimension of spectral measure falls in $ (1/2,1) $ is a general phenomenon.

    \begin{proposition}\label{THM7.1}
        Let $ \alpha\in\mathrm{DC}_d $, $ 0<\e_0<2\pi h'<2\pi h\leq \delta<\infty $. If $ v\in C^\omega_{h}(\T^d,\R) $ satisfies 
        \[
            \|v\|_h\leq c_*(\alpha,h,h').
        \]
        Then for any $ E\in\Sigma_{v,\alpha} $,  and any $ \epsilon>0 $, there exists a sequence  $ (\hat{k}_{s})_{s\geq \hat{s}_0} $\footnote{We remark that such sequence is a subsequence of $(\delta-\e_0)$-resonances associated to $\alpha$ and $2\rho_{\tilde{v},\alpha}(E)$.},  $ \tilde{v}\in C^\omega_{h'}(\T^d,\R) $ with $
            \|\tilde{v}-v\|_{h'}\leq \epsilon, 
$ such that $( (\alpha,S_E^{\tilde{v}}) )$ is almost reducible:

\begin{enumerate}
\item $ (\alpha,S_E^{\tilde{v}}) $ is 
$ (C_0'(v,\alpha,E,h,h',\epsilon),C_0(\alpha,\delta),c_0,(\hat{k}_{s})_{s\geq \hat{s}_0}) $-good, with $ c_0<2\pi h $.  
\item For any $s\geq \hat{s}_0$, denoting $n=|\hat{k}_s|$, we have estimates
\begin{align*}
    -\dfrac{\ln |\nu_n|}{n}&=\zeta_n=2\pi h+o(1),\\
  -\dfrac{\ln \|2\rho_s\|_\T}{n} &= \widehat{\eta}_n=\delta+o(1). 
\end{align*}

\end{enumerate}

    \end{proposition}

    \begin{proof}

    As we mentioned, the whole proof further develops the fibered Anosov-Katok construction which first appeared in  \cite{KXZ2020AnosovKatok}, the advantage of our approach is that it works for \textit{all} Liouvillean rotation number. We first introduce the parameters in the construction.

Recall that
$$
\delta(\alpha,2\vartheta)=\limsup_{k\to\infty}-\frac{\ln\lVert 2\vartheta+\langle k,\alpha\rangle\rVert_\T}{|k|}.
$$ 
For any $ \epsilon\in(0,1) $, fix any $ \vartheta\in (-\frac{\epsilon}{8\pi},\frac{\epsilon}{8\pi}) $ satisfies $ \delta(\alpha,2\vartheta)=\delta\geq 2\pi h $. Thus for any $\e_0>0$, there exist sequence  $ (k_{s})_{s\geq 2} $  such that
\begin{equation*}
\|2\vartheta -\langle k_s,\alpha \rangle\|_{\R/\Z}\leq e^{-|k_s| (\delta-\e_0/2) }.
\end{equation*} 
Once we have this, let $ k_1=\bf{0} $, for $ s\geq 1 $, we inductively  define $(k_{s})$, $(\lambda_{s})$ and $(t_{s})$ as follows:
 \begin{equation}\label{k-n-1}
      -\frac{\ln \|2\vartheta-\langle k_{s},\alpha\rangle\|_\T}{|k_{s}|}=\delta+o(1),
 \end{equation}
  \begin{equation} \label{lambda-n}
         \lambda_{s}=e^{-2\pi h|k_{s+1}|},
    \end{equation} 
     \begin{equation} \label{t-n}
        \begin{aligned}
            &t_{s}=\sqrt{\lambda_{s}^2+\|2\vartheta-\langle k_{s+1},\alpha\rangle\|_\T^2}.
        \end{aligned}
     \end{equation} 
 Denote $ \Sigma^+=[0,1/2)$ and $\Sigma^-=[1/2,1) $. Write
    \[
        \begin{aligned}
            2\vartheta-\langle k_{s},\alpha\rangle =j_s+x_s=[2\vartheta-\langle k_{s},\alpha\rangle]+\{2\vartheta-\langle k_{s},\alpha\rangle \}.
        \end{aligned}
    \] 
    Then by Pigeonhole principle, there are infinitely many $ k_{s} $ satisfies one of the following:
    \begin{enumerate}[label=$(\mathrm{\roman*})$\ ,align=left,widest=ii,labelsep=0pt,leftmargin=0.5em]
        \item $ x_{s}\in\Sigma^+ $, $ j_{s}\in 2\Z $,
        \item $ x_{s}\in\Sigma^+ $, $ j_{s}\in 2\Z+1 $,
        \item $ x_{s}\in\Sigma^- $, $ j_{s}\in 2\Z $,
        \item $ x_{s}\in\Sigma^- $, $ j_{s}\in 2\Z+1 $.
    \end{enumerate} 
Without loss of generality, one will just  assume
\begin{equation} \label{Special-case}
    \begin{aligned}
        \bullet\ & \vartheta\geq 0,  x_{s}\in\Sigma^+ ,  j_{s}\in 2\Z \text{ holds for all }  k_{s} 
    \end{aligned}
\end{equation} 
The other case can be dealt with similarly.    Just note by symmetry, if $ \vartheta $ has infinitely many $ k_{s} $ such that $ x_{s}\in\Sigma^\pm $, then $ -\vartheta $ has infinitely many $ k_{s} $ such that $ x_{s}\in\Sigma^{\mp} $.  Moreover,  
the Diophantine condition $\alpha\in\mathrm{DC}_d$ implies that 
\begin{equation}\label{k-n-3}
    |k_{s}|\geq e^{\frac{1}{\tau}(\delta+o(1)) |k_{s-1}|}
\end{equation}
and consequently  one can just assume
\begin{equation}\label{k-n-2}
 \sum_{s=2}^{\infty} e^{-(2\pi h-2\pi h'-o(1))|k_{s}|}<\frac{\epsilon}{2}.
\end{equation}

With the above selected sequence,  we will split the proof into three steps:\\

    {\bf \noindent Step 1: Refined Anosov-Katok construction }
    
\begin{lemma}\label{AKClemma}   Let $ \alpha\in DC_d$, $ 0<2\pi h'<2\pi h\leq \delta<\infty $. For any $ \epsilon>0$,
there exist $ A_{\infty}(\cdot)\in C^\omega_{h'}(\T^d,\mathrm{SL}(2,\R)),$
with 
$$ \|A_{\infty}(\cdot)-\mathrm{Id} \|_{h'}<\epsilon $$ 
such that $(\alpha,A_{\infty})$ is almost reducible:
        \[
        B_{s}(\cdot+\alpha)^{-1}A_{\infty}(\cdot)B_{s}(\cdot)=\begin{pmatrix}
            e^{2\pi \ii\rho_{s}}&\nu_{s}\\
             0&e^{-2\pi \ii\rho_{s}}
    \end{pmatrix}+F_{s}(\cdot).
    \] 
Moreover, we have the following:
\begin{enumerate}
\item $ (\alpha,A_{\infty}) $ is $ (1,C_0(\alpha,\delta),c_0,(k_{s})_{s\geq 2}) $-good, with $ c_0<2\pi h<\infty $.
\item $ -\dfrac{\ln |\nu_s|}{|k_s|}=2\pi h+o(1) $,
\item  $ -\dfrac{\ln \|2\rho_s\|_\T}{|k_s|}= \delta+o(1) $.
\end{enumerate}  

                    \end{lemma}
    \begin{proof}
     As  $ \mathrm{SL}(2,\R) $ is isomorphic to  $ \mathrm{SU}(1,1) $, we just carry out the Anosov-Katok construction in the group $\mathrm{SU}(1,1)$.
        Assume we are at the $ s $-th ($ s\geq 0 $) step of  construction, we start at the cocycle $ (\alpha,\bar{A}_s) $, where 
    \[
        \bar{A}_s=\exp \pi \ii \begin{pmatrix}
            t_s&\lambda_s\\
            -\lambda_s&-t_s
        \end{pmatrix},
    \] 
    and its eigenvalues are $ e^{\pm \pi\ii \|2\vartheta-\langle k_{s+1},\alpha\rangle\|_\T} $. 
    The goal of the $ (s+1) $-th step of the construction, is to construct suitable transformation $ G_s(\cdot) $ and perturbation $ \bar{f}_s (\cdot) $ such that if we perturb the cocycle $ (\alpha,\bar{A}_s) $ to $ (\alpha,\bar{A}_se^{\bar{f}_s(\cdot)}) $, the transformation $ G_s (\cdot)$ conjugate the cocycle $ (\alpha,\bar{A}_se^{\bar{f}_s(\cdot)})$ into $ (\alpha,\bar{A}_{s+1}) $.

    By  \eqref{t-n} and Lemma \ref{normalellip}, there exists $ D_s\in\mathrm{SU}(1,1) $ such that
    \[
        D_s^{-1}\bar{A}_sD_s=\tilde{A}_s:=
            \exp \pi\ii  \begin{pmatrix}
                \|2\vartheta-\langle k_{s+1},\alpha\rangle\|_\T&0\\
                 0&-\|2\vartheta-\langle k_{s+1},\alpha\rangle\|_\T
        \end{pmatrix},
    \] 
    with $ D_0=\mathrm{Id} $, and for $ n\geq 1 $,
    \begin{equation}\label{eqDn} 
        \|D_s\|^2=\frac{t_s+\lambda_s}{\sqrt{t_s^2-\lambda_s^2}}\in[\frac{t_s}{\|2\vartheta-\langle k_{s+1},\alpha\rangle\|_\T},\frac{2t_s}{\|2\vartheta-\langle k_{s+1},\alpha\rangle\|_\T}].
    \end{equation} 

    Denote 
    \[
        H_{m}(\cdot)=\exp \pi\ii  \begin{pmatrix}
        \langle m,\cdot\rangle&0\\
        0&-\langle m,\cdot\rangle
        \end{pmatrix}\in \mathrm{SU}(1,1),
    \] 
    then we have the following:

    \begin{lemma}
        Let $ s_*=k_{s+2}-k_{s+1}\in\Z^d $. There exists  $ \tilde{f}_s\in C^\omega_{h'}(\T^d,\mathrm{su}(1,1)) $ with 
        \begin{equation} \label{eqfn}
            \|\tilde{f}_s\|_{h'}\leq  e^{-(2\pi h-2\pi h'-o(1)) |k_{s+2}|}
        \end{equation} 
        such that
   \begin{equation} \label{eqfn-1}
        H_{s_*}(\cdot+\alpha)^{-1}\tilde{A}_se^{\tilde{f}_s(\cdot)}H_{s_*}(\cdot)\\
            =\exp \pi\ii \begin{pmatrix}
                t_{s+1}&\lambda_{s+1} \\
                -\lambda_{s+1} &-t_{s+1}
            \end{pmatrix}=\bar{A}_{s+1}.    
       \end{equation} 
    \end{lemma}
    \begin{proof}
    Recall that 
       \[
        \begin{aligned}
            2\vartheta-\langle k_{s+1},\alpha\rangle =j_{s+1}+x_{s+1}=[2\vartheta-\langle k_{s+1},\alpha\rangle]+\{2\vartheta-\langle k_{s+1},\alpha\rangle \}.
        \end{aligned}
    \] 
    then by our construction $ x_{s+1},x_{s+2}\in\Sigma^+ $, $ j_{s+1},j_{s+2}\in 2\Z $, consequently we have
    \begin{equation} \label{keyOb}
        \begin{aligned}
            \|2\vartheta-\langle k_{s+1},\alpha\rangle\|_\T-\langle s_*,\alpha\rangle \mod 2\Z=\|2\vartheta-\langle k_{s+2},\alpha\rangle\|_\T.
        \end{aligned}
    \end{equation}

 As  $\tilde{A}_s $ commutes with  $H_{s_*}(\cdot)$, direct computation shows that
        \[
            \begin{aligned}
                Z
            =&\ \tilde{A}_s^{-1}H_{s_*}(\cdot+\alpha)\exp \pi\ii \begin{pmatrix}
                t_{s+1}&\lambda_{s+1} \\
                -\lambda_{s+1} &-t_{s+1}
            \end{pmatrix}H_{s_*}(\cdot)^{-1}\\
            =&\ \left(H_{s_*}(\cdot+\alpha)\tilde{A}_s^{-1}H_{s_*}(\cdot)^{-1}\right)\left(H_{s_*}(\cdot)\exp \pi\ii \begin{pmatrix}
                t_{s+1}&\lambda_{s+1} \\
                -\lambda_{s+1} &-t_{s+1}
            \end{pmatrix}H_{s_*}(\cdot)^{-1}\right)= e^X e^Y
        \end{aligned}
        \] 
        where $ X\in \mathrm{su}(1,1)$,  $  Y\in C^\omega(\T^d,\mathrm{su}(1,1)) $ are of the form
        \[
            \begin{aligned}
            X
            =&-\pi \ii \begin{pmatrix}
                \|2\vartheta-\langle k_{s+1},\alpha\rangle\|_\T-\langle s_*,\alpha\rangle &0\\
                 0&-\|2\vartheta-\langle k_{s+1},\alpha\rangle\|_\T+\langle s_*,\alpha\rangle\end{pmatrix}\\
                 =&-\pi \ii \begin{pmatrix}
                    \|2\vartheta-\langle k_{s+2},\alpha\rangle\|_\T&0\\
                     0&-\|2\vartheta-\langle k_{s+2},\alpha\rangle\|_\T\end{pmatrix},\\
            Y=&\ \pi\ii \begin{pmatrix}
                    t_{s+1}&\lambda_{s+1}e^{2\pi\ii \langle s_*,\cdot\rangle} \\
                    -\lambda_{s+1}e^{-2\pi\ii \langle s_*,\cdot\rangle} &-t_{s+1}
                \end{pmatrix},
            \end{aligned}
        \] 
        by (\ref{k-n-1}-\ref{t-n}) and \eqref{k-n-3}, we have estimates
        \[
            \begin{aligned}
                &\|X\|\leq  e^{-(\delta-o(1))|k_{s+2}|}\leq e^{-(2\pi h-2\pi h'-o(1)) |k_{s+2}|},\\
                &\|Y\|_{h'}\leq 2\pi e^{-2\pi h |k_{s+2}|}e^{2\pi h' |k_{s+2}-k_{s+1}|}\leq e^{-(2\pi h-2\pi h'-o(1)) |k_{s+2}|}.
            \end{aligned} 
        \] 

We recall the following well-known Baker-Campbell-Hausdorff (BCH) Formula:
    \begin{lemma}\label{BCHF}
        Let $ \mathcal{A} $ be a Banach algebra over $ \R $ or $ \C $. Let $ X, Y\in\mathcal{A} $ with $$ \|X\|+\|Y\| <\frac{\log 2}{2},$$ then there exists $ \tilde{Z}\in\mathcal{A} $ such that
        \[
            e^Xe^Y=e^{\tilde{Z}},
        \] 
        where $\tilde{Z}$ takes the form
        \[
            \tilde{Z}=X+Y+\frac{1}{2}[X,Y]+\frac{1}{12}([X,[X,Y]]+[Y,[Y,X]])+\text{\rm Higher order terms}.
        \] 
    \end{lemma}
    Applying Lemma \ref{BCHF}, there exists $ \tilde{f}_s\in C^\omega_{h'}(\T^d,\mathrm{su}(1,1)) $ such that 
            $Z=e^{X}e^Y=e^{\tilde{f}_s},
           $ i.e. \eqref{eqfn-1} holds. Furthmore, we have estimate
    \[
        \|\tilde{f}_s\|_{h'}\leq 2(\|X\|+\|Y\|_{h'})\leq  e^{-(2\pi h-2\pi h'-o(1)) |k_{s+2}|}. 
    \] 
    The result follows.
    \end{proof}

    Notice that if one let $ \bar{f}_s(\cdot)=D_s\tilde{f}_s(\cdot)D_s^{-1} $,
    then
    \begin{equation} \label{AKC1}
        \begin{aligned}
            H_{s_*}(\cdot+\alpha)^{-1}D_s^{-1}\bar{A}_se^{\bar{f}_s(\cdot)}D_s H_{s_*}(\cdot)
            &= H_{s_*}(\cdot+\alpha)^{-1}\tilde{A}_se^{\tilde{f}_s(\cdot)}H_{s_*}(\cdot)\\
            &=\exp \pi\ii \begin{pmatrix}
                t_{s+1}&\lambda_{s+1} \\
                -\lambda_{s+1} &-t_{s+1}
            \end{pmatrix}=\bar{A}_{s+1}.
        \end{aligned}
    \end{equation} 
    We finished the $ (s+1) $-th step of the construction, as desired.

    Let $ G_s(\cdot)=D_s H_{s_*}(\cdot) $ and $ \bar{B}_0=\mathrm{Id} $, $ \bar{B}_s=G_0\cdots G_{n-2}G_{n-1} $. Then we can construct the desired cocycle sequences:
    \begin{equation} \label{AKC2}
        A_s(\cdot)=\bar{B}_s(\cdot+\alpha)\bar{A}_s \bar{B}_s(\cdot)^{-1}.
    \end{equation} 
    In the following, we will show the cocycle $ (\alpha,A_s(\cdot)) $ converges. To show this, just notice that by (\ref{AKC1}) and (\ref{AKC2}), 
     \[
        \|A_0-\mathrm{Id}\|_{h'}\leq \pi t_0 = 2\pi \vartheta \leq \frac{\epsilon}{4},
    \]  and furthermore
    \begin{equation}\label{n+1-n}
        A_{s+1}-A_s=\bar{B}_s(\cdot+\alpha)(\bar{A}_se^{\bar{f}_s(\cdot)}-\bar{A}_s)\bar{B}_s(\cdot)^{-1}.
    \end{equation} 
By (\ref{eqDn}-\ref{eqfn}), one can estimate
    \begin{equation} \label{eqBn}
        \begin{aligned}
            \|B_{s}\|_{h'}^2 &\leq \prod_{i=0}^{s-1}\|D_i\|^2\|H_{i_*}\|_{h'}^2\\
            &\leq e^{2\pi h'\sum_{i=0}^{s-1}(|k_{i+1}-k_{i+2}|)}\prod_{i=1}^{s-1} \frac{2t_i}{\|2\vartheta-\langle k_{i+1},\alpha\rangle\|_\T}\\
            &\leq e^{2\pi h' \sum_{i=0}^{s-1}(|k_{i+1}-k_{i+2}|) +(\delta-2\pi h+o(1))\sum_{i=2}^{s}|k_i|}\\
            &\leq e^{2\pi h' (1+o(1))|k_{s+1}|+(\delta-2\pi h+o(1))|k_{s}|},
         \end{aligned}
     \end{equation} 
As a consequence, by \eqref{k-n-3}, we have estimates 
\begin{equation}\label{bn0}
   \|\bar{B}_{s}\|_{0} \leq e^{(\delta-2\pi h+o(1))|k_{s}|/2}\leq  |k_{s+1}|^{C_0}, 
\end{equation}
    \begin{equation} 
            \begin{aligned}
            \|A_{s+1}-A_s\|_{h'}&\leq 2\|\bar{B}_s\|_{h'}^2\|\bar{f}_s\|_{h'}\leq 2\|\tilde{f}_s\|_{h'} \|D_s\|^2\|B_{s}\|_{h'}^2\\
            &\leq e^{-(2\pi h-2\pi h'-o(1))|k_{s+2}|}
            e^{2\pi h' (1+o(1))|k_{s+1}|+(\delta-2\pi h+o(1))|k_{s}|}\\
           &\leq e^{-(2\pi h-2\pi h'-o(1))|k_{s+2}|}.
        \end{aligned}
    \end{equation} 
    Therefore, by \eqref{k-n-2}, the cocycle $ (\alpha,A_s(\cdot)) $ converges to some cocycle $ (\alpha,A_{\infty}(\cdot))\in \T^d\times \mathrm{SU}(1,1) $, with estimate
    \[
        \|A_\infty-\mathrm{Id} \|_{h'}\leq \|A_0-\mathrm{Id}\|_{h'}+\sum_{s=0}^{\infty}\|A_{s+1}-A_s\|_{h'} <\frac{\epsilon}{4}+\frac{\epsilon}{2}<\epsilon.
    \] 
    And by \eqref{k-n-2}, (\ref{n+1-n}), \eqref{eqBn}, we have
    \[
        \begin{aligned}
            &\bar{B}_s(\cdot+\alpha)^{-1}A_{\infty}(\cdot)\bar{B}_s(\cdot)\\=&\bar{B}_s(\cdot+\alpha)^{-1}A_{n}(\cdot)\bar{B}_s(\cdot)+\sum_{j=n}^\infty \bar{B}_s(\cdot+\alpha)^{-1}(A_{j+1}(\cdot)-A_j(\cdot))\bar{B}_s(\cdot)\\
            =&\bar{A}_s+\sum_{j=s}^\infty (\bar{B}_s^{-1}\bar{B}_j)(\cdot+\alpha)(\bar{A}_j e^{\bar{f}_j(\cdot)}-\bar{A}_j)(\bar{B}_s^{-1}\bar{B}_j)^{-1}(\cdot):=\bar{A}_s+\bar{F}_s(\cdot),
        \end{aligned}
    \] 
    with estimates for $ s\geq 1 $, 
    \begin{equation}\label{fn0}
        \begin{aligned}
            \|\bar{F}_s\|_{h'}\leq& 2\sum_{j=s}^{\infty} \|\bar{B}_s^{-1}\bar{B}_j\|_{h'}^2\|\bar{f}_j\|_{h'}\leq 2\sum_{j=s}^{\infty}\|\bar{f}_j\|_{h'}\prod_{i=0}^{j-1}\|D_i\|^2\|H_{i_*}\|_{h'}^2\\
            \leq& \sum_{j=s}^{\infty}  e^{-(2\pi h-2\pi h'-o(1)) |k_{j+2}|}\leq e^{-(2\pi h-2\pi h'-o(1)) |k_{s+2}|}\to 0.
        \end{aligned}
        \end{equation}
    Thus $ (\alpha,A_\infty) $ is almost reducible. 

          Combines Lemma \ref{lem:quanti-Schurlem} and (\ref{k-n-1}-\ref{t-n}), there exist unitary $ U_s\in\mathrm{SL}(2,\C) $ such that
    \[
        U_s^{-1}\exp \pi \ii \begin{pmatrix}
            t_s&\lambda_s\\
            -\lambda_s&-t_s
        \end{pmatrix} U_s=\begin{pmatrix}
            e^{2\pi \ii\rho_{s+1}}&\nu_{s+1}\\
             0&e^{-2\pi \ii\rho_{s+1}}
    \end{pmatrix}
    \] 
    with 
    \begin{eqnarray*}
         \|2\rho_{s+1}\|_\T= \|2\vartheta-\langle k_{s+1},\alpha\rangle\|_\T=e^{-(\delta+o(1)) |k_{s+1}|},\label{Esti-1}\\
     \nu_{s+1}\in  [\pi\lambda_s,2\pi \lambda_s]\sim e^{-(2\pi h+o(1)) |k_{s+1}|}.\label{Esti-2}
    \end{eqnarray*}
     Let $ B_{s+1}=\bar{B}_sU_s $, then $(2)$, $(3)$ follows.

    The rest is to show the goodness of the cocycle. First we show  \begin{equation}\label{rho=theta}
    \rho(\alpha,A_{\infty})=\vartheta\mod \Z .
\end{equation}
Indeed, as 
        \[
            \deg G_s=\deg H_{s_*}=k_{s+2}-k_{s+1},\qquad \deg \bar{B}_s=\sum_{i=0}^{s-1}(k_{i+2}-k_{i+1})=k_{s+1},
        \] 
        then by assumption \eqref{Special-case}, and \eqref{rotationnumber1},
        \begin{equation} 
            \begin{aligned}
                &\left\|\rho(\alpha,A_{\infty})-\frac{\langle \deg \bar{B}_s, \alpha\rangle}{2}\mod \Z -\rho(\alpha,\bar{A}_s) \right\|_\T\\
                &\qquad\qquad=\inf_{j\in\Z}\left\vert\rho(\alpha,A_\infty)-\frac{\langle k_{s+1}, \alpha\rangle}{2}\mod \Z-\frac{\|2\vartheta-\langle k_{s+1},\alpha\rangle\|_\T}{2} -j\right\vert\\
                &\qquad\qquad=\inf_{j\in\Z}\left\vert\rho(\alpha,A_\infty)-\frac{\langle k_{s+1}, \alpha\rangle}{2}\mod \Z-\vartheta+\frac{\langle k_{s+1}, \alpha\rangle}{2}-\frac{j_{s+1}}{2}-j\right\vert\\
                &\qquad\qquad\leq C \|\bar{F}_s\|_0^{\frac{1}{2}}\to 0,
            \end{aligned}
        \end{equation}  
        which implies $ \rho(\alpha,A_\infty)=\vartheta \mod \Z $, by the selection, this further implies that $(k_{s})$ is a subsequence of the $(\delta- \frac{\e_0}{2})$-resonances of $2\rho(\alpha,A_{\infty})$. By \eqref{lambda-n}, \eqref{t-n}, \eqref{bn0} and \eqref{fn0}, we conclude that $(\alpha,A_{\infty})$  is $ (1,C_0,c_0,(k_{s})_{s\geq 2}) $-good.
        \end{proof}
\begin{remark}
    In the proof of other cases, the key difference is \eqref{keyOb}. It might be adjusted slightly, these adjustments will lead  those $ \bar{A}_s $ in the almost reducible progress should be replaced by $ -\bar{A}_s $, or $ \pm \bar{A}_{s}^{-1} $. But all these adjustments will not change any estimates in the next step.
\end{remark}

\smallskip

\noindent\textbf{Step 2: Schr\"odinger form}. In this step,  we will perturbate $ v $ to  $ \tilde{v} $, such that $ S_E^{\tilde{v}} $ can be conjugated to the $ A_{\infty} $ we constructed.

    \begin{lemma}\label{conjugatetonearId}
        Let $ \alpha\in \mathrm{DC}_d $,  $ 0<h'<h $.  Suppose that $(\alpha, A)$ is almost reducible in the band $\{|\Im \theta|<h\}$ but is not uniformly hyperbolic. Then, for any $ \epsilon>0 $, there exists $ C_{\epsilon}\in C^\omega_{h'}(\mathbb{T}^d,\mathrm{PSL}(2,\R)) $ such that:
         \[
            \|A(\theta)-C_{\epsilon}(\theta+\alpha)\mathrm{Id}C_{\epsilon}^{-1}(\theta)\|_{h'}<\epsilon.
         \]
    \end{lemma}
    \begin{proof}
        The proof is essentially contained in \cite[Section 4.2]{KXZ2020AnosovKatok}. We summerized it into this version. 
    \end{proof}
    To finish the proof of Proposition \ref{THM7.1}, we need the following observation, which states that any non-Schr\"odinger perturbation of Schr\"odinger cocycle can be converted to a Schr\"odinger
    cocycle. 
    \begin{theorem}\cite[Theorem 18]{avila2015global}\label{nonSchtoSch}
        Let $ v\in C^\omega_h(\T^d,\R) $ be non-identically zero. There exists $ \tilde{\varepsilon}>0 $, such that if $ A\in C^\omega_h(\T^d,\mathrm{SL}(2,\R)) $ satisfies $ \|A-S_E^{v} \|_h<\varepsilon<\tilde{\varepsilon} $, then there exist $ \tilde{v}\in C^\omega_h(\T^d,\R) $  with $ \|\tilde{v}-v\|_h<\varepsilon $,  $ \tilde{B}\in C^\omega_h(\T^d,\mathrm{SL}(2,\R)) $, 
        with $ \|\tilde{B}-\mathrm{Id}\|_h<\varepsilon $,
         such that 
        \[
            \tilde{B}(\theta+\alpha)^{-1}A(\theta)\tilde{B}(\theta)=S_E^{\tilde{v}}(\theta).
        \] 
    \end{theorem}
    
Now we finish the proof of Proposition \ref{THM7.1}.   
Take $ \tilde{h}=\frac{h+h'}{2} $. Suppose that $v$ is a small analytic potential in the perturbative regime, i.e., $\|v\|_h\leq c_*:= c_*(\alpha,h,\tilde h)$, then $ (\alpha,S_E^v) $ is almost reducible in the band $\{|\Im \theta|<\tilde{h}\}$, by Proposition \ref{reducibility-stru-1} (we can take $\mathfrak{c}$ to be uniform with respect to $E\in \Sigma_{v,\alpha}$). 
On the other hand, if $ E\in\Sigma_{v,\alpha} $, then $ (\alpha,S_E^v) $ is not uniformly hyperbolic.
  For any $ 0<\epsilon<\tilde{\e} $, by Lemma \ref{conjugatetonearId}, there exist $ C_{\epsilon}\in C^\omega_{h'}(\T^d,\mathrm{PSL}(2,\R)) $, such that 
  \begin{equation}\label{eq:close-to-id-1}
      \|S_E^v-C_{\epsilon}(\theta+\alpha)\mathrm{Id}C_{\epsilon}^{-1}(\theta)\|_{h'}\leq \frac{\epsilon}{2}.
  \end{equation}
    
    By Lemma \ref{AKClemma}, there exists a sequence $(k_{s})$, $ A_{\infty}\in C^\omega_{h'}(\T^d,\mathrm{SL}(2,\R))  $ satisfies 
    \begin{equation}\label{close-1}
    \|A_\infty-\mathrm{Id}\|_{h'}<\frac{\epsilon}{2\|C_{\epsilon }\|_{h'}^2}.    
    \end{equation} 
    such that 
\begin{enumerate}
\item $ (\alpha,A_{\infty}) $ is $ (1,C_0(\alpha,\delta),c_0,(k_{s})_{s\geq 2}) $-good,  with $ c_0<2\pi h<\infty $. 
\item $ -\dfrac{\ln |\nu_s|}{|k_s|}=2\pi h+o(1) $,
\item  $ -\dfrac{\ln \|2\rho_s\|_\T}{|k_s|} =\delta+o(1) $.
\end{enumerate}
Note \eqref{eq:close-to-id-1} and \eqref{close-1} imply that
    \[
        \|S_E^v-C_{\epsilon}(\cdot+\alpha)A_{\infty}(\cdot)C_{\epsilon}^{-1}(\cdot)\|_{h'}<\epsilon.
    \] 
  Consequently by Theorem \ref{nonSchtoSch}, there exist $ \tilde{v}\in C^\omega_{h'}(\T^d,\R) $, $ \tilde{B}_{\epsilon}\in C^\omega_{h'}(\T^d,\mathrm{SL}(2,\R)) $ such that 
    \[
        \tilde{B}_{\epsilon}(\cdot+\alpha)^{-1}C_{\epsilon}(\cdot+\alpha)A_{\infty}(\cdot)C_{\epsilon}^{-1}(\cdot)\tilde{B}_{\epsilon}(\cdot)=S_E^{\tilde{v}}(\theta),
    \] 
    with $ \|\tilde{v}-v\|_{h'}<\epsilon $. 
    In other word, for any $ \epsilon>0 $ there exist $ \tilde{v} $ such that $ \|\tilde{v}-v\|_{h'}<\epsilon $ such that 
    \[
        C_{\epsilon}^{-1}(\cdot+\alpha)\tilde{B}_{\epsilon}(\cdot+\alpha)S_E^{\tilde{v}}(\theta)\tilde{B}_{\epsilon}(\cdot)^{-1}C_{\epsilon}(\cdot)=A_{\infty}(\cdot)
    \] 
    with estimate
    \[
        \|\tilde{B}_{\epsilon}(\cdot)^{-1}C_{\epsilon}(\cdot)\|_0\leq C_0'(v,\alpha,E,h,h',\epsilon).
    \] 

Since $ (\alpha,A_{\infty}) $ is $ (1,C_0(\alpha,\delta),c_0,(k_{s})_{s\geq 2}) $-good, and $(k_{s})_{s\geq 2}$ is a subsequence of the $(\delta-\frac{\e_0}{2})$-resonances of $2\rho(\alpha,A_{\infty})$. 
Denote  $ k^*=\deg \tilde{B}_{\epsilon}^{-1}C_{\epsilon}$, $\hat{k}_s=k_{s}+k^*$. By \eqref{degree}, we have 
\[
2\rho(\alpha,S_E^{\tilde{v}})-\langle \hat{k}_s,\alpha \rangle = 2\rho(\alpha,A_{\infty})-\langle k_{s},\alpha \rangle \mod \Z.
\]
Noting that $|\hat{k}_s|=(1+o(1))|k_{s}|$ as $s$ grows, there exists $\hat{s}_0>0 $ such that $(\hat{k}_{s})_{s\geq \hat{s}_0 }$ is a subsequence of the $(\delta-\e_0)$-resonances of $2\rho(\alpha,S_E^{\tilde{v}})$. 
Then one can easily check that $  (\alpha,S_E^{\tilde{v}})$ is $(C_0',2C_0,\frac{c_0}{2},(\hat{k}_{s})_{s\geq \hat{s}_0}) $-good, and the resulting estimates follow directly.
\qedhere
   
\end{proof}

\subsection{Proof of Theorem \ref{THM1.2}}
    Fix any $ h'\in(0,h) $, $\hat{\epsilon}>0$. Let $\delta$ be such that $ 2\pi h< \delta< \infty $.  Suppose that $ \alpha\in \mathrm{DC}_d $. If  $\|v\|_h\leq c_*$, 
    then Proposition \ref{THM7.1} shows that there exists a sequence  $ (\hat{k}_{s}) $, $ \tilde{v}\in C^\omega_{h'}(\T^d,\R)  $ with $\|\tilde{v}-v\|_{h'}<\hat{\epsilon} $ such that for any $\e>0$, one can apply Proposition \ref{THM5.1} at $s\geq s_0(\e)\geq \hat{s}_0 $.

    Therefore, for any $ \varepsilon>0 $, there exist  $ \epsilon_*=\epsilon_*(\e)
    >0 $, and a continuous function $ f(\epsilon): (0, \epsilon_*]\mapsto (1/2,1] $ such that 
                    \[
                        \epsilon^{f(\epsilon)+\varepsilon}\leq \mu_{\tilde{v},\alpha,\theta}(E-\epsilon,E+ \epsilon)\leq \epsilon^{f(\epsilon)-\varepsilon}.
                    \]        
    Let $n=|\hat{k}_s|$, and recall that 
\begin{equation}\label{eq:eta_n-as-before}
    -\frac{\ln \|2\rho(\alpha,S_E^{\tilde{v}})-\langle \hat{k}_{s},\alpha\rangle\|_\T
        }{n}=\eta_n \in (0,\infty], 
\end{equation}
    then  Proposition \ref{THM7.1} (3) further gives $\eta_n=\delta+o(1)>2\pi h$ as $s$ becomes large enough. This means $  f(\epsilon) $ always takes the form 
     \begin{equation} 
                f(\epsilon)=\left\lbrace\begin{aligned}
                    &\frac{1}{2}+\frac{2\pi h|\hat{k}_{s}|}{2\ln\epsilon^{-1}} &  &\text{for } \epsilon^{-1}\in[e^{2\pi h|\hat{k}_{s}|}, e^{(2\eta_{n}-2\pi h)|\hat{k}_{s}|}],\\
                    &\frac{1}{1-b_s}-\frac{b_s}{1-b_s} 
                    \frac{2\pi |\hat{k}_{s+1}|}{\ln\epsilon^{-1}} & &\text{for } \epsilon^{-1}\in[e^{(2\eta_{n}-2\pi h)|\hat{k}_{s}|},e^{2\pi h|\hat{k}_{s+1}|}],
                \end{aligned}\right.
     \end{equation}
     where $ b_s=\frac{(\eta_n-2\pi h)|\hat{k}_{s}|}{2\pi h|\hat{k}_{s+1}|-\eta_n |\hat{k}_{s}|} $. Moreover, we can replace $ f(\epsilon) $ by \eqref{eq:AKf(epsilon)}.
    
\subsection{Proof of Theorem \ref{thm:irreducible}} 

     Assume $\alpha\in \op{DC}_d $, let $ 0<2\pi h'< 2\pi h<\delta<\infty$ be given. Take $\title{h}=\frac{h+h'}{2}$.
By Proposition \ref{reducibility-stru-1}, if
     \begin{equation}\label{epsilon}
         \|A-R\|_h\leq \mathfrak{c}= \mathfrak{c}(\alpha,R,h,\tilde{h}), 
     \end{equation}
then we have $(\alpha,A)$ is almost reducible in the band $\{\theta: |\Im \theta|<\title{h}\}$, and  $\delta(\alpha,2\rho(\alpha,A))<\infty$ implies that 
$(\alpha,A)$ is not uniformly hyperbolic.
  For any $ \hat{\epsilon}>0 $, by Lemma \ref{conjugatetonearId}, there exist $ C_{\hat{\epsilon}}\in C^\omega_{h'}(\T^d,\mathrm{PSL}(2,\R)) $, such that 
  \begin{equation}\label{eq:close-to-id}
      \|A-C_{\hat{\epsilon}}(\theta+\alpha)\mathrm{Id}C_{\hat{\epsilon}}^{-1}(\theta)\|_{h'}\leq \frac{\hat{\epsilon}}{2}.
  \end{equation}
    By Lemma \ref{AKClemma}, there exists a sequence $(k_{s})$, $ A_{\infty}\in C^\omega_{h'}(\T^d,\mathrm{SL}(2,\R))  $ satisfies 
  \begin{equation}\label{eq:close-to-id-2}
    \|A_\infty-\mathrm{Id}\|_{h'}<\frac{\hat{\epsilon}}{2\|C_{\hat{\epsilon} }\|_{h'}^2}. 
  \end{equation}
    such that  
\begin{enumerate}
\item  $(\alpha,A_{\infty}) $ is $ (1,C_0(\alpha,\delta),c_0,(k_{s})_{s\geq 2}) $-good, with $ c_0<2\pi h<\infty $.
\item \begin{equation}\label{zetas} -\dfrac{\ln |\nu_s|}{|k_s|}=2\pi h+o(1),   \end{equation}
\item  \begin{equation}\label{etas} -\dfrac{\ln \|2\rho_s\|_\T}{|k_s|} =\delta+o(1).   \end{equation}
\end{enumerate}
Then let $A'(\theta)=C_{\hat{\epsilon}}(\theta+\alpha)A_{\infty}(\theta)C_{\hat{\epsilon}}^{-1}(\theta)$,  \eqref{eq:close-to-id} and \eqref{eq:close-to-id-2} implies that 
    $
        \|A-A'\|_{h'}<\hat{\epsilon}.
    $
With same discussion as in previous subsection,  we have  $$\delta(\alpha,2\rho(\alpha,A'))=\delta(\alpha,2\rho(\alpha,A_{\infty}))=\delta$$ and 
$  (\alpha,A')$ is $(C_0',2C_0,\frac{c_0}{2},(\hat{k}_{s})_{s\geq \hat{s}_0}) $-good, moreover, \eqref{zetas} and \eqref{etas} still holds.

 Below we show that $(\alpha,A')$ is not reducible. We argue by contradiction. By Remark \ref{rem:general}, if $(\alpha,A')$ is reducible, then by Proposition \ref{Prop:AsyFormula-2},
we have
\begin{equation}\label{pk2}
\det P_{(k),\pm}\thickapprox k^2 \text{ or } k^4.
\end{equation}
Again by  Remark \ref{rem:general}, this will imply  the criteria we established in subsection 5.1:  the universal asymptotic structure of $\det P_{(k),\pm}$, still hold for general quasi-periodic cocycles. More precisely, \eqref{zetas}-\eqref{etas} guarantee that one can apply Lemma \ref{lem:possible}- Lemma \ref{doas}, and for $k$ large enough,
\[
k^{\psi(k)-C\e}\leq\det P_{(k),\pm}\leq k^{\psi(k)+C\e}
\]
where if we denote $n=|\hat{k}_s|$, $N=|\hat{k}_{s+1}|$, define $\eta_n$ as in  \eqref{eq:eta_n-as-before}, then
\begin{equation}
\psi (k)=\left\lbrace\begin{aligned}
    &4-\frac{4\pi h n}{\ln k} &  &\text{for } k\in[e^{hn}, e^{\eta_{n}n}],\\
    &2+\frac{2\eta_nn-4\pi h n}{\ln k}\left[1-\frac{\ln k-\eta_n n}{2\pi hN-\eta_n n}\right] & &\text{for } k\in[e^{\eta_{n}n},e^{hN}].
\end{aligned}\right.
\end{equation}   

Meanwhile, the assumption  $2\pi h<\delta $ allows us to choose $\e\ll 1$ such that
\[
2<4-\frac{4\pi h}{\delta}-C\e  <4,
\]
and by our choice on $(k_s)$, we have 
\[
        -\frac{\ln \|2\rho(\alpha,S_E^{\tilde{v}})-\langle \hat{k}_{s},\alpha\rangle\|_\T
        }{|\hat{k}_s|}=\eta_n\to \delta, 
 \]
as $s$ grows to infinity. Thus,  there exists an infinite subsequence $(m_i)$ (actually $m_i=e^{\eta_{n}n}$ is sufficient for our choice) such that
\[
m_i^{4-\frac{4\pi h}{\delta}-C\e}\leq\det P_{(m_i),\pm}\leq m_i^{4-\frac{4\pi h}{\delta}+C\e},
\]
however, this contradicts with \eqref{pk2}.

	\appendix
\section{Structure quantitative almost reducibility}
\begin{proposition}\cite{Universal}\label{reducibility-stru-1}
    Let $\alpha\in \op{DC}_d(\kappa,\tau)$,  $A\in C_{h}^\omega(\T^d,\op{SL}(2,\R))$ with $ h> \tilde{h}>0$, $R\in \op{SL}(2,\R)$. 
    If 
     $$\|A(\cdot)-R\|_{h}\leq \epsilon \leq \mathfrak{c}:=\frac{D_0(\kappa,\tau,d)}{\|R\|^{C_0}}(h-\tilde{h})^{C_0\tau},  $$
    then there exist $B_s\in C_{\tilde{h}}^\omega(\T^d, \op{PSL}(2,\R))$,  $f_s\in C_{\tilde{h}}^\omega(\T^d, \op{sl}(2,\R))$  such that 
    $$
    B_s^{-1}(\theta+\alpha)A(\theta)B_s(\theta)=A_se^{f_s(\theta)}=M^{-1}exp \left(
    \begin{array}{ccc}
     i t^s &  \nu^s\\
    \bar{\nu}^s &  -i t^s
     \end{array}\right)Me^{f_s(\theta)}
    $$
    with estimates
    \begin{eqnarray}
     \label{es2}
    |\nu^s| &\leq& 2 \epsilon_{i_s}^{\frac{15}{16}}e^{-2\pi|\mathfrak{m}_{s}|\tilde{h}},\\
       \label{es0} \|f_s\|_{\tilde{h}} &\leq&  \epsilon_{i_s}e^{- \frac{\pi}{8} |\mathfrak{m}_{s+1}|(h-\tilde{h})}.
    \end{eqnarray}
     where $\epsilon_s\leq \epsilon^{2^s}$.   Moreover,  the conjugation $B_s(\cdot)$ takes the form 
      $$B_s(\theta)=\tilde{B}_s(\theta)R_{\frac{\langle \mathfrak{m}_{s},\theta\rangle}{2}}e^{Y_s(\theta)},$$
     with estimate 
      \begin{eqnarray} \label{es1}
     \|Y_s\|_{\tilde{h}} &<&   e^{-2\pi|\mathfrak{m}_{s}|^4\tilde{h}},\\
      \label{es3}
    \|\tilde{B}_s\|_{\tilde{h}}&<&C(\alpha)|\mathfrak{m}_{s}|^{\tau}e^{2 \epsilon_{i_{s-1}}^{\frac{1}{18\tau}} |\mathfrak{m}_{s}|},\\
       \label{es4}
    \|B_s\|_0&<&C(\alpha)|\mathfrak{m}_{s}|^{2\tau},\\
    \label{es5} 
    |\deg B_s - \mathfrak{m}_{s} | &\leq&  2\epsilon_{i_{s-1}}^{\frac{1}{18\tau}} |\mathfrak{m}_{s}|.
    \end{eqnarray} 
\end{proposition}

\begin{remark}\label{kamres}
The proof of Proposition \ref{reducibility-stru-1} is based on a novel KAM scheme, and  here we call $\mathfrak{m}_{s}$ the KAM resonances. 
\end{remark}

\begin{lemma}\cite{Universal} \label{relation1}
    Under the assumption of Proposition  \ref{reducibility-stru-1}, then for any $\epsilon_0>0$, let $(\ell_s)$ be the $\epsilon_0$-resonances of $2\rho(\alpha,A)$, there exist $B_{s}\in C_{\tilde{h}}^\omega(\T^d, \op{PSL}(2,\R))$,   such that 
    \begin{equation}\label{conjugacy}
    B_s^{-1}(\theta+\alpha)A(\theta)B_s(\theta)=A_se^{f_s(\theta)}=M^{-1}exp \left(
    \begin{array}{ccc}
     i t^s &  \nu^s\\
    \bar{\nu}^s &  -i t^s
     \end{array}\right)Me^{f_s(\theta)}
    \end{equation}
    with $\ell_i=\deg B_s$ provided $\ell_i$ is large enough (the choice of $s$ depends on $\ell_i$). 
    \end{lemma}

\section*{Acknowledgements} 
 This work was partially supported by National Key R\&D Program of China (2020 YFA0713\\300) and Nankai Zhide Foundation. J. You was also partially supported by NSFC grant (11871286). X. Li was supported by NSFC grant (123B2005). Q. Zhou was supported by NSFC grant (12071232).

\bibliographystyle{plain}
\bibliography{ExactLocalRef}
\end{document}